\documentclass[sigconf,balance=false]{acmart}

\usepackage{popets}
\setcopyright{popets}
\copyrightyear{YYYY}

\acmYear{YYYY}
\acmVolume{YYYY}
\acmNumber{X}
\acmDOI{XXXXXXX.XXXXXXX}
\acmISBN{}
\acmConference{Proceedings on Privacy Enhancing Technologies}
\usepackage{amsmath,amsfonts,amsthm}

\usepackage{graphicx}

\usepackage{array}
\usepackage{makecell}

\usepackage{changepage} % This will allow us to customize the margins for subcases 

\theoremstyle{plain}
\newtheorem{lem}{Lemma}
\newtheorem{thm}{Theorem}
\newtheorem{cor}{Corollary}
\newtheorem{prop}{Proposition}
\newtheorem{obs}{Observation}

\usepackage[linesnumbered,ruled]{algorithm2e}

\theoremstyle{definition} %% this is to make the examples in non-italic text
\newtheorem{defn}{Definition}
\newtheorem{examplex}{Example}
\newenvironment{ex}
  {\pushQED{\qed}\examplex}
  {\popQED\endexamplex}

\newcommand{\R}{\mathbb{R}}

\renewcommand{\P}{\mathbb{P}}
\newcommand{\E}{\mathbb{E}}

\newcommand{\argmax}{\text{argmax}}

\renewcommand{\O}{\mathbb{O}}

\usepackage{dsfont}

\usepackage{enumitem}

\usepackage{mathtools}

\usepackage{tikz-cd}
\usepackage{adjustbox}

\begin{document}

\title[Preserving Target Distributions With Differentially Private Count Mechanisms]{Preserving Target Distributions With Differentially Private \\Count Mechanisms}

%%%%%%%%%%%% Authors' Info %%%%%%%%%%%%%%%%%

\author{Nitin Kohli}
\affiliation{%
  \institution{UC Berkeley Center for Effective Global Action}
  \city{}
  \state{}
  \country{}}
\email{nitin.kohli@berkeley.edu}

\author{Paul Laskowski}
\affiliation{%
  \institution{UC Berkeley School of Information}
  \city{}
  \state{}
  \country{}}
\email{paul@ischool.berkeley.edu}


\renewcommand{\shortauthors}{Kohli and Laskowski}

%%%===================%%%
%%% =*= Abstract  =*= %%%
%%%===================%%%

\begin{abstract}
    Differentially private mechanisms are increasingly used to publish tables of counts, where each entry represents the number of individuals belonging to a particular category. A \textit{distribution of counts} summarizes the information in the count column, unlinking counts from categories. This object is useful for answering a class of research questions, but it is subject to statistical biases when counts are privatized with standard mechanisms. This motivates a novel design criterion we term \textit{accuracy of distribution}.

    This study formalizes a two-stage framework for privatizing tables of counts that balances  accuracy of distribution with two standard criteria of accuracy of counts and runtime. In the first stage, a \textit{distribution privatizer} generates an estimate for the true distribution of counts. We introduce a new mechanism, called the cyclic Laplace, specifically tailored to distributions of counts, that outperforms existing general-purpose differentially private histogram mechanisms. In the second stage, a \textit{constructor algorithm} generates a count mechanism, represented as a transition matrix, whose fixed-point is the privatized distribution of counts. We develop a mathematical theory that describes such transition matrices in terms of simple building blocks we call $\epsilon$-scales. This theory informs the design of a new constructor algorithm that generates transition matrices with favorable properties more efficiently than standard optimization algorithms. We explore the practicality of our framework with a set of experiments, highlighting situations in which a fixed-point method provides a favorable tradeoff among performance criteria.
\end{abstract}

\keywords{Differential Privacy, Distribution Preservation, Optimal Count Mechanisms, Convex Polytopes, Fixed-Point Analysis}

\maketitle

%%%==================%%%
%%% =*= Sections =*= %%%
%%%==================%%%

\section{Introduction}
\label{intro}
When studying populations, data is often organized as a table of counts, where each row records the number of people belonging to a particular category. Common categories include geographic regions, schools, and companies, among others. We assume that there are $N$ categories, numbered $0$ through $N-1$. We assume that each category $i$ has a corresponding count $d_i$.  We also require that every individual belongs to no more than one category. An example of such a dataset is provided in Figure \ref{fig:diagram}.A, with categories corresponding to U.S. states and counts recording the number of individuals who contracted bird flu.

For organizations that hold such data -- including census bureaus \cite{abowd2018us}, humanitarian groups \cite{kohli2024privacy}, and public health departments \cite{savi2023standardised} -- there is a tension between releasing the data to facilitate wider research and analysis, and protecting the privacy of the individuals contained inside. In recent years, some organizations have started to publish \textit{tables of privatized counts} in place of raw data. Here, we are referring specifically to tables that preserve the row structure of the original data; these are tables that contain an approximate count for every category in the original dataset. In a typical scenario, independent noise is added to each count, with the noise magnitude tailored to meet the common standard of differential privacy \cite{dwork2006calibrating}. 

%% Retrieved from https://www.cdc.gov/bird-flu/situation-summary/index.html
\begin{figure}[!t]
\centering
\includegraphics[width=\linewidth]{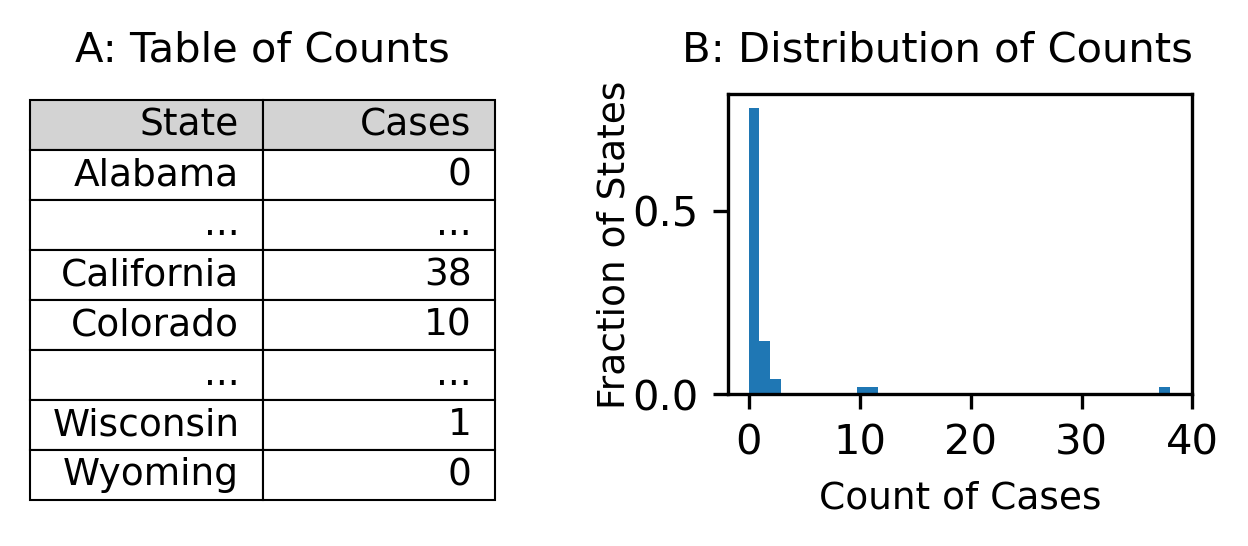}
    \vspace*{-0.15cm}
\caption{A. Table of counts of bird flu cases by U.S. state as of February 27, 2025. B. The corresponding distribution of counts. In this figure, $N = 50$ and $n = 41$.} 
\label{fig:diagram} 
\end{figure}

We can distinguish two broad classes of questions that can be answered using a table of counts. The first are questions that must be formulated using the name of at least one category. For example, a researcher might want to know how many bird flu cases have occurred in Colorado. We refer to these as \textit{category-based questions}. In contrast, some questions can be written without referencing any category names. Examples include the following policy-relevant questions.  

\begin{itemize}
    \item What fraction of US flights have zero air marshals onboard? This number raises security concerns, and in particular was cited \cite{Hinson2024AirMarshals} in motivating a 2024 congressional bill \cite{NoFAMSAtBorderAct2024}.
    \item What fraction of company boards have more than 2 women? This threshold has been argued as important for changing company culture \cite{kramer2006critical} and tracked in a recent report \cite{pwc2017acds}.
    \item What is the average number of students overseen by each guidance counselor in the U.S.? A 2025 study estimated the number to be 405, far exceeding the 250 recommended by the American School Counselor Association \cite{NACAC2025SchoolCounseling}.
\end{itemize} 

We refer to questions like these as \textit{count distribution questions}, as it is possible to compute them after first preprocessing the count column. We assume each count $d_i$ is in the set $\{0,1,....n-1\}$, with higher counts truncated if necessary. We assume that the upper bound $n-1$ is given publicly; if not, some privacy budget may be reserved to set $n$ (for example, $n$ can be set to a privatized $99^{th}$ percentile).

We first define the \textit{histogram of counts} as a vector $\eta = (\eta_0,..., \eta_{n-1})$, where $\eta_i = |\{j \in \{0,...,N-1\} : d_j = i\}|$ is the number of categories with count $i$. For example, in Figure \ref{fig:diagram} $\eta_0$ equals the number of states with no bird flu cases. Note that $N = \sum_{i=0}^{n-1} \eta_i$. We then define the \textit{distribution of counts} as a vector $\zeta = (\zeta_0,...,\zeta_{n-1}) $ where $\zeta_i = \eta_i / N$ denotes the proportion of categories with count $i$. Figure \ref{fig:diagram}.B depicts the distribution of counts for the count table in Figure \ref{fig:diagram}.A.

A problem arises when existing differential privacy techniques are used to answer count distribution questions. Within the literature, utility is most often captured via the typical error for a single count. Thus, canonical mechanisms --- including the Laplace \cite{dwork2006calibrating}, geometric \cite{ghosh2009universally}, Gaussian \cite{dwork2006our}, and discrete Gaussian \cite{canonne2020discrete} --- are not designed to preserve the distribution of counts. In general, each bin of the distribution of counts will have non-zero bias. For example, the geometric mechanism with outputs bounded to be non-negative tends to inflate the number of zeros in the count table. 

Given the existence of such biases, one possible solution would be to independently create two separate data products: a table of privatized counts and a privatized distribution of counts (with the privacy budget split between the two). However, without additional constraints, there would be no guarantee of consistency between the two data products.  As Long et al. explain, ``Stakeholders [...] have expressed a strong desire for consistent data [...]. Such consistency has historically formed the bedrock of statistical use of the data, making it seamless to combine data across tables and easy to integrate it into traditional models and analyses'' \cite{jason2020}. 

For the reasons above, we believe there is interest in a unified data product: a table of privatized counts that provides accurate answers to both category-based and count-distribution questions.

%%%%%%%%%%%%%%%%%%%%%%%%%%%
\subsection{Distribution-Preserving Count Mechanisms}
\label{subsec:dist_preserving_count}

Conceptually, a \textit{count mechanism} takes a ``true'' input count and outputs a (randomized) output count. Mathematically, it is useful to represent such a mechanism as an $n \times n$ transition matrix; this is a matrix $T \in \R^{n \times n}_{\ge 0}$ with $T\mathds{1}_n=\mathds{1}_n$, where $\mathds{1}_n$ denotes the $n$-dimensional column vector of all ones. In alignment with the task of counting, all vectors and matrices will be indexed starting from 0. We will use capital letters for matrices and lowercase letters for their elements. For count mechanism $T$, $t_{i,j}$ represents the conditional probability that the output count is $j$ given that the input count is $i$. 

When the input count to a count mechanism $T$ is drawn from a distribution $z$, written as an $n$-dimensional row vector, the distribution of the output count may be expressed as $zT$. We say that a count mechanism $T$ has \textit{fixed point $z$} if $zT= z$.

For the moment, assume that the true distribution of counts $\zeta$ is publicly known (we will relax this assumption in Section \ref{subsec:proposed_workflow}). We envision the following: The data holder selects count mechanism $T$ with fixed point $\zeta$. They then take each row $i$ of the dataset, passes the count $d_i$ as an input to $T$, and records the output in row $i$ of the privatized table. The resulting output may be summarized with its distribution of privatized counts, which we label $\hat z$. 

Because of the fixed-point constraint, for every output count $j \in \{0,...,n-1\}$, $\E[\hat z_j] = \zeta_j$, so the output distribution of counts is correct in expectation.\footnote{To see why, let $\mathcal{I}_{i,j}$ be the indicator random variable corresponding to the event that row $i$ of the count table (with count $d_i$) is mapped to $j$ under $T$. Then  $\mathcal{I}_{i,j} \sim \text{Bernoulli}(t_{d_i,j})$ and 
$
    \hat z_j = \frac{1}{N} \sum_{i=0}^{N-1} \mathcal{I}_{i,j} =  \frac{1}{N} \sum_{l=0}^{n-1}\sum_{i \text{ : } d_i = l} \mathcal{I}_{i,j}
$
The inner sum is the sum of $\eta_l$ Bernoulli random variables, each with expectation $t_{l,j}$, so 
$
    \E[\hat z_j] = \frac{1}{N} \sum_{l=0}^{n-1}\eta_l t_{l,j}  = \sum_{l=0}^{n-1} \zeta_l t_{l,j}  = (\zeta T)_j = \zeta_j
$
where the last equality follows by the fixed-point constraint. Also 
$
\text{Var}(\hat z_j) = \frac{1}{N} \left(\zeta_j - \sum_{l=0}^{n-1} \zeta_l t_{l,j}^2 \right)
$
so we expect the sampling error of the output distribution to be small for larger datasets.} While this is a promising observation, one obstacle remains to using a fixed-point constraint in this way: distribution $\zeta$ is often not publicly known. We next develop a workflow for situations in which $\zeta$ cannot be directly used.

%%%%%%%%%%%%%%%%%%%%%%%%%%%

\subsection{A Two-Stage Counting Framework}
\label{subsec:proposed_workflow}
\begin{figure}[!ht]
\centering
\includegraphics[width=\linewidth]{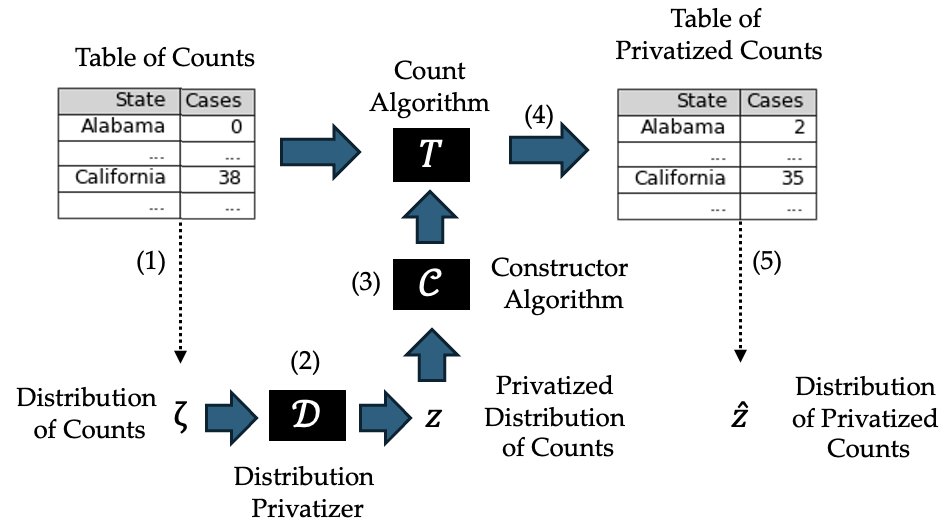}
    \vspace*{0.75mm}
\caption{A framework for constructing and utilizing count mechanisms when the true distribution must be protected.}
\label{fig:workflow} 
\end{figure}

When choosing a count mechanism, a common obstacle is that the true distribution $\zeta$ is not public knowledge and must be kept secret by the data holder. To overcome this obstacle, we follow the strategy of adding a preprocessing step to generate a privatized approximation to $\zeta$. The total privacy budget, which we label $\epsilon_t$, must therefore be divided as $\epsilon_t = \epsilon_1 + \epsilon_2$, with $\epsilon_1$ used to privatize $\zeta$, and $\epsilon_2$ reserved for the count mechanism.

When the distribution is privatized before a count mechanism is chosen, we will refer to the overall system as a \textit{two-stage counting framework}. Figure \ref{fig:workflow} lays out the operation of such a framework, which involves five steps.

\begin{enumerate}
    \item The table of counts is summarized by its distribution of counts $\zeta$.
    \item A distribution privatizer  $\mathcal{D}$ is used to privatize $\zeta$, using privacy budget $\epsilon_1$, forming privatized distribution $z$ (the \textit{target distribution}).
    \item Constructor algorithm $\mathcal{C}$ is used to create an $\epsilon_2$-differentially private count mechanism $T$. In normal operation, we require that $T$ has fixed point $z$.
    \item Each row of the table of counts is passed into count mechanism $T$, one at a time, with the output recorded in a table of privatized counts using privacy budget $\epsilon_2$.
    \item Users with count distribution questions may summarize the table of privatized counts with its distribution of privatized counts $\hat z$.
\end{enumerate}

To fully implement this framework, one must specify a distribution privatizer $\mathcal{D}$, a constructor algorithm $\mathcal{C}$, and several parameter values. To guide the choices that go into a deployment, we first identify a set of three performance criteria. These are intended to represent the top concerns of data users and data holders.

\begin{enumerate}
    \item \textbf{Accuracy of Distribution.} We conceptualize \textit{distribution error} as a difference between $\zeta$ and $\hat z$. Although our framework involves a fixed-point constraint, there are two reasons that these distributions will not be identical. First, the fixed-point of $T$ is $z$, which is an approximation to $\zeta$. Second, there will be sampling error as each row of the table is selected randomly using probabilities in $T$.  Distribution error can be operationalized using Wasserstein, total variation, or Kolmogorov-Smirnov distance. A researcher with a count distribution question may be primarily concerned with accuracy of distribution. 
    \item \textbf{Accuracy of Counts.}  We conceptualize \textit{count error} as a typical distance between a single input count and the corresponding output count. This is commonly operationalized as mean squared error or expected absolute deviation. A researcher with a category-based question may primarily be concerned with accuracy of counts.
    \item \textbf{Runtime.} In many scenarios, the space of count mechanisms has high dimensionality and care must be taken that the constructor algorithm $\mathcal{C}$ can run efficiently. As we will see in Section \ref{experiments}, standard algorithms from linear programming are often impractical in our setting.
\end{enumerate}

There are obstacles to designing a system that meets all three performance criteria to a satisfactory level. First, one might hope to adapt existing algorithms to play the role of the distribution privatizer $\mathcal{D}$ and the constructor algorithm $\mathcal{C}$. While this works for some datasets, as we explore in Section \ref{experiments}, existing algorithms often lead to reduced accuracy or impractically long runtimes. Moreover, the development of new constructor algorithms is complicated by a lack of mathematical theory describing private count mechanisms with fixed-point constraints.

%%%%%%%%%%%%%%%%%%%%%%%%%%%

\subsection{Summary of Contributions}

Our study pioneers the development of count mechanisms that incorporate a fixed-point constraint to provide high accuracy of distribution. We make four main contributions: the first three are mathematical, corresponding to different portions of the two-stage counting framework; the last contribution is experimental, and explores the practicality of our approach. Table \ref{tab:notation} provides a table of notation used throughout the paper.

\begin{table}[ht]
%\vspace{.3cm}
\begin{tabular}{|>{\raggedright\arraybackslash}w{c}{1.5cm}|p{5.75cm}|}
\hline
\textbf{Notation} & \textbf{Description} \\
\hline
$n$ & Maximum allowed count is $n-1$ \\
$N$ & Number of categories \\
$\eta$ & Histogram of counts \\
$\zeta$ & Distribution of counts \\
$\mathcal{D}$ & Distribution privatizer with budget $\epsilon_1$\\
$\mathcal{C}$ & Constructor algorithm with budget $\epsilon_2$\\
$z$ & Fixed point, output of $\mathcal{D}$ \\
$T$ & Count algorithm, output of $\mathcal{C}$ \\
$\mathds{1}_n$ & $n$-dimensional vector of 1's\\
$F$ & Fixed-point polytope for $z$\\
$U$ & Unfixed point polytope\\
$R_F$ & Representation Fixed-point polytope of $z$\\
$R_U$ & Representation Unfixed point polytope\\
$M_F$ & Movement matrices for $R_F$ \\
$M_U$ & Movement matrices for $R_U$ \\
$s$ & $\epsilon$-scale \\
$p$ & pattern of an $\epsilon$-scale  \\
$k$ & Number of $\epsilon$-scales ($2^{n-1})$ \\
$\Psi$ & Scale matrix \\
$ex(\cdot)$ & Extreme points of a polytope \\
$\langle \cdot, \cdot \rangle$ & Frobenius inner product \\
$\kappa$ & Column selector function \\
\hline
\end{tabular}
\vspace{1mm}
\caption{Frequently used notation.}
\label{tab:notation}
\end{table}

\textbf{Distribution Privatizer.} We provide a novel mechanism to privatize a distribution of counts with a differential privacy guarantee (Theorem \ref{thm:cyclic_laplace}). Our mechanism, called the \textit{cyclic Laplace mechanism for distributions of counts} (Definition \ref{defn:cyclic_laplace}), employs Laplace noise, but transfers probability between adjacent positions in the distribution of counts, which is specifically tailored to mask the presence or absence of an individual in this context. We provide evidence that this technique yields more accurate distributions than existing general-purpose differentially private histogram mechanisms.

\textbf{Theory of Count Mechanisms with Fixed Points.} We provide a compact description of count mechanisms that satisfy differential privacy and a fixed-point constraint using an atomic element that we call an \textit{$\epsilon$-scale} (Definition \ref{defn:scale}). Roughly speaking, an $\epsilon$-scale is a probability distribution over $\{0,1,2,...,n-1\}$ in which a differential privacy constraint binds for every pair of neighboring positions. The name comes from the notion that, if one were to take a ``walk'' along a scale, every step from one position to the next would either be a step up or step down by the same multiplicative constant. Our results will show that $\epsilon$-scales are fundamental building blocks of differentially private count mechanisms.

We begin by formally defining $F_\epsilon[z]$ as the set of count mechanisms that satisfy $\epsilon$-differential privacy with fixed-point $z$ (Definition \ref{defn:F_defined}); we often just write $F$ for brevity when $\epsilon$ and $z$ are clear from context. Viewed geometrically, $F$ is a convex polytope, and hence it can be characterized using its extreme points. 

We show that any differentially private fixed-point count mechanism in $F$ can be described using a conic combination of $\epsilon$-scales (Corollary \ref{cor:F_col_scales}). Building on this observation, we construct a \textit{representation polytope} $R_F$ that encodes how much of each scale contributes to each column of a mechanism $T$ (Definition \ref{defn:RF_definition}). This polytope is related to $F$ by a linear map $\Psi$ whose matrix representation is comprised of the $\epsilon$-scales. In Theorem \ref{thm:RF_to_F_surjection}, we prove that any mechanism in $F$ has a representation (not necessarily unique) in $R_F$. Furthermore, Proposition \ref{prop:extreme_surjection_R_to_F} links each extreme point of $F$ to at least one extreme point of $R_F$ through $\Psi$.

These results suggest that we can learn about the extreme points of $F$ by studying the extreme points of $R_F$. In Theorem \ref{thm:u_general_extreme}, we characterize these extreme points in terms of the scales that make them up. Specifically, the extreme points of $R_F$ are exactly those for which certain weighted differences of scales are linearly independent -- a condition we call  \textit{$\Psi$-affinely simplified} (Definition \ref{defn:psi_aff_simple}). 

Taken together, our results show that every extreme matrix of $F$ can be represented by a relatively small number of $\epsilon$-scales (Corollary \ref{cor:nonzero_RF}), fulfilling certain conditions (Theorem \ref{thm:u_general_extreme}). The algebraic structure that we discover provides a framework for reasoning about differentially private fixed-point count mechanisms, which can be leveraged for algorithmic design.

\textbf{Constructor Algorithm.} Next, we consider the design of a constructor algorithm. Our primary focus is on the development of constructor algorithms whose output has fixed-point $z$, which we refer to as \textit{fixed-point constructor algorithms}. Since all such constructors share the same fixed-point constraint, we expect them to behave similarly in terms of accuracy of distribution (we will later confirm this in experiments). We therefore follow an approach in which we focus on accuracy of counts and use $\mathcal{C}$ to find a count mechanism with low count error. It is sometimes possible to minimize count error exactly using standard linear programming algorithms; however, the runtime of existing solvers can be prohibitively long for large $n$. For such settings, we provide Algorithm \ref{alg:heuristic}, which leverages our scale representation and builds a conic combination of scales in a greedy fashion to output an extreme point of $F_{\epsilon_2}[z]$ (Theorem \ref{thm:heur_output_extreme_F}). Our algorithm runs as fast as $\mathcal{O}(n^2)$ (Theorem \ref{thm:heur_halts}), demonstrating its practicality even when $n$ is large.

Our mathematical framework can also be utilized to design constructor algorithms that do not involve a fixed-point constraint. When used in a two-stage framework, this forms a relevant baseline against which fixed-point methods can be compared. We let $U$ denote the set of $\epsilon$-differentially private count mechanisms that may send $z$ to a different distribution (Definition \ref{defn:U}). We leverage our framework and a key result from Ghosh et al. \cite{ghosh2009universally} to create a \textit{two-staged unfixed optimum constructor} (Algorithm \ref{alg:ghosh_scale}), which returns the lowest count-error mechanism in $U$ in $\mathcal{O}(n^2)$ time (Theorem \ref{thm:unfixed_runtime}).

%One might wonder if the two-staged unfixed optimum constructor can be used in isolation without the rest of the two-stage counting framework. Even though $U$ is not affected by $\zeta$, the count error measures we use -- expected absolute deviation and mean squared error -- do depend on the distribution. If $\zeta$ were not privatized, changing an individual's data could change the point in $U$ that minimizes count error -- for example, altering columns of the transition matrix from positive to zero. This would violate differential privacy; hence, we continue to use the two-stage counting framework with the two-staged unfixed optimum, but without the fixed-point requirement in Step 3.

\textbf{Empirical Results and Practical Guidance.} In Section \ref{experiments}, we report on a set of experiments in which our fixed-points methods are used to generate tables of privatized counts. We are interested in how different fixed-point constructors perform relative to each other, and also how they perform relative to four \textit{unfixed baselines} -- mechanisms representing prior work that do not enforce a fixed-point constraint. We divide these questions into three components, aligned with our three performance criteria: accuracy of distribution, accuracy of counts, and runtime. 

First, we find that the inclusion of a fixed-point constraint usually yields a significant improvement to accuracy of distribution, relative to all unfixed baselines. This is to be expected, as the unfixed baselines do not encode a preference for accuracy of distribution. Next, we find that that switching from an unfixed baseline to a fixed-point method generally degrades accuracy of counts. This is also to be expected, as fixed-point methods must output a count mechanism in polytope $F$, which is a strict subset of the unfixed polytope $U$. Nevertheless, we find the difference to be moderate in magnitude -- sometimes a matter of a few percentage points. Finally, we find that the fastest runtime is always attained by one of the unfixed baselines. Fixed-point constructors that minimize count error exactly (i.e., using the interior point method or simplex algorithm) exhibit large growth rates and are impractical for all but the smallest values of $n$; in these situations, our heuristic constructor provides a promising alternative, with execution time that is experimentally comparable to the fastest unfixed baseline.

Taken together, our experimental results shed light on the practical tradeoffs from adopting a fixed-point constraint. Throughout Section \ref{experiments}, we highlight scenarios where a fixed-point constraint provides substantial improvements in terms of accuracy of distribution, with modest performance losses in both accuracy of counts and runtime.

\vspace{.2cm}

\noindent \textbf{Organization of Paper:} In Section \ref{related_works}, we discuss related works. In Section \ref{privatizing_z}, we propose a distribution privatizer that is $\epsilon$-differentially private and argue that it has favorable error characteristics. In Section \ref{characterization} we develop a mathematical framework that describes mechanisms in $F$ using $\epsilon$-scales, and provide conditions that characterize all extreme points. In Section \ref{algorithmic}, we discuss constructor algorithms, including novel algorithms that leverage $\epsilon$-scales to run efficiently. In Section \ref{experiments}, we explore the practicality of our approach with a set of experiments. We conclude with a discussion of future avenues of inquiry in Section \ref{disc}. For readability, all proofs of the mathematical results described in the main text are deferred to the Appendix.

\section{Related Works}
\label{related_works}
Differential privacy was introduced in 2006 by Dwork, McSherry, Nissim, and Smith \cite{dwork2006calibrating}. One of the earliest and most studied problems in the field is that of privately computing a count of individuals \cite{dwork2006calibrating}. The modern toolkit for privatizing counts includes the Laplace \cite{dwork2006calibrating} and Gaussian mechanisms \cite{dwork2006our}, and their discrete analogs (the geometric \cite{ghosh2009universally, fernandes2019utility} and discrete Gaussian mechanisms \cite{canonne2020discrete}), sometimes followed by a postprocessing step to ensure that desirable data properties are met \cite{boninsegna2025differential, us2021disclosure, cumings2023disclosure}.

A number of authors present methods for privatizing distributions, which can serve as a distribution privatizer in our two-stage framework. Typically in this literature, the outputs corresponding to bins are correlated, in order to improve the utility for certain sets of queries. For example, \citet{xiao2010differential} apply wavelet transformations to privatize ordinal data while maintaining high accuracy for range count queries. \citet{hay2010boosting} introduce hierarchical methods to improve the accuracy of range queries for differentially private histograms.   \citet{qardaji2013understanding} improve upon the hierarchical method of Hay et al. by specifying an error measure and recommending optimal branching factors. Like the above studies, our cyclic Laplace mechanism also utilizes noise that is shared between histogram bins, but our sharing strategy is tailored specifically to the setting of distributions of counts. Sharing a focus on privatized distributions, \citet{li2020estimating} propose a mechanism that adds noise to individual datapoints before reconstructing a distribution. Their noise is drawn from a square wave to meet the standard of local differential privacy, a stricter standard than the notion of (central) differential privacy considered in this paper. 

Several studies address privatizing counts, and are especially relevant to the constructor algorithm of our two-stage framework. In a series of papers, Geng et al. show that the staircase mechanism is the differentially private algorithm that minimizes a certain class of cost functions with real-valued outputs \cite{geng2014optimal, geng2013optimal, geng2015staircase}. Li et al. introduce the matrix mechanism to privately answer a collection of count-based queries using (potentially correlated) Laplace and Gaussian noise to minimize the total error \cite{li2015matrix}. Sadeghi et al. utilize techniques from linear programming to develop a differentially private count mechanism that adds integer-valued noise over a fixed and finite range, and numerically demonstrate its improved privacy properties over the discrete Gaussian mechanism \cite{sadeghi2020differentially}. 

In the study closest to ours, Ghosh et al. study the problem of minimizing count error for count mechanisms without a fixed-point constraint \cite{ghosh2009universally}. Under mild assumptions on an objective function, they find that all optimal mechanisms are equivalent to a geometric mechanism under post-processing -- a result that we leverage in developing our two-staged unfixed optimum constructor. Unlike Ghosh et al., we provide a linear algebra framework that provides new intuition for the structure of count mechanisms, and we focus on a scenario with a fixed-point constraint.

Kairouz et al. \cite{kairouz2016extremal} and Holohan et al. \cite{holohan2017extreme} both mathematically analyze the behavior of optimal algorithms satisfying local differential privacy \cite{erlingsson2014rappor, yang2023local}. As in our setting, the authors represent mechanisms as transition matrices and apply linear constraints corresponding to a differential privacy guarantee. In contrast to the task of counting, the authors allow neighboring datasets to be associated with any pair of inputs, not just consecutive rows. This leads to the local differential privacy polytope. We study different polytopes in this paper, but our focus on characterizing extreme points is shared by the above authors.

The mathematical techniques we use to analyze our polytopes of interest, $U$ and $F$, are rooted in the algebraic study of matrix polytopes more broadly. The classic Birkhoff–von Neumann Theorem states that permutation matrices are the extreme points of the set of doubly stochastic matrices \cite{dufosse2018further}. Jurkat and Ryser \cite{jurkat1967term} generalize this result to the transportation polytope with row and column sums equal to $\varphi \in \R^{1\times n}$. Hartfiel \cite{hartfiel1974study} characterizes the extreme points of the set of stochastic matrices with a fixed point $z$. Since we incorporate privacy constraints into the definition of our polytopes, our results can be viewed as extensions of the above results, though our algebraic tools are somewhat different.

\section{Privatizing Distributions of Counts}
\label{privatizing_z}
In this section, we consider the design of a distribution privatizer $\mathcal{D}$, which forms part of our two-stage counting framework. The role of this component is to take the secret distribution of counts $\zeta$ and return an approximation meeting $\epsilon$-differential privacy. One existing candidate algorithm is the $n$-dimensional Laplace mechanism, which adds independent Laplace noise\footnote{\label{foot:laplace_dist}The $\text{Laplace}(b)$ distribution is the probability distribution over $\R$ whose density function is $\varphi(x;b) = \frac{1}{2b}\exp(-|x|/b)$, expectation is 0, and variance is $2b^2$.} to each dimension proportional to $L_1$ global sensitivity \cite{dwork2006calibrating}. Compared to a generic probability distribution however, a distribution of counts brings additional structure that the classic Laplace mechanism does not leverage. We therefore provide a design for an alternative algorithm.

A \textit{randomized algorithm} $\mathcal{R}$ accepts an input and outputs a random variable that takes values in some set $\O$. Differential privacy depends on the concept of \textit{neighboring inputs} \cite{dwork2006calibrating}. We adopt the presence-absence variant of neighboring, which means that two inputs are neighbors if one input corresponds to adding or removing exactly one individual from the other. In the context of count tables, this means that two count tables $d$ and $d'$ are neighboring if and only if there exists a row $i \in \{0,...,N-1\}$ such that $|d_i - d'_i| = 1$ and for all $j \ne i$, $d_j = d'_j$. 

\begin{defn}
\label{defn:differential_privacy} 
A distribution privatizer satisfies $\epsilon$-differential privacy if for every pair of neighboring count tables $d$ and $d'$ and for every measurable set $\Theta \subseteq \O$, $\P(\mathcal{R}(d) \in \Theta) \le e^\epsilon \P(\mathcal{R}(d') \in \Theta)$.
\end{defn}

Throughout this study, we assume that $\epsilon>0$. When a distribution privatizer is used in the two-stage counting framework, we use the value of $\epsilon = \epsilon_1$.

\begin{defn}
\label{defn:cyclic_laplace} 
We define the \textit{cyclic Laplace mechanism for distributions of counts} as a randomized algorithm that accepts a parameter $\epsilon$ and table of counts $d$ as inputs. It then computes its distribution of counts $\zeta$ and independently samples $L_0, ..., L_{n-1} \sim \text{Laplace}(1/(N\epsilon))$, defines $L_n = L_0$, and outputs $V(\zeta) = (V_0,...,V_{n-1})$ where $V_i = \zeta_i + L_i - L_{i+1}$ for all $i \in \{0,...,n-1\}$.
\end{defn}

\begin{obs}
\label{obs:cyc_lap_sums}
    By construction, $\sum_{i=0}^{n-1} V_i = \sum_{i=0}^{n-1} \zeta_i = 1$. However, there is no guarantee that $V(\zeta)$ is a probability distribution because some elements may be negative, so post-processing may be required.
\end{obs}

To provide some intuition for this mechanism, consider our earlier example in Figure \ref{fig:diagram}. A switch to a neighboring dataset could represent an extra individual in Alabama contracting bird flu, increasing the count for Alabama from 0 to 1, while all other rows remain unchanged. As a result, there would then be one extra 1 in the column of counts, and one fewer 0. The histogram of counts therefore changes only in positions 0 and 1, and only by 1 in these positions. In general, switching to a neighboring dataset always has the effect of shifting one unit from one histogram position to an adjacent position. In the cyclic Laplace mechanism for distributions of counts, such a change can be ``masked'' if the Laplace random variable corresponding to the pair of positions changes its value by one. This intuition is formalized in the proof of Theorem \ref{thm:cyclic_laplace}.

\begin{thm}
\label{thm:cyclic_laplace}
    The cyclic Laplace mechanism for distributions of counts satisfies $\epsilon$-differential privacy.
\end{thm}

A benefit of the cyclic Laplace mechanism for distributions of counts is that it does not alter the cumulative sum of the distribution of counts by very much. For index $i \in \{0,...,N-1\}$, the variance of the cumulative sum of the output of the cyclic Laplace mechanisms for distributions of counts is given by,
$$
\text{Var}\left[\sum_{j\le i} V_j \right] = 
\text{Var}\left[\sum_{j\le i} \zeta_j + L_j - L_{j+1} \right] = 
\text{Var}\left[L_0 - L_{j+1} \right] 
= \frac{4}{N^2\epsilon^2} 
$$
This result can be compared to the classic Laplace mechanism, which would output $\tilde V(\zeta)=(\tilde V_0,...,\tilde V_{n-1})$ where $\tilde V_j = \zeta_j + \tilde L_j$ for independent Laplace random variables $\tilde L_j \sim \text{Laplace}(2/(N\epsilon))$ (the $L_1$ sensitivity is $2/N$ since adding an individual increases $\eta_i$ for one $i$ and reduces $\eta_{i'}$ for an adjacent $i'$). The resulting variance is
$$
\text{Var}\left[\sum_{j\le i} \tilde V_j \right]  = \text{Var}\left[\sum_{j\le i} \zeta_j + \tilde L_j \right] = \text{Var}\left[\sum_{j\le i} \tilde L_j \right]= 
(i+1)\frac{8}{N^2\epsilon^2}
$$
which is larger by a factor of $2(i+1)$, meaning that the cumulative sum gets progressively less precise from position $0$ to position $n-1$.

Note that the sums in the variances above are not technically cumulative distributions, because $V$ and $\tilde V$ may not be valid probability distributions. To use either of these mechanisms in the two-stage counting framework, it is necessary to post-process the output so it is a valid probability distribution. Unfortunately, the simple expressions for variance above may not hold after post-processing. 

We perform experiments to measure the distributional accuracy of the cyclic Laplace in practice. As points of comparison, we consider three mechanisms: the classic Laplace, the Privelet mechanism with the Haar wavelet transform of \citet{xiao2010differential}, and the Hierarchical mechanism with weighted averaging, mean consistency, and a branching factor of 16 as recommended by \citet{qardaji2013understanding} (hereafter the QYL mechanism). The latter two mechanisms are designed to improve accuracy when answering range queries. 

We use a dataset describing schools which will be introduced later in Section \ref{experiments}, top coding the counts at $n-1$ for different values of $n$. We operationalize distribution error using Wasserstein-1 and Kolmogorov-Smirnov (KS) distance.  Figure \ref{fig:distribution_privatizers} shows the results of this study for Wasserstein distance; the results for KS distance are similar. The cyclic Laplace mechanism outperforms the other mechanisms for all values of $n$.

More generally, this cyclic approach can be combined with other differentially private mechanisms as building blocks, such as the Gaussian mechanism (see Appendix \ref{app_proofs_privatizing_z} for details).

\vspace{-0.175cm}

\begin{figure}[!t]
\centering
\includegraphics[width=\linewidth]{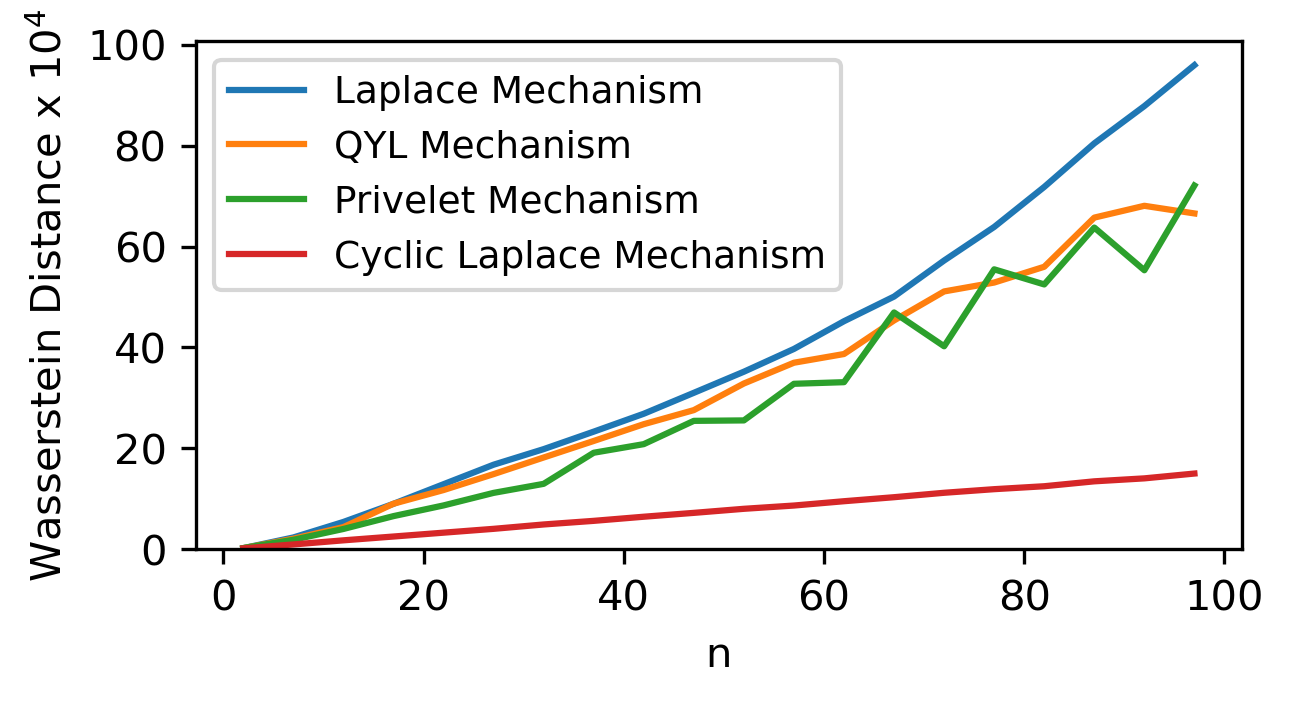}
    \vspace*{0.15mm}
\caption{Comparison of distribution error for the cyclic Laplace mechanism against other mechanisms with $\epsilon_1 = 1$.} 
\label{fig:distribution_privatizers} 
\end{figure}

\section{Describing Fixed-Point Count Mechanisms}
\label{characterization}

In this section, we provide a compact description of count mechanisms that satisfy differential privacy and the fixed-point constraint. A key aim of this effort is to provide a framework to assist in designing constructor algorithms for the two-stage counting framework.

%%%%%%%%%%%%%%%%%%%%%%%%%%%%%%%%%%%%%%%%%%%%%%%%%%%%%%%%%%%%%%%

\subsection{Mathematical Formulation}

Recall that from Section \ref{subsec:dist_preserving_count} that a count mechanism is represented by an $n \times n$ transition matrix $T$. To incorporate differential privacy into this setting, we compare the conditional probability of an output $j \in \{0,...,n-1\}$ based on the presence or absence of an individual; in the count setting, this requires us to compare an input of $i$ against an input of $i+1$, for $i \in \{0,...,n-2\}$. This is expressed for a single column as follows.

\begin{defn}
\label{defn:neighbor_indist}
    A column vector $v \in \R^n$ is $\epsilon$-\textit{neighbor indistinguishable} if $e^{-\epsilon} v_{i+1} \le v_{i} \le e^{\epsilon} v_{i+1}$ for all $0 \le i \le n-2$.
\end{defn}

Since $\epsilon > 0$ throughout our study, it follows from the definition that all entries of an $\epsilon$-neighbor indistinguishable vector are non-negative. Denote the $j^{th}$ column of a matrix $T$ as $T_j$. We can now define differential privacy over count mechanisms as follows.\footnote{We have chosen to define differential privacy over distribution privatizers and count mechanisms separately because they accept different inputs (the table of counts versus a single count). While we have defined these separately, they are actually connected as follows. We say two non-negative integers $x,y$ are \textit{neighboring counts} if $|x-y|=1$. Then, we two count tables $d$ and $d'$ are neighboring if and only if $d_i = d'_i$ for $N-1$ rows $i \in \{0,...,N-1\}$ and there exists a unique row $j  \in \{0,...,N-1\}$ such that $d_j$ and $d'_j$ are neighboring counts. These definitions ensure, by parallel composition, that if a count mechanism satisfies $\epsilon$-differential privacy, the independent application of the mechanism to every row in a count table also satisfies $\epsilon$-differential privacy.}

\begin{defn} 
\label{defn:U}
     A count mechanism $T$ satisfies $\epsilon$-\textit{differential privacy} if and only if every column $T_j$ is $\epsilon$-neighbor indistinguishable. Let $U_{\epsilon} \subset \R^{n \times n}$ denote the set of all $\epsilon$-\textit{differentially private count mechanisms}. Since we have placed no fixed-point restrictions on $U_{\epsilon}$, we will additionally refer to it as \textit{unfixed}. When $\epsilon$ is clear from context, we will denote this set as $U$.
\end{defn}

In the context of the two-stage counting framework, we require that $T$ satisfies $\epsilon_2$-differential privacy.

The privacy constraints above can be understood as a set of linear constraints on a vector space. To make this explicit, for $i \in \{0,...,n-2\}$, we define $\pi_i^+$ to be the $n$-dimensional row vector with position $i$ equal to 1, position $i+1$ equal to $-e^{\epsilon}$, and all other positions are 0; similarly, $\pi_i^-$ is the $n$ dimensional row vector with position $i$ equal to -1, position $i+1$ equal to $e^{-\epsilon}$, and all other positions are 0. Let $\Pi \in \R^{(2n-2) \times n}$ be the matrix whose rows are given by the $2n-2$ differential privacy constraints $\pi_0^-,\pi_0^+,...,\pi_{n-2}^-,\pi_{n-2}^+$. Given this notation, a column vector $v$ is $\epsilon$-neighbor indistinguishable if and only if $\Pi v \preceq 0_{2n-2}$, where $\preceq$ denotes the element-wise inequality $\le$. Note that a constraint of this form defines a convex cone in $\R^n$. This yields the following standard result from geometry.
 
\begin{obs}
\label{obs:conic_combination}
    Any conic combination of $\epsilon$-neighbor indistinguishable vectors is also $\epsilon$-neighbor indistinguishable.  
\end{obs}

Similarly, a count mechanism $T$ is $\epsilon$-differentially private if and only if $\Pi T \preceq 0_{(2n-2) \times n}$. Next, we restrict the set of such mechanisms by adding a fixed-point constraint.

\begin{defn}
\label{defn:F_defined}
     Let $F_{\epsilon}[z] \subset U_\epsilon$ denote the set of all $\epsilon$-differentially private count mechanisms with fixed point $z$. When $\epsilon$ and $z$ are clear from context, we omit them and write $F$ instead. 
\end{defn}

We follow the convention of Holohan et al. and refer to a \textit{convex polyhedron} as the intersection of finitely many closed half-spaces, and refer to a bounded convex polyhedron as a \textit{convex polytope} \cite{holohan2017extreme}. We have defined both $U$ and $F$ as a finite set of linear constraints on $\R^{n \times n}$. Moreover, both sets are bounded because each element of a count mechanism $T$ must be between $0$ and $1$. Hence $U$ and $F$ are both convex polytopes. 

A fact from geometry is that any convex polytope can be characterized by its extreme points. We recall that $A$ is an \textit{extreme point} of polytope $P$ if $A$ cannot be written as a non-trivial convex combination of elements of $P$. That is, $A$ cannot be written as $A = \alpha X + (1-\alpha)Y$ for some distinct $X,Y \in P$ and $\alpha \in (0,1)$. We denote the set of extreme points of a convex polytope $P$ as $ex(P)$. We therefore view $ex(U)$ and $ex(F)$ as representatives of desirable classes of count mechanisms. We now aim to describe them in a way that builds intuition.

%%%%%%%%%%%%%%%%%%%%%%%%%%%%%%%%%%%%%%%%%%%%%%%%%%%%%%%%%%%%%%%
\subsection{$\epsilon$-Scales as Building Blocks}

We now introduce the important mathematical object we call a scale. This is a probability vector in which a differential privacy constraint binds for every pair of consecutive entries. Let $\Delta_n$ denote the set of probability column vectors of length $n$, $\Delta_n = \{v \in \R^{n}_{\ge 0} \text{ : } v^\top\mathds{1}_n = 1\}$.

\begin{defn}
\label{defn:scale}
We define an $\epsilon$-\textit{scale} as a probability vector $s \in \Delta_n$ such that for all $i \in \{0,...,n-2\}$, either $s_i = e^{\epsilon}s_{i+1}$ or $s_i = e^{-\epsilon}s_{i+1}$.
\end{defn}

Using our algebraic notation, we can equivalently define an $\epsilon$-scale as a vector $s \in \R^n$ such that $\mathds{1}_n^\top s = 1$ and either $\pi_i^+ s = 0$ or $\pi_i^- s = 0$ for all $i \in \{0,...,n-2\}$.

An elementary counting argument shows that there are $k = 2^{n-1}$ $\epsilon$-scales. For the remainder of this paper, we fix a matrix $\Psi \in \R^{n \times k}_{\ge 0}$ whose columns are the $k$ $\epsilon$-scales in any arbitrary order. In our notation, the $j^{th}$ scale is denoted $\Psi_j$, for $j = 0,...,k-1$.

Lemma \ref{lem:U_col_scales} shows that scales can be viewed as building blocks of each column of any matrix in $U$.

\begin{lem}
\label{lem:U_col_scales}
    If $\ T \in U$, then every column of $T$ can be written as a conic combination of $\epsilon$-scales.
\end{lem}

Since $F \subset U$, every column of matrices in $F$ can be written as conic combinations of scales as well.

\begin{cor}
\label{cor:F_col_scales}
    If $\ T \in F$, then every column of $T$ can be written as a conic combination of $\epsilon$-scales.
\end{cor}

Hinting at the utility of scales, we begin by highlighting a result that characterizes certain extreme points of $U$ in terms of $\epsilon$-scales.

\begin{thm}
\label{thm:nonzero_U}
    For $T \in U$ with all positive entries, $T \in ex(U)$ if and only if the columns of $T$ are linearly independent multiples of $\epsilon$-scales. 
\end{thm}

Given this result, one might wonder whether $\epsilon$-scales can be leveraged to give a compact description of the extreme points of $F$. This turns out to be true; however, additional machinery is required, which we develop in the remainder of this section. 

%%%%%%%%%%%%%%%%%%%%%%%%%%%%%%%%%%%%%%%%%%%%%%%%%%%%%%%%%%%%%%%
\subsection{Representing $F$ with $\epsilon$-Scales}

As we showed in the prior subsection,  every matrix $T \in F$ can be formed from $\epsilon$-scales. To proceed further, we require a more precise accounting of how much of each scale contributes to each column. 

By Corollary \ref{cor:F_col_scales}, every column of $T \in F$ can be expressed as a conic combination of scales. This means that there exist non-negative $b_{i,j}$ such that $T_j = \sum_{i=0}^{k-1} b_{i,j} \Psi_i$. In other words, $b_{i,j}$ quantifies the amount of scale $i$ that enters column $j$ of $T$. Arranging these $b_{i,j}$ into a $k \times n$ matrix $B$, we have $T = \Psi B$. We refer to $B$ as a \textit{representation matrix}, and denote the set of all such $B$ as $R_F$. A more mathematically precise definition is given below.

\begin{defn}
\label{defn:RF_definition}
    We define the \textit{fixed point representation polytope}\footnote{$R_F$ is a polytope, as it is bounded and formed from a finite number of linear constraints.} $R_{F}$ as
$$
R_{F} = \big\{B \in \R^{k \times n}_{\ge 0} \text{ : }  \Psi B \mathds{1}_n = \mathds{1}_n \text{ and } z\Psi B = z \big\}
$$
\end{defn}

For the remainder of this paper, we slightly overload notation and use $\Psi$ to refer to both the matrix $\Psi$ and the linear map formed by left multiplication by $\Psi$. The next theorem shows that $\Psi$ relates any matrix in $F$ to at least one matrix in $R_{F}$. 

\begin{thm}
\label{thm:RF_to_F_surjection}
    $\Psi $ is an affine surjection from $R_{F}$ to $F$.
\end{thm}

Thus, the surjection $\Psi$ allows us to express a matrix $T \in F$ using a matrix $B$ in the new space $R_{F}$ where the neighbor-indistinguishability constraints within each column of $T$ are replaced by linear constraints on the rows and columns of $B$. In Proposition \ref{prop:extreme_surjection_R_to_F}, we show the linear map $\Psi$ also provides a relation between the extreme points of $F$ and the extreme points of $R_{F}$. 

\begin{prop}
\label{prop:extreme_surjection_R_to_F}
    For all $T \in ex(F)$ there exists $B \in ex(R_{F})$ such that $T = \Psi B$. 
\end{prop}

Hence, every extreme point in $F$ can be written as an extreme point in $R_{F}$ under the mapping $\Psi$. However, this representation is not necessarily unique. In particular, it is possible for multiple extreme points of $R_{F}$ to map to the same extreme point in $F$. Additionally, it is possible for an extreme point of $R_{F}$ to map to a non-extreme point of $F$ under $\Psi$. We demonstrate both phenomena when $n=3$ in Example \ref{ex:unpreserved_extreme_RF_to_F} of Appendix \ref{app_details_F}. Additionally, in Corollary \ref{cor:k_perp_characterization_F} of Appendix \ref{app_geom_ext}, we provide a geometric condition on extreme points of $R_F$ that indicates exactly when they are mapped to an extreme point of $F$.

%%%%%%%%%%%%%%%%%%%%%%%%%%%%%%%%%%%%%%%%%%%%%%%%%%%%%%%%%%%%%%%
\subsection{Characterizing the Extreme Points of $R_{F}$}

Up to this point, we have argued that differentially private count mechanisms with fixed point $z$ are represented by the extreme points of $F$, and there is a relationship between these extreme points and those of the representation polytope $R_F$. Next, we characterize the extreme points of $R_{F}$ in terms of the scales they encode.

Given a column vector $x \in \R^k$, let $H(x)$ be the set of vectors
$$
h = (z\Psi_u)^{-1}\Psi_u - (z\Psi_v)^{-1}\Psi_v
$$
where $v$ is the smallest index such that $x_v > 0$ and $u$ is another index such that $x_u > 0$.\footnote{Note that $(z\Psi_j)^{-1}$ is well-defined for all $j \in \{0,...,k-1\}$, since $z\Psi_j > 0$ because $z$ is a probability row vector and $\Psi_j$ is an $\epsilon$-scale.} Note that, under this definition, if a vector $x$ has either 0 or 1 positive entry, $H(x) = \emptyset$. Also, for $h \in H(x)$, $zh = 0$ by routine computation.

\begin{defn}
    \label{defn:psi_aff_simple} We call a $k \times n$ matrix $B$ \textit{$\Psi-$affinely simplified} if the additive union\footnote{The \textit{additive union} of sets $X$ and $Y$ is the multiset denoted as $A \uplus B$, where the multiplicity of $s \in X \uplus Y$ equals the sum of the multiplicity of $s$ in $X$ and of $s$ in $Y$ \cite{blizard1989multiset}.} $\uplus_{j=0}^{n-1} H(B_j)$ is linearly independent. 
\end{defn}

Let $M_{F}(\Psi)$ be the set of $k \times n$ matrices $\mu$ such that $z\Psi \mu = 0_n^\top$ and $\Psi \mu \mathds{1}_n = 0_n$. By construction of $M_{F}(\Psi)$, if $X,Y \in R_{F}$, then $X-Y \in M_{F}(\Psi)$. Conceptually, $M_{F}$ represents ``movement matrices'' that can potentially be added to matrices in $R_{F}$ to move through the affine space. The following lemma and theorem formalize the equivalence between the extreme points of $R_{F}$, the notion of $\Psi$-affinely  simplified matrices, and these movement matrices.

\begin{lem}
\label{lem:psi_simple_relationship_movement_matrices}
    Suppose $B \in R_{F}$. Then $B$ is $\Psi$-affinely  simplified if and only if there does not exist a non-zero matrix $\mu \in M_F(\Psi)$ such that $\mu_{i,j}$ is zero whenever $b_{i,j}$ is zero.
\end{lem}

\begin{thm}
\label{thm:extreme_f_characterization_general}
    Given scale matrix $\Psi$, $B \in ex(R_{F})$ if and only if $B \in R_{F}$ is $\Psi$-affinely simplified.
\end{thm}

Hence, the extreme points of our representation space are those that are $\Psi$-affinely simplified. Additionally, Theorem \ref{thm:extreme_f_characterization_general} can be used to quantify combinatorial properties of $B \in ex(F)$, such as the number of positive entries.

\begin{cor}
\label{cor:nonzero_RF}
    Let $\mathcal{P}$ denote set of indices where $z$ is positive. Then extreme matrices in $R_{F}$ contain at least $|\mathcal{P}|$ and at most $|\mathcal{P}| + n - 1$ positive entries.
\end{cor}

Thus the extreme points of $F$ can be represented using relatively few scales. This suggests that constructor algorithms $\mathcal{C}$ that work at the granularity of scales could be used to build an extreme point of $F$ efficiently. We showcase an example of this in the next section.

\section{Constructor Algorithms}
\label{algorithmic}
In this section, we turn our attention to the design of the constructor algorithm. This component of the two-stage counting framework outputs a single count mechanism from within a set of allowable count mechanisms. We first discuss fixed-point constructors before describing the two-staged unfixed optimum constructor.

\subsection{Fixed-Point Constructors}

The goal of a fixed-point constructor is to accept target distribution $z$ and output a single mechanism inside $F_\epsilon[z]$. Recall from the introduction that we identified three performance criteria for a two-stage counting framework. The first criterion, accuracy of distribution, is supported by the inclusion of a fixed-point constraint. Moreover, in the experiments we describe later, we observe that all fixed-point constructors we create behave very similarly with respect to this criterion. We therefore set accuracy of distribution aside and focus the following discussion on the other criteria: accuracy of counts and runtime.

%When designing a fixed-point constructor, we primarily stress two of these three: accuracy of counts and runtime. The reason for this is that the final criterion, accuracy of distribution, is already being supported by the inclusion of a fixed-point constraint. 

Beginning with accuracy of counts, we consider a natural class of count-error measures -- those that are linear in the coefficients of the transition matrix. Given a count mechanism $T$, and an $n \times n$ matrix of weights $W$, we define the count error of $T$ under $W$ with the Frobenius inner product  $\langle W,T\rangle = tr[W^\top T] =\sum_{i=0}^{n-1}\sum_{j=0}^{n-1} w_{i,j} t_{i,j}$. For example, $w_{ij} = z_i |i-j|$ yields expected absolute deviation:
$$\langle W,T\rangle = \sum_{i=0}^{n-1}\sum_{j=0}^{n-1} z_i |i-j| t_{i,j} =\E_{i \sim z} \E_{j \sim \text{Row $i$ of $T$}}\big[|i-j|\big] 
$$
Similarly, setting $w_{ij} = z_i (i-j)^2$ encodes mean squared error. Since $F_\epsilon[z]$ is a polytope, the optimum of such a function can be achieved at one of the extreme points $ex(F_\epsilon[z])$ (Lemma \ref{lem:linear_alg_opt}).

Given a linear count-error function, it is useful to divide the task of choosing a mechanism in $F$ into two cases. First, when $n$ is small, a mechanism that minimizes count error can be found using standard linear programming algorithms. These include open-source and commercial variants of the simplex algorithm and the interior point method. We refer to a constructor algorithm that minimizes count error over $F$ as a \textit{count-error-minimizing fixed-point constructor}.

For larger values of $n$, the runtime of linear programming algorithms can be prohibitively large. For situations like these, we provide a heuristic algorithm that outputs a single extreme point of $F_\epsilon[z]$. Our algorithm is inspired by our scale representation, but only requires polynomial space, and can run as fast as $O(n^2)$. 

\subsection{Heuristic Constructor Algorithm}

To motivate our algorithm, we begin with a piece of intuition. To make the output count close to the input count, we want the probability near the diagonal to be high, and the probability far away from the diagonal to be low. This leads us to the notion of a \textit{single-peaked scale at position $j$}, which increases to a single peak at row $j$, then decreases until the end. The following definitions make this notion precise.
\begin{defn}
A \textit{pattern} is an $(n-1)$-tuple $p \in \{-1,+1\}^{n-1}$. We say an $\epsilon$-scale $s$ \textit{has pattern} $p$ if, for all $j \in \{0,...,n-2\}$, $p_j = 1$ if $s_j< s_{j+1}$, and $p_j = -1$ otherwise.
\end{defn}
\begin{defn}
\label{defn:single-peaked}
    The \textit{single-peaked pattern} at $j \in \{0,...,n-1\}$ is the pattern $p$ such that $p_i = 1$ for all $i < j$ and $p_i = -1$ otherwise. We will refer to an $\epsilon$-scale with a single-peaked pattern at $j$ as a \textit{single-peaked scale at position $j$}.
\end{defn}
One might hope that a mechanism $T$ can be constructed solely from singled-peaked scales. This is possible for $T \in U$, but is not possible for $T \in F$ in general given the fixed-point constraint. However, single-peaked scales can provide a useful starting point for constructing a count mechanism in $F$. This informs the design of Algorithm \ref{alg:heuristic}.

The core idea underlying our algorithm is as follows. Imagine starting with the $n \times n$ zero matrix $A$ and then adding scales to it until the row sum and fixed-point constraints of $F$ are fulfilled. By Observation \ref{obs:conic_combination}, adding positive multiples of $\epsilon$-scales always maintains $\epsilon$-neighbor indistinguishability. As we add scales, we can keep track of how much still needs to be added to each row, $r = \mathds{1}_n - A \mathds{1}_n $, and how much still needs to be added to each column, $c = z - z A$. At the end of the algorithm, we must have $r = 0_n$ and $c=0_n^\top$. If at any point when adding scales, an element of $c$ becomes negative, it will not be possible to arrive at $c = 0_n^\top$ by adding more scales. Similarly, if at any point $r$ becomes not $\epsilon$-neighbor-indistinguishable, it will not be possible to add more scales to end with $r=0_n$ (Proposition \ref{prop:max_compute} quantifies how much of a scale $s$ can be added without triggering one of these two conditions).

Applying this idea, Algorithm \ref{alg:heuristic} proceeds iteratively, filling one column at a time. When operating on a particular column $j$, if it is possible to add some multiple of the single-peaked scale at position $j$ while preserving the $\epsilon$-neighbor-indistinguishability of $r$, we greedily add as much of the scale as possible. Otherwise, we adjust the single-peaked scale as necessary so that it matches the pattern of $r$ at any row where a privacy bound holds for $r$. This guarantees that we are able to add some positive multiple of the scale to $A$. The details can be observed in the proof of Theorem \ref{thm:heur_halts}. We repeat the process of adding scales to column $j$ until the column is at capacity ($c_j = 0$).

When execution begins, and any time a column has been filled to capacity, the algorithm must select a new column $j$ to operate on. The selection of a new column is delegated to a function $\kappa$, which is an input parameter. The only requirement for the results of this section is that $\kappa$ always returns a column index $j$ with $c_j > 0$. For example, $\kappa$ could select the column $j$ for which $c_j$ is maximal and non-zero. In principle, one could choose many other options for $\kappa$, which may result in varying runtimes and have a significant effect on the accuracy of counts. We observe this phenomenon for three possible choices for $\kappa$ in the experiments of Section \ref{experiments}. 

\begin{algorithm}
%\DontPrintSemicolon % Some LaTeX compilers require you to use \dontprintsemicolon    instead
\KwIn{
\begin{itemize}
    \item Fixed point $z$
    %\vspace{-3mm}
    \item Privacy parameter $\epsilon$
    %\vspace{-3mm}
    \item Column selection function $\kappa$ 
\end{itemize}
}
%\vspace{-3mm}
%\KwOut{$n \times n$ matrix}
\vspace{2mm}
Initialize an $n \times n$ matrix $A$ of 0's, $r = \mathds{1}_n$, and $c = z$\\
\While{$r \ne 0_n$ (equivalently, $c \ne 0_n^\top$)}{
Select a column $j = \kappa(z, r, c, A)$ such that $c_j > 0$.\\
\While{$c_j > 0$} {
Assign $p$ as the single-peaked pattern for $j$.\\
Adjust $p$ as follows: for any $i$ which $\frac{r_{i+1}}{r_{i}} = \exp(-p_i \epsilon)$, reassign $p_i \leftarrow -p_i$.\\
Let $s$ be the $\epsilon$-scale with pattern $p$.\\
Compute $q = \max\{\gamma \text{ : } \gamma zs \le c_j \text{ and }r - \gamma s \text{ is $\epsilon$-neighbor indistinguishable}\}$.\\
Update $A_j \leftarrow A_j + qs$, $r \leftarrow r - qs$, and $c_j \leftarrow c_j - q z s$.\\
}
}
\textbf{Return} $A$
\caption{Heuristic Constructor Algorithm}
\label{alg:heuristic}
\end{algorithm}

The following two theorems demonstrate that Algorithm \ref{alg:heuristic} can run efficiently and outputs an extreme point of $F$.

\begin{thm} 
\label{thm:heur_halts}
If the worst-case runtime of selection algorithm $\kappa$ is $\mathcal{O}(\tau)$, then the worst-case runtime of Algorithm \ref{alg:heuristic} is $\mathcal{O}\left(\max\{\tau n, n^2\}\right)$.
\end{thm}

\begin{thm} 
\label{thm:heur_output_extreme_F}
 Algorithm \ref{alg:heuristic} outputs a matrix in $ex(F)$. 
\end{thm}

\smallskip

\noindent \textbf{Selection Guidelines for $\kappa$ in Practice.} The reader might wonder how to choose a column selector in order to have the lowest count error for a given dataset. Surprisingly, once a data holder generates the privatized distribution of counts $z$, they can then compute the count error analytically for any number of columns selectors without incurring additional privacy loss. 

The rationale is an iterated application of the post-processing guarantee. First, distribution $z$ is the only input to Algorithm \ref{alg:heuristic} that is computed from the underlying data. Therefore, when Algorithm \ref{alg:heuristic} with a column selector $\kappa$ produces an output $T$ using $z$, no additional privacy loss is incurred. Second, the metrics we have chosen to represent count error $\langle W, T \rangle$ depend only on $z$ and $T$, so they can be computed without additional privacy loss. This argument holds for any number of columns selectors $\kappa$ tested. Hence, a data holder can create a transition matrix for every available column selector and keep the one with the lowest count error without incurring additional privacy loss beyond the  $\epsilon_1$ used to generate $z$.

%%%%%%%%%%%%%%%%%%%%%%%%%

\subsection{Two-Stage Unfixed Optimum Constructor}

While the main focus of this section was the creation of fixed-point constructors, we will want to compare our techniques against a set of baselines that do not enforce a fixed-point constraint. One such baseline can be developed by adapting our two-stage framework. We envision a constructor algorithm that minimizes count error over the unfixed polytope $U$ instead of $F$. Even though $U$ does not depend on the distribution $z$, the distribution privatization step is still required since common measures of count error depend on the distribution of counts (e.g., expected total deviation is given by $\langle W, T\rangle$, with $w_{ij} = z_i |i-j|$). For that reason, we now develop a constructor that selects a count mechanism from $U$. This will form the basis of our \textit{two-stage unfixed optimum}. 

As was the case for $F$, it would theoretically be possible to minimize count error over $U$ with standard linear programming algorithms, but these would take impractically long for some of the datasets we study. We therefore seek to leverage the structure of polytope $U$ to find the lowest count error point more efficiently.

First, in Appendix \ref{app_characterizing_U}, we show that count mechanisms in $U$ can be represented using scales in a similar way as those in $F$, yielding a new representation polytope, $R_U$ (Theorem \ref{thm:RU_to_U_surjection}).  Next, in Appendix \ref{app_ghosh_scale} we leverage a result from Ghosh et al. \cite{ghosh2009universally}, which tells us exactly which scales make up the optimal point in $U$ for a natural class of count error measures, which we call \textit{row-wise concentrating}.

\begin{defn}
    \label{defn:row_wise_concentrating}
    A weight matrix $W \in \R^{n \times n}_{\ge 0}$ is \textit{row-wise concentrating} if $w_{i,j}$ is non-decreasing in $|i-j|$ for each $i \in \{0,...,n-1\}$.
\end{defn}

Theorem \ref{thm:ghosh_which_scales} tells us that the scales that make up a count mechanism in $U$ that minimizes a row-wise concentrating weight matrix are exactly the single peaked scales. We let $\Sigma$ be the $n \times n$ matrix, where column $j$ is the single-peaked scale at position $j$. 

We will further restrict attention to a narrower class of weight matrices, which will enable certain optimizations, stated below.

\begin{defn}
    A weight matrix $W$ is \textit{row-wise convex} if for every $i \in \{0,...,n-1\}$ and every $j \in \{0,...,n-2\}$, $w_{i,j+1} - w_{i,j}$ is non-decreasing in $j$.
\end{defn}

Note that commonly used count-error measures, including expected absolute deviation and mean squared error, are row-wise concentrating and row-wise convex. Row-wise convexity is useful because it ensures that once we find a column for a single-peaked scale that minimizes count error in a local sense, it is also optimal over all columns (Lemma \ref{lem:weight_matrix_scale_convex}). Putting these ideas together, Algorithm \ref{alg:ghosh_scale} places the single peaked scales into a transition matrix one at a time. 

\begin{algorithm}
%\DontPrintSemicolon % Some LaTeX compilers require you to use \dontprintsemicolon    instead
\KwIn{
\begin{itemize}
    \item Privacy parameter $\epsilon$
    \item Weight matrix $W$ that is both row-wise concentrating and row-wise convex
    %\vspace{-3mm}
\end{itemize}
}
%\vspace{-3mm}
%\KwOut{$n \times n$ matrix}
\vspace{2mm}
Initialize $\Sigma$ and  $\{\omega_l\}_{l=0}^{n-1}$ as described in Observation \ref{obs_ghosh_weights}.\\
Initialize an $n \times n$ matrix $A$ of 0's, $j=0$, and $l=0$. \\
\While{$l \le n-1$}{
$current\_error \leftarrow \omega_l (W_j)^\top \Sigma_l$ \\
\uIf{$j < n-1$} {
$next\_error \leftarrow \omega_l (W_{j+1})^\top \Sigma_l$
}
\Else {
$next\_error \leftarrow \infty$
}
\uIf{$next\_error > current\_error$} {
$A_j \leftarrow A_j + \omega_l \Sigma_{l}$\\
$l \leftarrow l+1$
}
\Else {
$j \leftarrow j + 1$
}
}
\textbf{Return} $A$
\caption{Two-Stage Unfixed Optimum Constructor}
\label{alg:ghosh_scale}
\end{algorithm}

The following two theorems establish that the two-stage unfixed optimum runs efficiently and outputs a count mechanism that minimizes count error.

\begin{thm}
\label{thm:unfixed_runtime}
The worst-case runtime of Algorithm \ref{alg:ghosh_scale} is $\mathcal{O}(n^2)$.
\end{thm}

\begin{thm}
    \label{thm:unfixed_optimal}
    Algorithm \ref{alg:ghosh_scale} outputs a count mechanism $O \in ex(U)$ that minimizes count error. That is, for all $\hat O \in U$, $\langle W, O \rangle \le \langle W, \hat O \rangle$.
\end{thm}

\section{Experiments}
\label{experiments}
In this section, we explore the \textit{practicality} of fixed-point methods for privatizing a table of counts. We perform a set of simulations in which we measure the performance of different privatization methods. Our simulations are grouped into three experiments, corresponding to our three performance criteria: accuracy of distribution, accuracy of counts, and runtime. As we detail in this section, there are some scenarios in which a fixed-point method can dramatically improve accuracy of distribution with only moderate cost to accuracy of counts and runtime; however, considerable nuance can be found through each of our experiments. 

\begin{figure*}[h]
    \centering
    \includegraphics[width=\linewidth]{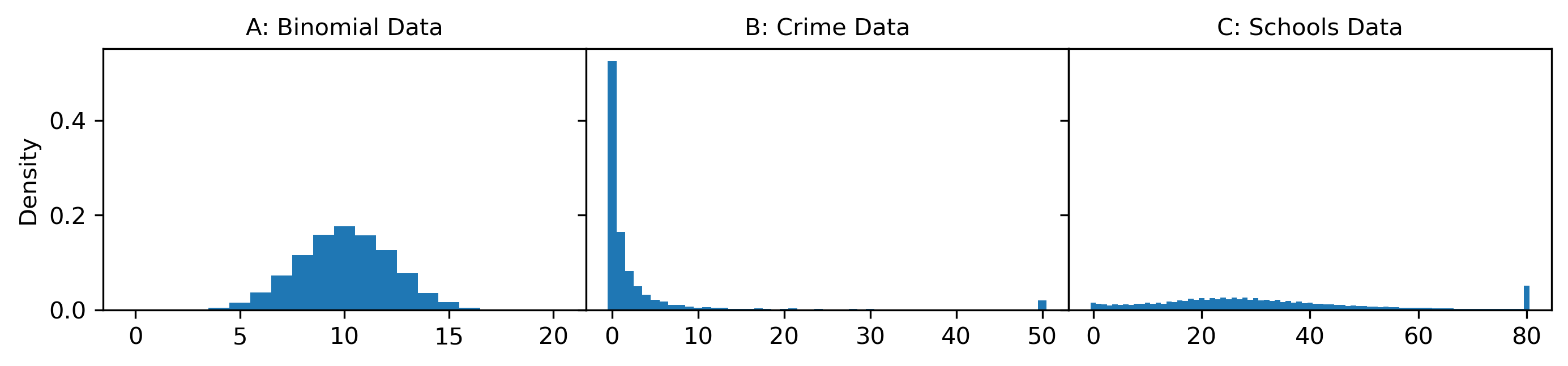}
    \vspace*{0.75mm}
    \caption{Distributions of the three datasets used in our experiments: (A) synthetic draws from a binomial distribution, (B) observed counts of homicides in US counties, and (C) observed counts of teachers in public schools.} 
    \label{fig:distributions}
\end{figure*}

Each of our simulations represents a complete privatization of a table of counts. All experiments were written in Python 3.12.7 and executed on a single machine. We include the following privatization methods.

\smallskip

\noindent \textbf{Fixed-Point Methods.} We implement our two-stage fixed-point framework. In the role of distribution privatizer $\mathcal{D}$, we select the cyclic Laplace mechanism, guided by the favorable results of Section \ref{privatizing_z}. While other choices for $\mathcal{D}$ are possible, fixing the distribution privatizer allows us to focus on the effects of other experimental parameters. As a robustness check, we replicated key experiments with each of the other distribution privatizers from Section \ref{privatizing_z} and observed similar trends. For the constructor algorithm $\mathcal{C}$, we implement the following choices.

\begin{itemize}[leftmargin=*]
    \item The simplex algorithm provided by the open-source SciPy Python library.
    \item The interior point method provided by the open-source SciPy Python library.
    \item Our heuristic algorithm, Algorithm \ref{alg:heuristic}. Furthermore, we explore several options for the column selector $\kappa$.
    \begin{itemize}
        \item Max probability. This selector always chooses the remaining column $j$ with maximum $z_j$.
        \item Min probability. This selector always chooses the remaining column $j$ with minimum $z_j$.
        \item Sandwich. This selector always follows a prescribed pattern: $0, n-1, 1, n-2, 2...$ until every column has been chosen.
    \end{itemize}
\end{itemize}

\smallskip

\noindent \textbf{Unfixed Baselines.} As points of comparison, we implement a set of privatization algorithms that do not enforce a fixed-point constraint.

\begin{itemize}[leftmargin=*]
    \item Two-stage unfixed optimum given in Algorithm \ref{alg:ghosh_scale}.
    \item Truncated geometric of \citet{ghosh2009universally}. This mechanism may not coincide with the output of the two-stage unfixed optimum, because it does not include the post-processing described by Ghosh el al. As such, there is no guarantee that the truncated geometric minimizes a given notion of count error.
    \item Staircase mechanism of \citet{geng2015staircase}, which minimizes count error for a certain class of mechanisms whose outputs are in $\R$. Because this mechanism outputs a real number, we post-process the output by mapping to the nearest integer in $\{0,...,n-1\}$.\footnote{The staircase mechanism includes a parameter $\gamma \in [0,1]$, which dictates the width the center-most step. When $\gamma = 1/2$ the staircase mechanism coincides with the truncated geometric. We follow the authors' recommendation and set $\gamma = (1+\exp(\epsilon/2))^{-1}$ to minimize the expected noise amplitude, resulting in a distinct mechanism.  However, in our experiments, we found that the performance of the staircase mechanism was visually indistinguishable from the truncated geometric.}
    \item Discrete Gaussian of \citet{canonne2020discrete}. This is the integer-analog of the commonly used Gaussian mechanism, and satisfies the weaker privacy notion of approximate $(\epsilon, \delta)$-differential privacy. We set $\delta = 1 / (N+1)$, maximizing the benefit to this baseline while satisfying the convention that $\delta < 1/N$ \cite{dwork2014algorithmic, page2018differential}. We post-process the output of this mechanism to ensure that it lies between $0$ and $n-1$.
\end{itemize}

Note that the latter three mechanisms do not utilize a probability distribution as input. Accordingly, they are implemented as stand-alone algorithms, not part of a two-stage framework.

For each of our privatization methods, we vary the following experimental parameters.

\begin{figure*}[h]
    \centering
    \includegraphics[width=\linewidth]{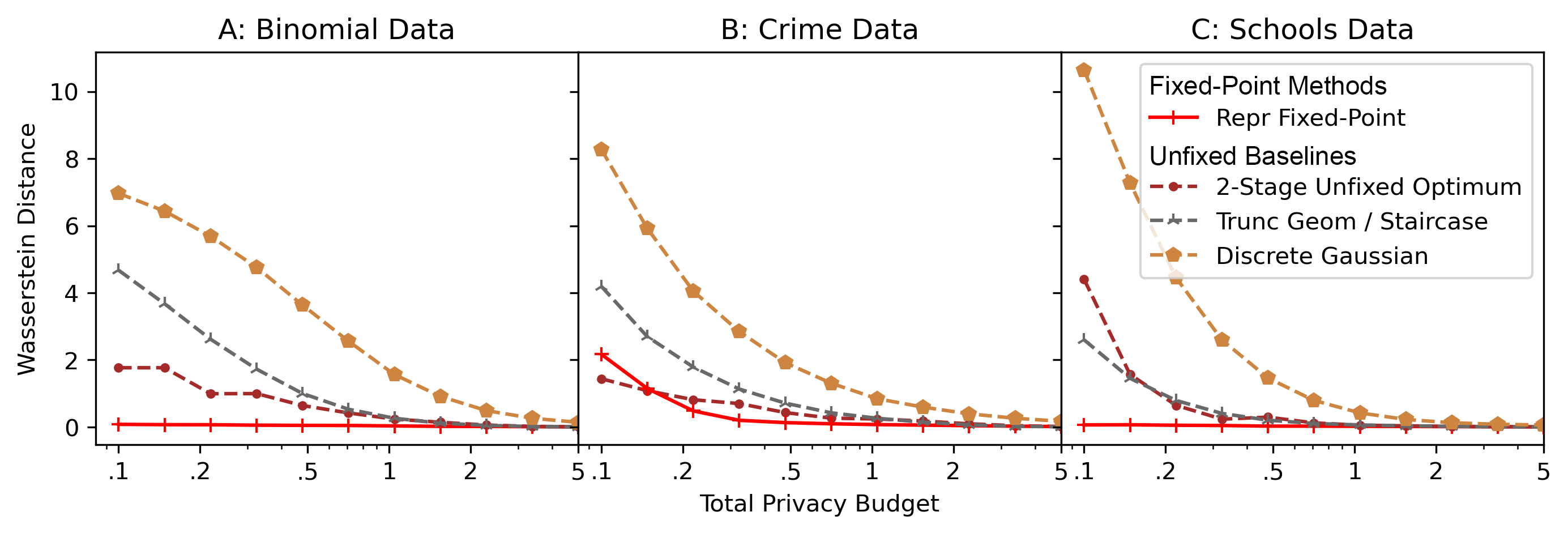}
    \vspace*{-0.15cm}
    \caption{Distribution error under different total privacy budgets on logarithmic scale. All fixed-point constructors were visually overlapping, so we include only the heuristic with sandwich selector to represent all fixed-point constructors.} 
    \label{fig:distribution_accuracy}
\end{figure*}

\smallskip

\noindent \textbf{Input Data.}  We select three datasets to have a range of distributional characteristics (see Figure \ref{fig:distributions}).

\begin{itemize}[leftmargin=*]
    \item 10,000 synthetic draws from a Binomial distribution with size parameter $20$ and probability parameter $1/2$. This is a symmetric distribution that closely approximates a normal distribution.
    \item Real counts of homicides in each US County.\footnote{Downloaded from https://www.kaggle.com/datasets/mikejohnsonjr/united-states-crime-rates-by-county} When conducting accuracy experiments, we top code these counts at 50. This dataset has a considerable right skew (5.83) and a substantial peak at zero (53\% of counties report zero homicides). It is our smallest dataset with 3,136 total counties.
    \item Real counts of full-time teachers in public elementary and secondary education facilities in the US.\footnote{Downloaded from https://public.opendatasoft.com/explore/dataset/us-public-schools/} When conducting accuracy experiments, we top code these counts at 80. This data was published by the US Department of Education for the 2017-2018 school year. This dataset has a mild positive skewness (1.29). It is our largest dataset, representing 95,969 schools.
\end{itemize}

\smallskip

\noindent \textbf{Privacy Parameters.} We set a total privacy budget, $\epsilon_t$, to values ranging from $0.1$ to $5$. As noted in Section \ref{subsec:proposed_workflow}, the two-stage counting framework requires this budget to be divided into two components: $\epsilon_1$ for the distribution privatizer $\mathcal{D}$, and $\epsilon_2$ for the count mechanism $T$, with $\epsilon_t = \epsilon_1 + \epsilon_2$. Let $f = \epsilon_1 / \epsilon_t$ represent the fraction allocated to the distribution privatizer. Setting $f$ is complicated by two stylized facts. First, dividing the privacy budget requires a value judgement, as it trades off between accuracy of counts and accuracy of distribution. Second, both types of accuracy generally depend on the underlying data, and existing methods of estimating these for candidate values of $f$ require considerable privacy budget. For these reasons, we pursue the development of a rule of thumb that is independent of data, and that practitioners can use to set $f$.

We detail the development of our rule of thumb in Appendix \ref{app_additional_figures}. To summarize our process, we begin with a \textit{combined error objective} of the form $\ln(\text{count\_error}(f)) + \ln(\text{distribution\_error}(f))$. This objective has the interpretation that, at the optimum, a multiplicative decrease in one type of error must be matched by an equal (or greater) multiplicative increase in the other. We perform a set of simulations on six new synthetic datasets, varying the total privacy budget, and computing the value of $f$ that minimizes the combined error objective. We observe an inverse relationship between $\epsilon_t$ and the optimal $f$ that appears in all datasets tested. We fit a model on these simulation results, yielding the following rule of thumb.
$$
\hat f(\epsilon_t)  = 0.106 + 0.533\exp{(-2.87 \epsilon_t)}
$$
We employ this rule of thumb  whenever we simulate the two-stage framework. Because the truncated geometric, staircase, and discrete Gaussian mechanisms do not require a two-stage framework, we grant them the entire privacy budget $\epsilon_t$.

\begin{figure*}[h]
    \centering
    \includegraphics[width=\linewidth]{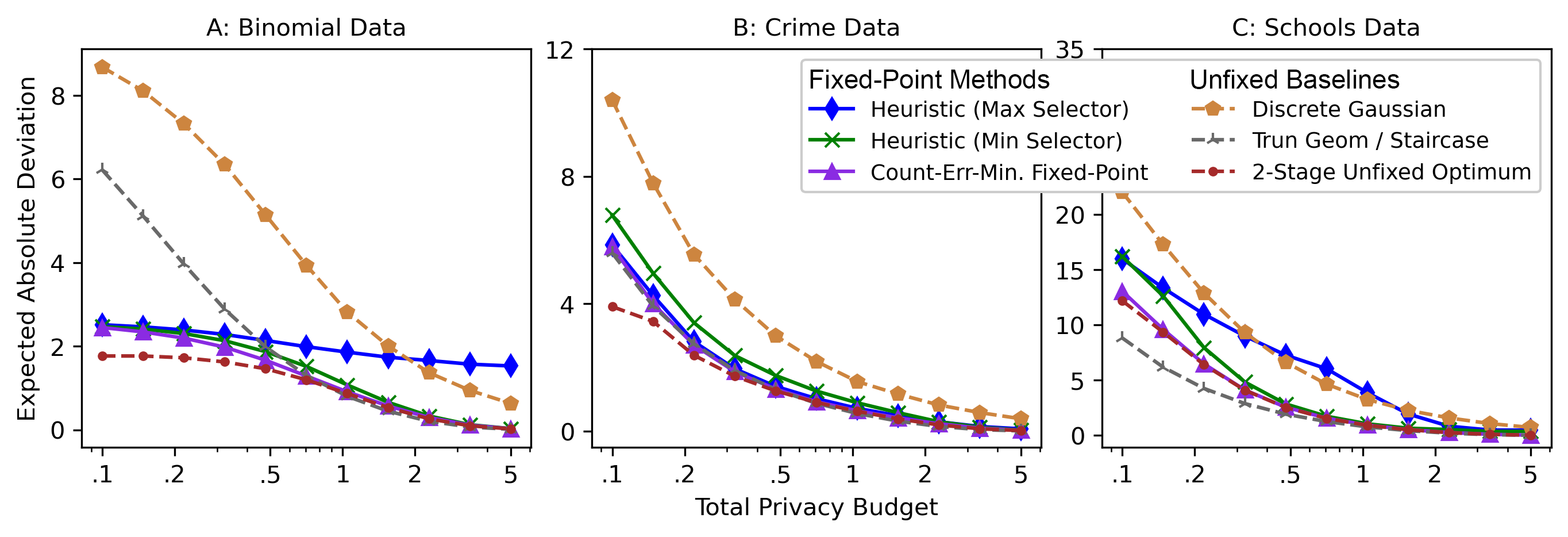}
    \vspace*{-0.15cm}
    \caption{Count error under different total privacy budgets on logarithmic scale. As expected, an unfixed baseline yields lower count error than the fixed-point methods.} 
    \label{fig:noise}
\end{figure*}

\subsection{Accuracy of Distribution}

In our first experiment, we investigate how the choice of privatization method affects accuracy of distribution. We run a set of simulations, varying the privatization  method, dataset, and privacy parameters as described above.  To operationalize accuracy of distribution, we use total variation distance, Kolmogorov–Smirnov (KS) distance, and Wasserstein-1 distance. 
Figure \ref{fig:distribution_accuracy} shows the Wasserstein distance between the output distribution of counts and the true input distribution from a set of simulations; the corresponding plots for KS distance and total variation distance are similar and we omit them. We found that all fixed-point methods were closely overlapping, so we represent them with a single curve. Similarly, the truncated geometric and staircase mechanisms were closely overlapping, so we also represent them with a single curve. 

As seen in Figure \ref{fig:distribution_accuracy}, the introduction of a fixed-point constraint often leads to a substantial improvement in accuracy of distribution. For example, for the binomial dataset at $\epsilon_t = 0.48$, switching from the best performing baseline (two-stage unfixed optimum) to a fixed-point constructor reduces Wasserstein distance from $0.64$ to $0.04$, a $94\%$ decrease in distribution error. For the crime and schools dataset, the corresponding decreases in distribution error are $74\%$  and  $87\%$ respectively. There are exceptions in which an unfixed baseline has slightly better accuracy of distribution than a fixed-point constructor. In particular, for the crime data at $\epsilon = 0.1$, Wasserstein distance is $1.44$ for the two-stage unfixed optimum, but rises to $2.17$ for the fixed-point methods. Aside from the crime data, the fixed point constraint seem to offer the greatest advantage when $\epsilon_t$ is small. In particular, at $\epsilon_t = 0.1$, for the schools dataset, the fixed-point methods feature a Wasserstein distance of $0.06$. The closest baseline is the truncated geometric mechanism, with a Wasserstein distance of $2.60$ -- more than 40-times larger.

Overall, we believe this first experiment presents promising evidence in favor of imposing a fixed-point constraint. While such a constraint does not lead to perfect equality between the output distribution of counts and the true input, it represents a substantial improvement over current practice.

\subsection{Accuracy of Counts}

Our second experiment investigates how our fixed-point methods affect accuracy of counts. As noted earlier, we expect our techniques to involve some performance loss in this dimension. One reason is that they select a count mechanism from the polytope $F$, which is a strict subset of the unfixed polytope $U$. A more intuitive reason is that fixed-point techniques split attention between accuracy of distribution and accuracy of counts, whereas prior techniques focus exclusively on accuracy of counts. While some performance loss is expected, we believe that its magnitude is important. If users are to adopt a fixed-point technique, the loss in accuracy of counts cannot be too large, relative to the gain in accuracy of distribution.

To measure how much fixed-point methods degrade accuracy of counts, we run a set of simulations, varying the privatization method, dataset, and privacy parameters as before. To operationalize count error, we use both expected absolute deviation and mean squared error. Figure \ref{fig:noise} shows the expected absolute deviation from a set of simulations; the corresponding plot for mean squared error is similar so we omit it. Additionally, we omit the curve for the sandwich selector, noting that it closely overlaps with the better of the max or min selector for all datasets. 

We first observe that the two-stage unfixed optimum is consistently near the bottom of each plot, which we expect since this method outputs the best possible count mechanism in the polytope $U_{\epsilon_2}$. Somewhat surprisingly, there are situations in which the truncated geometric and staircase mechanisms outperform the two-stage unfixed optimum.  The reason this is possible is that these mechanisms don't require a privatized estimate of $z$, and hence the entire privacy budget $\epsilon_t$ can be used to privatize the table of counts. The truncated geometric and staircase mechanisms have the most visible advantage in the schools data, with the advantage shrinking with higher $\epsilon_t$. One possible explanation relates to the shape of the schools data, which is the closest to a uniform among our three datasets. Recall that the two-stage unfixed optimum can be viewed as a truncated geometric mechanism with post-processing \cite{ghosh2009universally}. When data is near uniform, we expect that the optimal transition matrix in $U_{\epsilon_2}$ closely resembles the truncated geometric, so the benefits of post-processing are small. On the other hand, the truncated geometric mechanism enjoys the benefit of choosing a point from the larger polytope $U_{\epsilon_t}$, even though that point is not optimal.

Among the fixed-point methods, the count-error-minimizing fixed-point constructors (i.e. fixed-point constructors using simplex or the interior point method) have the lowest count error. These methods are never the lowest curve in the graph, and in particular, the two-staged unfixed optimum always has lower count error. One reason is that adding constraints reduces the set of feasible mechanisms (from $U_{\epsilon_2}$ to $F_{\epsilon_2}$). %, which may increase the minimum amount of count error in the polytope. 
Even though switching from the best baseline to a count-error minimizing fixed-point method does increase count error, the gap is often moderate. For example, for the crime data, at $\epsilon_t = 0.48$, switching from the best baseline (two-stage unfixed optimum) to a count-error-minimizing fixed-point method increases count error from 
$1.27$ to $1.32$, a $5.7\%$ increase (recall from the prior section that for the same switch, distribution error decreased by $74\%$). The largest performance gaps tend to occur for small values of $\epsilon_t$. 

The heuristic fixed-point constructors appear as a set of curves above the count-error-minimizing fixed point methods. Since these constructors also output a point in $F_{\epsilon_2}[z]$, they must result in at least as much count error as the minimum value over $F$. An interesting finding is that no column selector consistently outperforms the alternatives. In fact, by manipulating experimental parameters, we found situations in which each of the column selectors we studied was the best choice. As a general observation, the sandwich selector tends to have either the lowest count error, or close to the lowest for all datasets and parameter choices we tested. The max selector performed well for the highly skewed crime dataset, but poorly for the other two. Alternatively, the min selector performed well for the binomial and schools data, both of which exhibit more symmetry than the crime data. 

When using a heuristic constructor is necessary (because it is computationally impractical to find the fixed-point optimum), the cost of imposing a fixed-point constraint, relative to the best unfixed baseline, is larger. However, the additional cost of switching from a count-error-minimizing fixed-point constructor to a heuristic constructor is usually less than the cost of switching from the best unfixed baseline to a count-error-minimizing fixed-point constructor. For example, in the crime dataset, when $\epsilon_t = 1$, the expected absolute deviation is $0.55$ for the truncated geometric, $0.64$ for the count-error-minimizing fixed-point constructors, and $0.66$ for the best-performing heuristic (sandwich selector).

Overall, we find these results promising. As expected, count error does increase when adopting a fixed-point constraint. However, the increase is often a matter of a few percentage points and we expect that users will find this acceptable in some scenarios.

\subsection{Runtime}

In our final experiment, we measure how the choice of privatization method affects runtime. This is important because an algorithm must finish running in a reasonable time to be deployable.

We measure runtime using the system clock on a computer as it executes each privatization algorithm. Since each constructor algorithm must choose an $n \times n$ matrix, we expect that the bound on the number of counts, $n-1$, would have a strong effect on runtime. We therefore adjust the datasets used in this experiment so that we can vary $n$. For the binomial dataset, we continue taking draws with probability parameter $1/2$ but with size parameter set to $n-1$. For the crime and school datasets, we use the same raw data, but top code counts at $n-1$. 

\begin{figure}[!t]
    \centering
    \includegraphics[width=\linewidth]{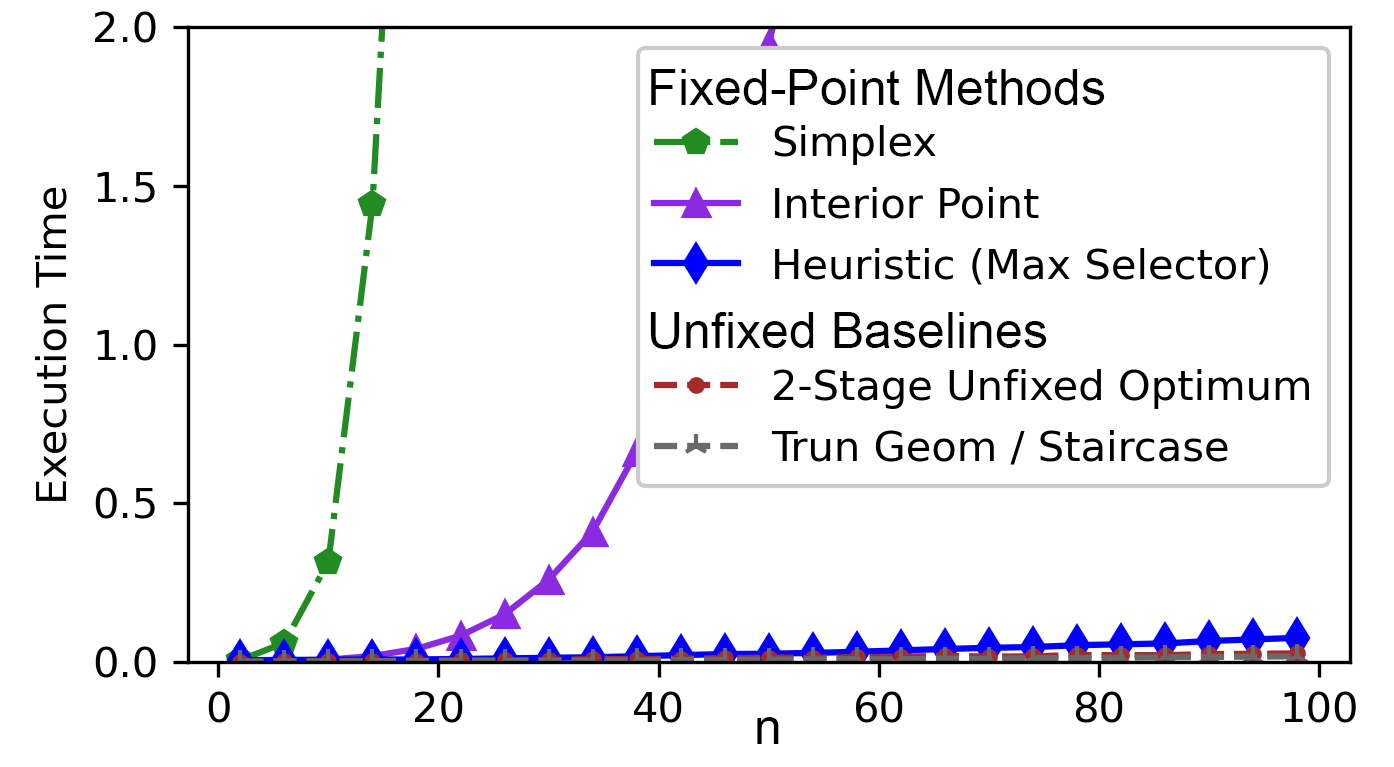}
    \vspace*{-0.15cm}
    \caption{Runtime for crime data with $\epsilon_t = \ln(2)$.} 
    \label{fig:execution}
\end{figure}

Measurements of runtime for different values of $n$ for the crime dataset are shown in Figure \ref{fig:execution}. Plots for the binomial and schools dataset are similar and we omit them. Additionally, we found that the selector function had no visual impact on the runtime of our heuristic algorithm; we therefore include just a single selector function in the plot. Omitted from the plot is the Discrete Gaussian, whose runtime is consistently off the plot due to its use of rejection sampling \cite{canonne2020discrete}. The truncated geometric and staircase mechanisms were closely overlapping, so they are represented by a single curve.

For all parameter choices, we find that the truncated geometric and staircase baselines provide the fastest runtime. These are followed closely by the two-stage unfixed optimum. Among the fixed-point constructors, we find that the one that uses the interior point method, and especially the one that uses simplex, exhibit rapid growth rates. This suggests that their use must be limited to smaller values of $n$. Our heuristic constructor, on the other hand, performs more efficiently, with runtimes closer to that of the unfixed baselines. To explore this further, we ran additional simulations with even larger values of $n$. We were able to execute our heuristic constructor with $n = 2,000$ in under 10 seconds on a single laptop.

Overall, these results provide further evidence supporting the practicality of a fixed-point constraint. While the fixed point leads to higher runtimes relative to the unfixed baselines, the use of heuristic constructors serves as an effective technique to ensure that runtimes remain reasonable, even for high values of $n$.

%\begin{figure*}[h]
%    \centering
%    \includegraphics[width=\linewidth]{figures/epsilon_differences_combined.png}
%    \vspace*{0.75mm}
%    \caption{\nitin{ADD IN HERE; this was generated using the crime data}} 
%    \label{fig:split_epsilons}
%\end{figure*}

\section{Discussion}
\label{disc}
This study proposes a framework for computing tables of privatized counts that balances three performance criteria: accuracy of distribution, accuracy of counts, and runtime. We support the design of such a framework with a mathematical theory of transition matrices that obey differential privacy and fixed-point constraints. We provide advances for two framework components: the distribution privatizer and the constructor algorithm. Our experiments provide evidence that a fixed-point constraint can be practical, often yielding a favorable tradeoff among our three performance criteria.

Our study leaves open several questions. We would like to know if, under a reasonable class of loss functions, we can provide a compact description of the global optimum of $F$, analogous to the one found by Ghosh et al. for the unfixed case \cite{ghosh2009universally}. A little more broadly, we are interested in whether an algorithm for finding the global optimum exists that is more efficient than today's general-purpose solvers. Furthermore, numerous variants of differential privacy have been proposed \cite{desfontaines2020sok}, some of which may be well-suited for a polytope-style analysis. In particular, the set of  $(\epsilon,\delta)$-differentially private count mechanisms forms a polytope in $\R^{n \times n}$. We would therefore like to investigate whether an analogue of $\epsilon$-scales can be applied to this setting. 

We are also interested in whether the upper bound $n$ can be dropped from our framework, alleviating the need for truncating counts. We believe this would require radically different techniques, including infinite-dimensional analogs of the mathematical tools developed in this paper.

Another interesting direction involves preserving distributions for other types of data beyond simple counts. For example, many applications rely on the joint distribution of two count variables. Mortality rates and disease rates are typically computed as a fraction, meaning that systematic distortions in the joint distribution of incidence counts and population counts can interfere with our understanding of health trends \cite{santos2020differential, hauer2021differential, li2023impacts}. Similarly, measures of segregation are based on the joint distribution of minority counts and majority counts, and can be biased upwards or downwards by privatization \cite{asquith2022assessing}. We believe the techniques developed in this study provide a promising starting point for developing privatization algorithms that can support applications like these.

%%%==========================%%%
%%% =*= Acknowledgements =*= %%%
%%%==========================%%%

\begin{acks}
This work was  greatly improved by comments
from our anonymous reviewers, our revision editor, Joshua Blumenstock, and members of UC Berkeley's Global Opportunity Lab. This research received no specific grant from any funding agency in the public, commercial, or not-for-profit sectors.
\end{acks}

%%%======================%%%
%%% =*= Bibliography =*= %%%
%%%======================%%%

\bibliographystyle{ACM-Reference-Format}
\bibliography{main}

%%%====================%%%
%%% =*= Appendices =*= %%%
%%%====================%%%

\newpage

\appendix
\section*{Appendix}
\label{appendix}  
\section{Proof and Other Details from Section \ref{privatizing_z}}
\label{app_proofs_privatizing_z}

%%%%%%%%%%%%%%%%%%%%%%%%%%%%%%%%%%%%%
\noindent \textbf{Proof of Theorem \ref{thm:cyclic_laplace}.} The proof proceeds in 5 steps. 
\smallskip

\noindent \textbf{Step 1:} \textit{Relate the cyclic Laplace $V(\zeta)$ to a new mechanism $\hat V(\zeta)$ to facilitate analysis.}  We introduce modified mechanism $\hat V(\zeta)$, which is identical to the cyclic Laplace mechanism $V(\zeta)$, except that it only outputs the first $n-1$ dimensions of the output, $\hat V = (V_0,..., V_{n-2}) \in \R^{n-1}$. By Observation \ref{obs:cyc_lap_sums}, the sum of the cyclic Laplace outputs must be 1. It is therefore possible to post-process the output of $\hat V(\zeta)$ by computing the dropped component, $V_{n-1} = 1 - \sum_{i=0}^{n-2}V_i$, thereby recovering the original mechanism. By the postprocessing guarantee, it is therefore sufficient to prove that $\hat V$ satisfied $\epsilon$-differential privacy. This reduction allows us to use the convenient language of probability density functions in the remainder of the proof.

\smallskip

\noindent \textbf{Step 2:} \textit{Compute the conditional probability density function of $\hat V$ given $L_{n-1} = l_{n-1}$.} The modified cyclic Laplace outputs $V_0,...,V_{n-2}$ where $V_i = \zeta_i + L_i - L_{i+1}$ for all $i \in \{0,...,n-2\}$. Repeatedly substituting $ L_{j+1} - \zeta_{j} + V_{j}$ for $L_j$ yields the following $n-1$ equations.
\begin{align*}
 L_{i} &= L_{n-1} + \sum_{j = i}^{n-2} V_j - \zeta_j  \text{ for all } i \in \{0,...,n-2\} 
\end{align*}
Given $l_{n-1} \in \R$, we condition on the event $L_{n-1} = l_{n-1}$. We also define the random vector $\hat L = (L_0,...,L_{n-2})$. For any $\hat l = (\hat l_0,...,\hat l_{n-2}) \in \R^{n-1}$, let $\phi(\hat l)$ denote the joint density of $\hat L$. By independence of the $L_i$'s, $\phi(\hat l) = \prod_{i=1}^{n-2} \varphi(\hat l_i; 1/(N\epsilon))$ where $\varphi(\cdot; 1/(N\epsilon))$ is the probability distribution function of the single-variable Laplace with scale parameter $1/(N\epsilon)$ (as described in Footnote \ref{foot:laplace_dist} of Section \ref{privatizing_z}). Since the scale parameter will not change throughout this proof, we omit it and write $\varphi(\hat l_i)$ in place of $\varphi(\hat l_i ;1/(N\epsilon))$ for brevity. Also by independence of the $L_i$'s, the conditional density of $\hat L$ given $L_{n-1} = l_{n-1}$ is also $\phi$.

Then we can write, $\hat L = g_{\zeta}(\hat V; l_{n-1})$ for function $g_{\zeta}(\ \cdot \ ; l_{n-1}):\R^{n-1} \to \R^{n-1}$,
\begin{align*}
 g_{\zeta}(\hat v; l_{n-1})_i = l_{n-1} + \sum_{j = i}^{n-2} v_j - \zeta_j \text{ for all } i \in \{0,...,n-2\}
\end{align*}
Note that $g_{\zeta}(\ \cdot \ ; l_{n-1})$ is a one-to-one function that relates random vector $\hat V$ with random vector $\hat L$.  For any $\hat v = (v_0,...,v_{n-2}) \in \R^{n-1}$, let $\rho_{\zeta}(\hat v|l_{n-1})$ denote the conditional probability density of $\hat V$ at $\hat v$, given $L_{n-1} = l_{n-1}$. By the standard change of coordinates for probability densities, we have:
\begin{align*}
\rho_{\zeta}(\hat v| l_{n-1}) &= |\text{det}(J(\hat v))| \ \phi(g_{\zeta}(\hat v; l_{n-1}))
\end{align*}
where $J(\hat v)$ is the Jacobian of $g_{\zeta}(\ \cdot \ ; l_{n-1})$ evaluated at $\hat v$. It can be verified that $J(\hat v)$ is upper triangular with ones on the diagonal, so $|\text{det}(J(\hat v))| = 1$. Hence $\rho_{\zeta}(\hat v| l_{n-1}) = \phi(g_{\zeta}(\hat v; l_{n-1}))$.

\smallskip

\noindent \textbf{Step 3:} \textit{Bound the ratio of conditional probability densities of $\hat V$ on neighboring inputs.} Consider two neighboring tables of counts $d$ and $d'$, with corresponding histograms $\eta$ and $\eta'$, and with corresponding distributions of counts $\zeta$ and $\zeta'$, respectively. Without loss of generality, suppose $d'$ is formed from $d$ by adding an individual to row $i$. Then $d'_i = d_i + 1$. 

Without loss of generality, suppose $d$ has count $f$ in row $i$, $d'$ has count $f+1$ in row $i$ and $d$ and $d'$ are otherwise identical. Let $\mathbb{I}(\cdot)$ be the indicator function. Then
$$
\eta'_{f} =  \sum_{j =0}^{N-1} \mathbb{I}(d'_j = f) = \sum_{j \ne i} \mathbb{I}(d'_j = f)
$$
because $\mathbb{I}(d'_i = f) = 0$. Also, 
$$
\eta_{f} =  \sum_{j =0}^{N-1} \mathbb{I}(d_j = f) = \sum_{j \ne i} \mathbb{I}(d_j = f) + 1 = \sum_{j \ne i} \mathbb{I}(d'_j = f) + 1
$$
Hence, $\eta'_{f} =\eta_{f} - 1$ and $\zeta'_{d_i} = \zeta_{d_i} - \frac{1}{N}$. By a similar argument, $\eta'_{f+1} =\eta_{f+1} + 1$ and $\zeta'_{f+1} = \zeta_{f+1} + \frac{1}{N}$. For all other $\bar f \notin \{f,f+1\}$, $\eta_{\bar f} = \eta'_{\bar f}$, so $\zeta_{\bar f} = \zeta'_{\bar f}$.

\smallskip

\noindent Next, consider the following ratio of conditional probability density functions.
\begin{align*}
    \frac{\rho_{\zeta}(\hat v| l_{n-1})}{\rho_{\zeta'}(\hat v| l_{n-1})} = \frac{\phi(g_{\zeta}(\hat v; l_{n-1}))}{\phi(g_{\zeta'}(\hat v; l_{n-1}))} = \frac{\prod_{i=0}^{n-2}\varphi(g_{\zeta}(\hat v; l_{n-1})_i)}{\prod_{i=0}^{n-2}\varphi(g_{\zeta'}(\hat v; l_{n-1})_i)} 
\end{align*}
Plugging in the values for $g_{\zeta}(\hat v; l_{n-1})_i)$ and $g_{\zeta'}(\hat v; l_{n-1})_i)$, we have
\begin{align*}
    \frac{\rho_{\zeta}(\hat v| l_{n-1})}{\rho_{\zeta'}(\hat v| l_{n-1})} = \frac{\prod_{i=0}^{n-2}\varphi(l_{n-1} + \sum_{j = i}^{n-2} v_j - \zeta_j )}{\prod_{i=0}^{n-2}\varphi(l_{n-1} + \sum_{j = i}^{n-2} v_j - \zeta_j' )} 
\end{align*}
Notice that $\sum_{j \le i} \zeta'_j = \sum_{j \le i} \zeta_j $ whenever $i \ne f$.  Canceling equal terms, 
\begin{align*}
    \frac{\rho_{\zeta}(\hat v| l_{n-1})}{\rho_{\zeta'}(\hat v| l_{n-1})}
    &= \frac{\varphi(l_0 + \sum_{j \le f} \zeta_j - v_j )}{\varphi(l_0 + \sum_{j \le f} \zeta'_j - v_j)}
\end{align*}
Since $\sum_{j \le f} \zeta'_j = \sum_{j \le f} \zeta_j - \frac{1}{N}$, we have
\begin{align*}
    \frac{\rho_{\zeta}(\hat v| l_{n-1})}{\rho_{\zeta'}(\hat v| l_{n-1})}
    &= \frac{\varphi(l_0 + \sum_{j \le f} \zeta_j - v_j )}{\varphi(l_0 + \sum_{j \le f} \zeta_j - \frac{1}{N} - v_j)} \\
    &= \frac{\exp(N\epsilon|l_0 + \sum_{j \le f} \zeta_j - v_j|)}{\exp(N\epsilon|l_0 + \sum_{j \le f} \zeta_j - \frac{1}{N} - v_j|)} \\
    &\le \exp(N\epsilon|1/N|) \\
    &= \exp(\epsilon)
\end{align*}
where the inequality follows by the triangle inequality. By a symmetric argument, $\frac{\rho_{\zeta}(\hat v| l_{n-1})}{\rho_{\zeta'}(\hat v| l_{n-1})} \ge \exp(-\epsilon)$. 

\smallskip

\noindent \textbf{Step 4:} \textit{Bound the ratio of (unconditional) probability densities of $\hat V$ on neighboring inputs $\zeta$ and $\zeta'$.} Define $\rho_{\zeta}(\hat v)$ as the unconditional density of $\hat V$. Then,
\begin{align*}
    \frac{\rho_{\zeta}(\hat v)}{\rho_{\zeta'}(\hat v)} &= \frac{\int_{\R}\rho_{\zeta}(\hat v| l_{n-1})\varphi(l_{n-1}) dl_{n-1}}{\int_{\R}\rho_{\zeta'}(\hat v| l_{n-1})\varphi(l_{n-1}) dl_{n-1}}
\end{align*}
Upper-bounding the numerator using Step 3,
\begin{align*}
    \frac{\rho_{\zeta}(\hat v)}{\rho_{\zeta'}(\hat v)} \le \frac{\int_{\R}e^\epsilon\rho_{\zeta'}(\hat v| l_{n-1})\varphi(l_{n-1}) dl_{n-1}}{\int_{\R}\rho_{\zeta'}(\hat v| l_{n-1})\varphi(l_{n-1}) dl_{n-1}} = e^\epsilon
\end{align*}
By a symmetric argument using the other bound from Step 3, we similarly deduce $\frac{\rho_{\zeta}(\hat v)}{\rho_{\zeta'}(\hat v)} \ge e^{-\epsilon}$.

\smallskip

\noindent \textbf{Step 5:} \textit{Deduce that the original cyclic Laplace mechanism satisfies differential privacy.} By Step 4, the ratio of probability density functions is between $e^{-\epsilon}$ and $e^{\epsilon}$ at every point $\hat v \in \R^{n-1}$. Hence, by standard integration, for any measurable set $\hat E \subseteq \R^{n-1}$, we have $e^{-\epsilon} \P(\hat V(\zeta') \in \hat E) \le \P(\hat V(\zeta) \in \hat E) \le e^{\epsilon} \P(\hat V(\zeta') \in \hat E)$. By the reduction argument from Step 1, this implies that the cyclic Laplace mechanism also satisfies $\epsilon$-differential privacy. \qed

%%%%%%%%%%%%%%%%%%%%%%%%%%%%%%%%%%%%%

\medskip

\noindent \textbf{Additional Applications of the Cyclic Approach.} Our cyclic approach for the Laplace mechanism can also be used with other differentially private mechanisms, such as the Gaussian mechanism \cite{dwork2006our, balle2018improving}. 

\begin{defn}
\label{defn:cyclic_gauss} 
We define the \textit{cyclic Gaussian mechanism for distributions of counts} as a randomized algorithm that accepts a parameter $\sigma$ and table of counts $d$ as inputs. It then computes its distribution of counts $\zeta$ and independently samples $G_0,...,G_{n-2}$ from Normal$(0, \sigma^2)$, defines $G_{n} = G_0$, and outputs $V_i = \zeta_i + G_i - G_{i+1}$ for all $i \in \{0,...,n-1\}$.
\end{defn}

This mechanism satisfies approximate differential privacy \cite{dwork2006our}, for appropriately chosen values of $\sigma$. To sketch a proof, we first rewrite the cyclic Gaussian mechanism as $V = \zeta + H$, where $H_i = G_i - G_{i+1}$. We perform the technique in Step 1 of Theorem \ref{thm:cyclic_laplace}, and construct a new mechanism $\hat V = (V_0,...,V_{n-2})$. This new mechanism $\hat V$ is equivalent to the Exact Correlated Gaussian Mechanism of \citet{xiao2021optimizing}, with an $(n -1) \times (n-1)$ covariance matrix whose diagonal entries are $2\sigma^2$, whose superdiagonal and subdiagonal entries are both $-\sigma^2$, and all other entries are 0. This covariance matrix is invertible, as it is a symmetric tridiagonal Toeplitz matrix, so its eigenvalues are given by $\lambda_m = 2\sigma^2 + 2\sigma^2 \cos\left( \frac{m \pi}{n}\right)$ for $m \in \{1,...,n-1\}$ \cite{willms2008analytic}, all of which are positive. Hence $\hat V$ satisfies approximate differential privacy by Theorem 3 of \citet{xiao2021optimizing} for appropriately chosen values of $\sigma$. By the post-processing argument presented in Step 1 of Theorem \ref{thm:cyclic_laplace}, the cyclic Gaussian mechanism for distributions of counts $V$ also satisfies approximate differential privacy.

\section{Proofs and Other Details from Section \ref{characterization}}
\label{app_proofs_characterization}

\subsection{Using $\Psi$ to Represent $F$}
\label{app_details_F}

%%%%%%%%%%%%%%%%%%%%%%%%%%%%%%%%%%%%%%%%%%%%%%

\noindent \textbf{Proof of Lemma \ref{lem:U_col_scales}.} To prove Lemma \ref{lem:U_col_scales}, the following technical lemma will be helpful; it establishes the linear independence of the differential privacy constraints and the row sum constraints on a transition matrix.

\begin{lem}
\label{lem:linear_indep_constraint_matrix}
        For any $\pi_i \in \{\pi_i^-,\pi_i^+\}$ for $i \in \{0,...n-2\}$, the set of row vectors $\{\pi_0,...,\pi_{n-2}, \mathds{1}_n^\top\}$ are linearly independent.
\end{lem}

\begin{proof}  (Of Lemma \ref{lem:linear_indep_constraint_matrix}.)
    Let $W$ be the $n \times n$ matrix whose rows are $\pi_0,...,\pi_{n-2}, \mathds{1}_n^\top$ in this order. Suppose $v \in \R^n$ is in the kernel of $W$. If $v_0 > 0$, then row $0$ of $Wv = 0_n$ specifies $\pi_0 v = 0$. If $\pi_0 = \pi_0^+$, then $\pi_0 v = v_0 - e^\epsilon v_1 = 0$, yielding $v_1 = e^{-\epsilon} v_0 > 0$. By similar argument, $\pi_0 = \pi_0^-$, also yields $v_1 > 0$. Repeating the process $n-1$ times, we know $v_i > 0$ for all $i \in \{0,..,n-1\}$. The last row of $Wv = 0_n$ specifies $\mathds{1}_n^\top v = 0$. However, the left hand side is a sum of all positive entries, yielding a contradiction. By similar argument, $v_0$ cannot be negative. Therefore, $v_0 = 0$. By similar procedure as above, this implies $v_i=0$ for all $i \in \{0,...,n-1\}$. Hence the kernel of $W$ only contains the zero vector, so $W$ is full rank, implying the rows are linearly independent, as required.
\end{proof}

\begin{proof} (Of Lemma \ref{lem:U_col_scales})
     Let $Q$ denote the set of probability vectors that satisfy $\epsilon$-neighbor indistinguishability.
     $$Q = \{t \in \Delta_n: \pi_i^+ t \le 0, \pi_i^- t \le 0 \text{ for all } 0 \le i \le n-2\}$$ 
     Then $Q$ is a convex polytope, and every scale $\Psi_j \in Q$. To prove the result, we will first show that the extreme points of $Q$ are the scales, $ex(Q) = \{\Psi_0,...,\Psi_{k-1}\}$.

    A point in a convex polytope in $\R^n$ is an extreme point if and only if there are $n$ linearly independent binding constraints. The privacy constraints guarantee that any $t \in Q$ has all positive entries. So for each $i \in \{0,...,n-1\}$ only one of $\pi_i^+$ or $\pi_i^-$ can bind. Therefore, all binding constraints are linearly independent by Lemma \ref{lem:linear_indep_constraint_matrix}. This implies that a point in $Q$ is extreme if and only if there are $n$ binding constraints. In turn, this is equivalent to the statement that $t$ is an $\epsilon$-scale. This establishes $ex(Q) = \{\Psi_0,...,\Psi_{k-1}\}$.

    We now complete the proof of the lemma. Suppose $T \in U$ and consider the $j^{th}$ column. We can write $T_j = \alpha_j q$ for some $\alpha_j \ge 0$ and $q \in Q$. Since the extreme points of $Q$ are the $\epsilon$-scales, we can write $q = \Psi \beta$ for $\beta \in \Delta_n$. Hence $T_j = \alpha_j \Psi \beta$, which is a conic combination of $\epsilon$-scales. 
\end{proof}

%%%%%%%%%%%%%%%%%%%%%%%%%%%%%%%%%%%%%%%%%%%%%%

\noindent \textbf{Proof of Theorem \ref{thm:nonzero_U}.} $(\implies)$ Suppose $T \in ex(U)$. A well known fact from geometry is that a point in a convex polytope in $\R^{n\times n}$ is an extreme point if and only if there are $n^2$ linearly independent binding constraints. Since $U$ includes $n$ row constraints, $T$ must satisfy $n^2 - n = n(n-1)$ tight privacy constraints. Since each column is non-zero, only a single differential privacy constraint can bind for every pair of adjacent entries, meaning that each column is associated with a maximum of $n-1$ tight privacy constraints. It follows that every column must have exactly $n-1$ tight privacy constraints, hence every column of $T$ is a  multiple of an $\epsilon$-scale. 

    The proof that the columns of $T$ are linearly independent follows by a similar argument to Theorem 2 of Holohan et al. \cite{holohan2017extreme}, which we include below for completeness. Assume for contradiction that there exists a linear dependence among the columns of $T$. Then there is a non-zero vector of weights, $\theta \in \R^n$ such that $\sum_{j=0}^{n-1} \theta_j T_j = 0_n$.  Let $W = T \ \text{diag}(\theta)$. Then $W \ne 0_{n\times n}$ as $\theta \ne 0_n$. Also, $W\mathds{1}_n = 0_n$ by the linear dependency.

    Choose any $\delta \in (0,1/\max_j\{|\theta_j|\}).$ Define $X = T-\delta W$ and $Y = T+\delta W$. By construction,  $(T \pm \delta W)\mathds{1}_n = T\mathds{1}_n \pm \delta W\mathds{1}_n = T\mathds{1}_n = \mathds{1}_n$, so the row sums of both $X$ and $Y$ are 1. Next, we show that all elements of $X$ are positive. For $i,j \in \{0,...,n-1\}$,
    \begin{align*}
        x_{i,j} &= t_{i,j} - \delta w_{i,j} \\
        &= t_{i,j} - \delta \sum_{u=0}^{n-1}t_{i,u}\text{diag}(\theta)_{u,j} \\
        &= t_{i,j}(1-\delta \theta_j) \\
        &> t_{i,j}\bigg(1-\frac{\theta_j}{\max_v\{|\theta_v|\}}\bigg) \\ &\ge 0
    \end{align*}
    By similar argument, all elements of $Y$ are positive. Next, we note that any column $j$ of $X$ is a positive multiple of the corresponding column of $T$, as $X_j = T_j(1-\delta \theta_j)$. Since $T_j$ is $\epsilon$-neighbor indistinguishable, so is $X_j$. Therefore, $X \in U$. By a similar argument, $Y \in U$ as well. But then $T = \frac{1}{2}X + \frac{1}{2}Y$, implying that $T \notin ex(U)$, yielding a contradiction.

    $(\impliedby)$ Suppose $T\in U$ and the columns of $T$ are linearly independent multiples of $\epsilon$-scales. Assume for contradiction that $T$ is not an extreme point. Then, there exists $X,Y \in U$, distinct, and $\alpha\in (0,1)$ such that $T = \alpha X + (1-\alpha) Y$. 
    
    For each column $j$ we show that $T_j$ is a constant multiple of $X_j$. This is equivalent to showing $x_{i,j}/t_{i,j} = x_{i+1,j}/t_{i+1,j}$ for all $i \in \{0,...,n-2\}$. Assume for contradiction that there exists some $i \in \{0,...,n-2\}$ such that $x_{i,j}/t_{i,j} \ne x_{i+1,j}/t_{i+1,j}$. We proceed by cases. If $t_{i,j} = e^{\epsilon} t_{i+1,j}$, then $x_{i,j} < e^{\epsilon} x_{i+1,j}$. Hence,
    \[
    \begin{split}
        e^{\epsilon} t_{i+1,j} & = t_{i,j} \\
    & = \alpha x_{i,j} + (1-\alpha) y_{i,j} \\
    & < \alpha e^{\epsilon} x_{i+1,j} + (1-\alpha) y_{i,j} \\
    & \le \alpha e^{\epsilon} x_{i+1,j} + (1-\alpha) e^{\epsilon} y_{i+1,j} \\
    & = e^{\epsilon} t_{i+1,j}
    \end{split}
    \]
    producing a contradiction. An analogous argument holds when $t_{i,j} = e^{-\epsilon} t_{i+1,j}$, implying that $X_j$ must be proportional to $T_j$. Similarly, $Y_j$ must be proportional to $T_j$. Hence, there must be weight vectors $\theta, \nu \in \R^n$ such that $X = T \ \text{diag}(\theta)$ and $ Y = T \ \text{diag}(\nu)$. Define $\beta = \theta - \nu$ and note that $\beta \ne 0_n$ since $X \ne Y$. Then 
    \begin{align*}
        0_n & = (X-Y)\mathds{1}_n \\
        &= \left( T \ \text{diag}(\theta) - T \ \text{diag}(\nu) \right) \mathds{1}_n \\
        &= T \ \text{diag}(\theta-\nu) \mathds{1}_n \\
        &= T \ \text{diag}(\beta) \mathds{1}_n
    \end{align*}
    which is a linear dependency on the columns of $T$, yielding a contradiction. \qed

%%%%%%%%%%%%%%%%%%%%%%%%%%%%%%%%%%%%%%%%%%%%%%

\medskip

\noindent \textbf{Proof of Theorem \ref{thm:RF_to_F_surjection}.} We first show that $F \subseteq \Psi(R_F)$. Given $T \in F$, it is also in $U$, so every column of $T$ can be expressed a conic combination of $\epsilon$-scales by Lemma \ref{lem:U_col_scales}. So for every column $j$ there exists a column vector $\beta(j)$ such that $T_j = \Psi \beta(j)$. Let $B \in \R^{k \times n}_{\ge 0}$ be the matrix whose $j^{th}$ column is $\beta(j)$. Then $T = \Psi B $. To complete the proof, we show $B \in R_F$. To do so, we show $z\Psi B = z$ and $\Psi B \mathds{1}_n = \mathds{1}_n$.

    We establish the first condition $z\Psi B = z$ via contradiction. Suppose there exists some $j \in \{0,...,n-1\}$ such that $z_j \ne (z\Psi B)_j = z \Psi B_j$. By construction, $T_j = \Psi B_j$, so $z_j \ne zT_j$, so $z \ne zT$, contradicting $T \in F$. Hence $z\Psi B = z$.
    
    The second condition holds as well, as $\Psi B \mathds{1}_n = T \mathds{1}_n = \mathds{1}_n$ since $T \in F$. Hence, given $T \in F$ there exists $B \in R_F$ such that $T = \Psi B$. This establishes $T \in \Psi(R_F)$.

    Next, we show that $F \supseteq \Psi(R_F)$. Let $ T \in \Psi(R_F)$. Then there exists some $B \in R_F$ such that $T = \Psi B$. By definition of $R_F$, $z T = z \Psi B = z$, and $T \mathds{1}_n = \Psi B \mathds{1}_n =\mathds{1}_n$. Column $j$ of $T$ is given by $\Psi B_j$. This is a conic combination of $\epsilon$-neighbor indistinguishable vectors, and is hence $\epsilon$-neighbor indistinguishable. Since every column of $T$ is $\epsilon$-neighbor indistinguishable, $T$ is $\epsilon$-differentially private, establishing $T \in F$. Hence, $\Psi(R_F) \supseteq F$, producing $\Psi(R_F) = F$. \qed

%%%%%%%%%%%%%%%%%%%%%%%%%%%%%%%%%%%%%%%%%%%%%%

\medskip

\noindent \textbf{Proof of Proposition \ref{prop:extreme_surjection_R_to_F}.} We provide two proofs of this result.

\begin{proof} (Approach 1) For two matrices $X$ and $Y$ of the same size, we write the Frobenius inner product as $\langle X, Y \rangle = tr[X^\top Y]$. A standard result from linear algebra is that a point $X$ in a convex polytope $P$ is an extreme point if and only if it is the unique maximizer of some linear function from $P$ to $\R$. We restate this result in our context as follows.

\begin{lem} 
\label{lem:linear_alg_opt}
    Given convex polytope $P$ in some vector space of real matrices $V$, a matrix $T \in P$ is an extreme point if and only if there exists $d \in V$ such that $T$ is the unique maximizer of $\phi(\cdot) = \langle \cdot, d\rangle$ in $P$. 
\end{lem}

    Now, we prove Proposition \ref{prop:extreme_surjection_R_to_F}. Suppose $T \in ex(F)$. Then by Lemma \ref{lem:linear_alg_opt} there exists a $c \in \R^{n \times n}$ such that $T$ is the unique maximizer of $\langle c, T' \rangle$ for all $T' \in F$. Consider $d = \Psi^\top  c$ and note for any $B' \in R_F$, $\langle d, B'  \rangle = tr[d^\top B'] = tr[c^\top \Psi B'] = \langle c, \Psi B' \rangle$.  Let $B$ be any extreme point of $R_F$ that maximizes $\langle d, B'  \rangle = \langle c, \Psi B' \rangle$ over all $B' \in R_F$. We must have $\Psi B = T$ because if $\Psi B = T' \ne T$, then for any $\bar{B}$ in the preimage of $T$ through $\Psi$, we have $\langle d, B \rangle = \langle c, \Psi B \rangle = \langle c, T' \rangle < \langle c, T \rangle = \langle c, \Psi\bar{B} \rangle  =\langle d, \bar{B} \rangle $ where the inequality follows from the maximality of $T$. This contradicts the maximality of $B$. Hence $B$ is an extreme point in $R_F$ such that $T = \Psi B$.
\end{proof}

\begin{proof} (Approach 2)
    Suppose $T \in ex(F)$. Since $\Psi$ is surjective, there exists some $W \in R_F$ such that $T = \Psi W$. Since $R_F$ is convex, $W = \theta_1 B^{(1)} + ... + \theta_l B^{(l)}$ for some $B^{(1)},...,B^{(l)} \in ex(R_F)$ and positive $\theta_i$'s that sum to 1.  Then 
    $$T = \Psi W = \sum_{i=1}^l \theta_l (\Psi B^{(i)})$$
    Since $T$ is extreme, the expression on the right-hand side must be the trivial convex combination. Hence, $\Psi B^{(i)} = T$ for all $i \in \{1,...,l\}$. Let $B = B^{(1)}$. Then there is some $B \in ex(R_F)$ such that $T = \Psi B$. 
\end{proof}

%%%%%%%%%%%%%%%%%%%%%%%%%%%%%%%%%%%%%
\begin{ex}
\label{ex:unpreserved_extreme_RF_to_F}  \textbf{$\Psi$ is not a 1-1 mapping between $ex(R_F)$ and $ex(F)$.} Consider the target distribution $z = \left(\frac{1}{3},\frac{1}{3},\frac{1}{3}\right)$ with $\epsilon = \ln(2)$ and $n=3$. There are 4 possible $\epsilon$-scales, displayed in the matrix $\Psi$ below.
$$ \Psi = 
    \begin{bmatrix}
    4/7 & 2/5 & 1/4 & 1/7\\
    2/7 & 1/5 & 1/2 & 2/7\\
    1/7 & 2/5 & 1/4 & 4/7\\
    \end{bmatrix}
    $$
Using the Avis-Fukuda vertex enumeration algorithm \cite{avis1991pivoting}, we generate all extreme points of $F$ and $R_{F}$. We find $|ex(F)| = 36$ and $|ex(R_{F})| = 78$. 

The matrices $B^{(1)}$ and $B^{(2)}$ below are extreme points of $R_{F}$, yet they both map to the same extreme point in $F$ under $\Psi$.

$$ B^{(1)} = 
    \begin{bmatrix}
    0 & 0 & 0 \\
    1 & 0 & 2/3 \\
    0 & 1 & 1/3 \\
    0 & 0 & 0 \\
    \end{bmatrix} \in ex(R_{F})
$$   

$$ B^{(2)} = 
    \begin{bmatrix}
    0 & 0 & 14/30 \\
    1 & 0 & 0 \\
    0 & 1 & 2/30 \\
    0 & 0 & 14/30 \\
    \end{bmatrix} \in ex(R_{F})
$$

$$ \Psi B^{(1)} = \Psi B^{(2)} = 
    \begin{bmatrix}
    2/5 & 1/4 & 7/20 \\
    1/5 & 1/2 & 3/10 \\
    2/5 & 1/4 & 7/20 \\
    \end{bmatrix} \in ex(F)
$$  
Additionally, matrix $B^{(3)}$ below is an extreme point of $R_F$, yet $\Psi B^{(3)} \notin ex(F)$, as the matrix $\Psi B^{(3)}$ is not an output of the Avis-Fukuda algorithm over $F$.
$$ B^{(3)} = 
    \begin{bmatrix}
    5/6 & 0 & 1/3 \\
    0 & 0 & 0 \\
    0 & 0 & 2/3 \\
    1/6 & 1 & 0 \\
    \end{bmatrix} \in ex(R_{F})
    $$
\end{ex}

%%%%%%%%%%%%%%%%%%%%%%%%%%%%%%%%%%%%%

\medskip

\noindent \textbf{Proof of Lemma \ref{lem:psi_simple_relationship_movement_matrices}.} $(\implies)$ We proceed by the contrapositive. Suppose there exists a non-zero movement matrix $\mu \in M_F(\Psi)$ such that $\mu_{i,j}$ is zero whenever $b_{i,j}$ is zero. Since $\mu \ne 0_{k\times n}$, there must be a column with at least one non-zero entry. Note that this column cannot have exactly one non-zero entry, because if $\mu_{i,j} \ne 0$ yet $\mu_{i',j} = 0$ for all $i' \ne i$, then $(z\Psi\mu)_j = (z\Psi)\mu_j = z\Psi_i\mu_{i,j} \ne 0$, implying $z\Psi\mu \ne 0_n^\top$, contradicting $\mu \in M_F(\Psi)$. Hence, there is some column of $\mu$ with at least two non-zero entries. 
    
    For column $j$, let $l(j)$ be the smallest row index $i$ such that $b_{i,j}>0$, and let $D(j)$ be the set of all other indices $i$ such that $b_{i,j}>0$. Then there exists some $j \in \{0,...,n-1\}$ such that $D(j) \ne \emptyset$ because $\mu$ has a column with at least 2 non-zero entries and $B$ must be non-zero at those positions. We can then write, 
    $$
        \Psi \mu_j = \Psi_{l(j)}\mu_{l(j),j} + \sum_{i \in D(j)} \Psi_{i}\mu_{i,j} 
    $$
    For all $i \in D(j)$, define $h^j_i$ as the vector
    $$
    h^j_i = (z\Psi_i)^{-1}\Psi_i - (z\Psi_{l(j)})^{-1}\Psi_{l(j)}
    $$
    Note that the set of vectors $\{h_i^j \text{ : } i \in D(j)\}$ is the set $H(B_j)$. Rearranging gives, 
    $$\Psi_i = z\Psi_i\left((z\Psi_{l(j)})^{-1}\Psi_{l(j)} + h^j_i\right)$$ 
    Substituting into the above equation for $\Psi \mu_j$ yields,
    $$
        \Psi \mu_j = \Psi_{l(j)}\mu_{l(j),j} + \sum_{i \in D(j)} z\Psi_i\left((z\Psi_{l(j)})^{-1}\Psi_{l(j)} + h^j_i\right) \mu_{i,j} 
    $$
    Next, we rewrite each term on the right-hand side of the above equation. The first term can be expressed as $z\Psi_{l(j)}(z\Psi_{l(j)})^{-1}\Psi_{l(j)} \mu_{l(j),j}$ and the sum can be rewritten as 
    $$
         \sum_{i \in D(j)} z\Psi_i (z\Psi_{l(j)})^{-1}\Psi_{l(j)} \mu_{i,j} + \sum_{i \in D(j)} z\Psi_i  \mu_{i,j} h^j_i 
    $$ 
    Hence,  
    \begin{align*}
        \Psi \mu_j &= z\Psi_{l(j)}(z\Psi_{l(j)})^{-1}\Psi_{l(j)} \mu_{l(j),j} \\
        &+\sum_{i \in D(j)} z\Psi_i (z\Psi_{l(j)})^{-1}\Psi_{l(j)} \mu_{i,j} \\
        &+\sum_{i \in D(j)} z\Psi_i  \mu_{i,j} h^j_i 
    \end{align*}
    which simplifies to
         $$\Psi \mu_j = (z\Psi_{l(j)})^{-1}\Psi_{l(j)}\sum_{i \in D(j) \cup \{l(j)\}} z\Psi_i \mu_{i,j} + \sum_{i \in D(j)} z\Psi_i  \mu_{i,j} h^j_i  
        $$
    By the definition of $M(\Psi)$, $z\Psi\mu = 0_n^\top$, meaning the first sum is $0$. Therefore, we have      
        $$\Psi \mu_j= \sum_{i \in D(j)} z\Psi_i  \mu_{i,j} h^j_i  $$    
    So then
    \begin{align*}
        0_n & =  \Psi \mu \mathds{1}_n = \sum_{j=0}^{n-1} \Psi \mu_j = \sum_{j=0}^{n-1} \sum_{i \in D(j)} z\Psi_i  \mu_{i,j} h^j_i 
    \end{align*}
    The above is a linear dependency among the vectors of the multiset $\uplus_{j=0}^{n-1} H(B_j)$, implying that $B$ is not $\Psi$-affinely simplified, as required.

    $(\impliedby)$ Again by the contrapositive, suppose $B$ is not $\Psi$-affinely simplified. Then $\uplus_{j=0}^{m-1} H(B_j)$ is linearly dependent. As above, let $l(j)$ be the smallest row $i$ for which $b_{i,j} > 0$ and $D(j)$ be the set of all other indices $i$ such that $b_{i,j}>0$. The linear dependency can be written as,
    $$
    \sum_{j=0}^{n-1} \sum_{x \in D(j)} \alpha_{x,j}\left((z\Psi_i)^{-1}\Psi_x - (z\Psi_{l(j)})^{-1}\Psi_{l(j)}\right) = 0_n
    $$
    for a collection of $\alpha_{x,j}$'s, not all zero. We can rewrite the left-hand side of the above equation as
    $$
    \sum_{j=0}^{n-1} \left(-\sum_{x \in D(j)} \alpha_{x,j} (z\Psi_{l(j)})^{-1}\right)\Psi_{l(j)}+ \sum_{j=0}^{n-1} \sum_{x \in D(j)} \left(\alpha_{i,j} (z\Psi_{x})^{-1}\right)\Psi_{x}
    $$
    Let $\mu$ be the matrix such that $\mu_{i,j}$ equals $-\sum_{x \in D(j)} (z\Psi_{l(j)})^{-1}\alpha_{x,j}$ when $i = l(j)$, equals $\alpha_{i,j} (z\Psi_i)^{-1}$ whenever $i \in D(j)$, and equals 0 otherwise. Since the $\alpha_{x,j}$'s are not all zero, the matrix $\mu \ne 0_{k \times n}$. Also, the sum above can be written in matrix notation as $\Psi \mu \mathds{1}_n = 0_n$. Additionally, the $j^{th}$ column constraint of $ \mu $ can be written,
    \begin{align*}
    (z\Psi \mu)_j &= z\sum_{i=0}^{n-1} \Psi_i \mu_{i,j} \\ 
    &= z\Psi_{l(j)}\mu_{l(j),j} + \sum_{x\in D(j)} z\Psi_{x}\mu_{x,j} \\
    &=  -z\Psi_{l(j)}\sum_{x \in D(j)}(z\Psi_{l(j)})^{-1} \alpha_{x,j} + \sum_{x\in D(j)} z\Psi_{x} (z\Psi_x)^{-1}\alpha_{x,j}\\
    &= -\sum_{x \in D(j)} \alpha_{x,j} + \sum_{x \in D(j)} \alpha_{x,j} \\
    &= 0    
    \end{align*}
    implying $z\Psi\mu = 0_n^\top$. Therefore $\mu \in M_F(\Psi)$. \qed

%%%%%%%%%%%%%%%%%%%%%%%%%%%%%%%%%%%%%

\medskip

\noindent \textbf{Proof of Theorem \ref{thm:extreme_f_characterization_general}.} $(\implies)$ We proceed by the contrapositive. Suppose $B$ is not $\Psi$-simplified. By the equivalence of $\Psi$-simplified and movement matrices established in Lemma \ref{lem:psi_simple_relationship_movement_matrices}, there exists a movement matrix $\mu \in M(\Psi)$ such that $\mu_{i,j} = 0$ whenever $b_{i,j} = 0$.  Choose any positive 
    
    $$
    \delta < \min_{\{(i,j) \text{ $|$ } \mu_{i,j} \ne 0\}} \bigg\{\frac{b_{i,j}}{|\mu_{i,j}|}\bigg\}
    $$ 
    Then $B \pm \delta \mu \in R_F$, as
    \begin{enumerate}
        \item For any $i,j$ entry, $\pm \delta \mu_{i,j} \le b_{i.j}$ with equality if and only if $b_{i,j} = 0$. Hence $(B \pm \delta \mu)_{i,j} \ge 0$.       
            \item $
    \Psi (B \pm \delta \mu) \mathds{1}_n = \Psi B  \mathds{1}_n \pm \delta \Psi \mu \mathds{1}_n = \mathds{1}_n \pm 0_n = \mathds{1}_n 
    $
        \item $
    z \Psi (B \pm \delta \mu)  = z \Psi B \pm \delta z \Psi \mu = z \pm 0_n^\top = z
    $
    \end{enumerate}
    Hence, $B$ can be represented as 
    $$
    B = \frac{1}{2}(B + \delta \mu) + \frac{1}{2} (B - \delta \mu)
    $$
    Therefore $B \notin ex(R_F)$.
  
    $(\impliedby)$ Proceed by the contrapositive yet again. Suppose $B \notin ex(R_F)$. Then there exists distinct $X,Y \in R_F$ and $\theta \in (0,1)$ such that $B = \theta X + (1-\theta)Y$. By non-negativity, the convex combination above implies $x_{i,j} = 0 = y_{i,j}$ whenever $b_{i,j} = 0$. Let $\mu = X - Y$. Then $\mu_{i,j} = 0$ whenever $b_{i,j} = 0$. Also, $z\Psi \mu = z\Psi X - z\Psi Y = z - z = 0_n^\top$ and $\Psi \mu \mathds{1}_n = \Psi X \mathds{1}_n - \psi Y \mathds{1}_n  = \mathds{1}_n - \mathds{1}_n = 0_n$. So $\mu \in M_F(\Psi)$, implying $B$ is not $\Psi$-simplified. \qed

%%%%%%%%%%%%%%%%%%%%%%%%%%%%%%%%%%%%%

\medskip

\noindent \textbf{Proof of Corollary \ref{cor:nonzero_RF}.} Let $B \in ex(R_F)$ and $\mathcal{P}$ be the set of positive indices of $z$. First, observe that for each $j \in \mathcal{P}$, $B_j \ne 0_k$ (as $z \Psi B_j = z_j \ne 0$), so there must be at least one positive entry in $B_j$, so $B$ must contain at least $|\mathcal{P}|$ positive entries. To show the upper bound on the number of positive entries is $|\mathcal{P}| + n - 1$, note that every element $h \in \uplus_{j=0}^{n-1}H(B_j)$ belongs to an $n-1$ dimensional subspace of $\R^{n}$ because of the constraint $zh = 0$. Since the elements are linearly independent, there can be at most $n-1$ of them. The number of positive entries in each column $j \in \mathcal{P}$ of $B$ is at most one more than the number of difference vectors in $H(B_j)$ (because $H(B_j)$ is constructed to include one vector corresponding to each index $i$ such that $b_{i,j}>0$ except the smallest). Hence the number of positive elements of $B$ is at most $|\mathcal{P}|$ plus the number of elements of $\uplus_{j=0}^{n-1}H(B_j)$, yielding a maximum of $|\mathcal{P}| + n - 1$ positive entries. \qed

%%%%%%%%%%%%%%%%%%%%%%%%%%%%%%%%%%%%%%%%%%%%%

\subsection{Characterizing extreme points that are preserved under $\Psi$.} 
\label{app_geom_ext}

In this section, we provide a necessary and sufficient condition that allows us to identify which extreme points in $R_{F}$ are mapped to extreme points of $F$. Our theorem is more general than required, referring to the extreme points of a convex polytope $P$ in an arbitrary $\hat{k} \times \hat{n}$ space of real matrices under affine surjection. We place hats over the dimensions as the following results hold for any dimensions $\hat{k} \ge \hat{n}$, and not only the specific values of $k$ and $n$ used in this paper. As before, for matrices $X$ and $Y$ of the same size, we write the Frobenius inner product as $\langle X, Y \rangle = tr[X^\top Y]$. We will make use of the following two definitions.

\begin{defn}
    Given a convex polytope $P \subset \R^{\hat{k} \times \hat{n}}$ and a matrix $D \in \R^{\hat{k} \times \hat{n}}$ we define the maximizing set of the function $\langle \cdot,D \rangle$ as
    $\text{opt}_{P,D} = \argmax_{B \in P} \langle B,D \rangle$.
\end{defn}

\begin{defn} The kernel of a linear surjection $\Phi: \R^{\hat{k} \times \hat{n}} \to \R^{\hat{n} \times \hat{n}}$ is denoted as 
$K[\Phi] = \{B \in \R^{\hat{k} \times \hat{n}} \text{ : } \Phi B = 0_{\hat{n} \times \hat{n}} \in \R^{\hat{n} \times \hat{n}}\}$. 
We denote the orthogonal complement of the kernel of $\Phi$ as the set
$K^{\perp}[\Phi] = \{ B \in \R^{\hat{k} \times \hat{n}} \text{ : } \langle B,B' \rangle = 0 \text{ for all } B' \in K[\Phi]\}$. 
\end{defn}

\begin{thm}
\label{thm:ex_preserve_surjection}
    Let $\Phi: \R^{\hat{k} \times \hat{n}} \to \R^{\hat{n} \times \hat{n}}$ be a linear surjection between convex polytopes $P \subset \R^{\hat{k} \times \hat{n}}$ and $\tilde{P} \subset \R^{\hat{n} \times \hat{n}}$. Suppose $B \in P$. Then $\Phi B \in ex(\tilde{P})$ if and only if there  exists $D \in K^{\perp}[\Phi]$ such that $B\in \text{opt}_{P,D}$ and $\Phi$ maps all the matrices in $\text{opt}_{P,D}$ to the same point in $\tilde{P}$.
\end{thm}

\begin{proof}
    ($\implies$) Since $\Phi B$ is an extreme point of $\tilde{P}$, by Lemma \ref{lem:linear_alg_opt} there exists $W \in \R^{\hat{n} \times \hat{n}}$ such that $\langle W , \Phi B \rangle > \langle W , T \rangle$ for all $T \in \hat{P} - \{\Phi B\}$. Consider $D = \Phi^\top  W \in \R^{\hat{k} \times \hat{n}}$. Then $D \in K^{\perp}[\Phi]$ because for any $\hat{B} \in K[\Phi]$, $\langle D,\hat{B} \rangle = tr[D^\top \hat{B}]=tr[W^\top (\Phi\hat{B})]=0$. Also, for any $B' \in P$,
\begin{align*}
\langle D, B' \rangle &= tr[D^\top  B'] = tr[(\Phi^\top  W)^\top  B'] = tr[ W^\top   (\Phi B') ] = \langle W , \Phi B' \rangle  
\end{align*}
Then for any $\tilde{B} \in P$, if $\Phi \tilde{B} = \Phi B$, then 
$$
\langle D, B \rangle =  \langle W, \Phi B \rangle = \langle W, \Phi \tilde{B} \rangle = \langle D, \tilde{B} \rangle
$$
Alternatively, if $\Phi \tilde{B} \ne \Phi B$, then
$$
\langle D, B \rangle =  \langle W, \Phi B \rangle > \langle W, \Phi \tilde{B} \rangle = \langle D, \tilde{B} \rangle
$$
where the strict inequality follows from the definition of $W$. Thus $B \in \text{opt}_{P,D}$ and $\Phi$ sends all of $\text{opt}_{P,D}$ to $\Phi B$.

\noindent ($\impliedby$) Suppose there  exists $D \in K^{\perp}[\Phi]$ such that $B\in \text{opt}_{P,D}$ and $\Phi$ maps all of $\text{opt}_{P,D}$ to the same point. 
Define $W \in \R^{\hat{n} \times \hat{n}}$ as $W = (\Phi^\top )^\dagger D$ where $\dagger$ indicates the Moore–Penrose pseudoinverse. For any $B' \in P$,
\begin{align*}
    \langle W, \Phi B' \rangle = tr[(\Phi B')^\top  W] = tr[ (\Phi B')^\top  (\Phi^\top )^\dagger D] = tr[ B'^\top  \Phi^\top  (\Phi^\top )^\dagger D]
\end{align*}
A standard result from linear algebra holds that $\Phi^\top  (\Phi^\top )^\dagger$ is a projection matrix onto $K^{\perp}[\Phi]$. Since $D$ is in $K^{\perp}[\Phi]$, $\Phi^\top  (\Phi^\top )^\dagger D = D$. Hence,
$$
\langle W, \Phi B' \rangle = tr[ B'^\top  \Phi^\top  (\Phi^\top )^\dagger D] = tr[ B'^\top  D] = \langle D, B' \rangle 
$$
Given some $T \in \tilde{P}$, the surjectivity of $\Phi$ guarantees the existence of a $B' \in P$ with $T = \Phi B'$. Either $B'$ is in $\text{opt}_{P,D}$ or it is not. If $B' \in \text{opt}_{P,D}$, since $\Phi$ maps all of $\text{opt}_{P,D}$ to the same point, we must have $T = \Phi B$. If $B' \notin \text{opt}_{P,D}$, then $T \ne \Phi B$ so
$$\langle W, T \rangle = \langle D, B' \rangle < \langle D, B \rangle = \langle W,  \Phi B  \rangle$$
Therefore, there exists a $W \in \R^{\hat{n} \times \hat{n}}$ such that $\Phi B$ is the unique maximizer, implying $\Phi B \in ex(\tilde{P})$ by Lemma \ref{lem:linear_alg_opt}. 
\end{proof}

Using Theorem \ref{thm:ex_preserve_surjection}, we provide a geometric characterization of extreme points of $R_{F}$ that are preserved under $\Psi$ by setting $\Phi = \Psi$, $\hat{k} = k$, $\hat{n} = n$, $P = R_{F}$, and $\tilde{P} = F$.

\begin{cor}
\label{cor:k_perp_characterization_F}
    Suppose $B \in R_{F}$. Then $\Psi B \in ex(F)$ if and only if there  exists $D \in K^{\perp}[\Psi]$ such that $B\in \text{opt}_{R_{F},D}$ and $\Psi$ maps all the matrices in $\text{opt}_{R_{F},D}$ to the same point.
\end{cor}

Since the above result holds for any $B \in R_{F}$, it also holds in particular when $B \in ex(R_{F})$. 

\section{Heuristic Constructor Analysis of Section \ref{algorithmic}}
\label{app_proofs_algorithmic}

\subsection{Runtime Analysis of Algorithm \ref{alg:heuristic}}

The following proposition shows how Step 8 can be computed efficiently and will help us establish the runtime of Algorithm \ref{alg:heuristic}. 

\begin{prop} 
\label{prop:max_compute} 
    The value of $q$ in Step 8 of the Algorithm \ref{alg:heuristic} can be computed as $ q = \min\left\{q_0,...,q_{n-1}, \frac{c_j}{zs}  \right\}$ where,
    $$q_i = \begin{cases}
        \frac{e^\epsilon r_{i+1} -
    r_{i} }{e^\epsilon s_{i+1} - s_{i}}, &p_i=1\\
        \frac{r_{i}  - e^{-\epsilon} r_{i+1} }{s_i -e^{-\epsilon} s_{i+1}}, &p_i=-1\\
    \end{cases}$$
\end{prop}

\begin{proof}
    In Step 8 of the algorithm, for any possible value of $\gamma$, $r-\gamma s$ is $\epsilon$-neighbor indistinguishable if and only if for every $0 \le i \le n-2$ the following two constraints hold.
    \begin{align*}
        e^{-\epsilon}(r_{i+1} - \gamma s_{i+1}) &\le r_i - \gamma s_i\\
        & \text{ and } \\
        r_i - \gamma s_i  &\le e^\epsilon (r_{i+1} - \gamma s_{i+1})
    \end{align*} 
    Let $\mathbb{L}_i$ be the set of values for $\gamma$ for which the first constraint holds and let $\mathbb{U}_i$ be the set of values for $\gamma$ for which the second constraint holds.
    $$
    \mathbb{L}_i= \Big\{ \gamma :  e^{-\epsilon}(r_{i+1} - \gamma s_{i+1}) \le r_i - \gamma s_i \Big\}$$
    $$\mathbb{U}_i = \Big\{ \gamma :  r_i - \gamma s_i  \le e^\epsilon (r_{i+1} - \gamma s_{i+1}) \Big\}
    $$
    Additionally, let 
    $$
    \mathbb{V} = \{ \gamma : \gamma zs \le c_j\} = \left(-\infty, \frac{c_j}{zs}\right)
    $$
    Then the value chosen for $q$ in Step 8 can be written as the maximum of the intersection of all of these sets.
     $$
    q = \max \left\{ \mathbb{V} \cap \bigcap_{0 \le i \le n-2} \mathbb{L}_i \cap \mathbb{U}_i  \right\}$$
    The sets $\mathbb{L}_i$ and $\mathbb{U}_i$ can be rewritten as, 
$$
\mathbb{L}_i= \Big\{ \gamma :\gamma \big( s_i -e^{-\epsilon}
s_{i+1}  \big)  \le r_{i}  - e^{-\epsilon} r_{i+1}  \Big\}
$$
$$
\mathbb{U}_i = \Big\{ \gamma :\gamma \big( e^\epsilon s_{i+1} -
s_{i}  \big)  \le e^\epsilon r_{i+1} -
r_{i}  \Big\}
$$
The right hand side of each inequality above is non-negative by $\epsilon$-neighbor indistinguishability of $r$. We consider two cases separately: when $p_i = 1$ and when $p_i = -1$.

In the case when $p_i = +1$, the left hand side of the inequality of $\mathbb{L}_i$ is zero, so $\mathbb{L}_i = \R$. Meanwhile, the left hand side of $\mathbb{U}_i$ is $\gamma$ times a positive quantity, so the inequality can be written as 
$$
\gamma \le \frac{e^\epsilon r_{i+1} -
r_{i} }{e^\epsilon s_{i+1} -
s_{i}} = q_i
$$
where $q_i$ is defined as in the statement of the lemma. Therefore $\mathbb{U}_i = (-\infty, q_i)$ so $\mathbb{L}_i \cap \mathbb{U}_i = (-\infty, q_i)$.

Alternatively, in the case when $p_i = -1$, the left hand side of the inequality of $\mathbb{U}_i$ is zero, so $\mathbb{U}_i = \R$. Meanwhile, the left hand side of the inequality of $\mathbb{L}_i$ is $\gamma$ times a positive quantity, so the inequality can be written, 
$$
\gamma \le \frac{r_{i}  - e^{-\epsilon} r_{i+1} }{s_i -e^{-\epsilon}
s_{i+1}} = q_i
$$
where $q_i$ is defined as in the statement of the lemma. Therefore $\mathbb{L}_i = (-\infty, q_i)$ so $\mathbb{L}_i \cap \mathbb{U}_i = (-\infty, q_i)$. Therefore, in either case,
$$
    q = \max \left\{ \left(-\infty, \frac{c_j}{zs}\right) \cap \bigcap_{0 \le i \le n-2} (-\infty, q_i) \right\} = \min\left\{q_0,...,q_{n-1}, \frac{c_j}{zs}  \right\} 
$$
completing the proof.
\end{proof}

Observe that each $q_i$ can be computed in constant time, so the value of $q$ in Step 8 can be computed in $\mathcal{O}(n)$ time. Using this result, we can now establish the runtime of Algorithm \ref{alg:heuristic} in Theorem \ref{thm:heur_halts}. 

%%%%%%%%%%%%%%%%%%%%%%%%%%%%%%%%%%%%%

\medskip

\noindent \textbf{Proof of Theorem \ref{thm:heur_halts}.} First, we show that the inner loop executes a maximum of $2n-1$ times during the entire program execution (and not per outer-loop execution). Our proof proceeds by demonstrating a map from the program state to the non-negative integers that decrements with each inner-loop iteration.

    At any point in the execution of the algorithm, define $\Omega$ to be the number of columns $j$ with $c_j > 0$, plus the number of rows $0 \le i \le n-2$ such that $r_i \notin \{e^{-\epsilon} r_{i+1}, e^{\epsilon} r_{i+1}\}$. We show that $\Omega$ decreases with each pass through the inner loop. Since $\Omega$ can be no more than the number of columns plus the number of rows minus 1, this establishes that the algorithm's inner loop executes at most $2n-1$ times.

    Once $c_j=0$ for some $0 \le j \le n-2$, column $j$ will never be selected again by the algorithm, so the bound $c_j = 0$ will continue to hold through the remaining execution. Similarly, once there is a row $i$ with $r_{i+1} = e^{\epsilon} r_{i}$, pattern $p$ will be set in Step 6 so that $p_i = +1$. Let $r'$ be the next value of variable $r$ assigned in Step 9. Then we have,
    \begin{align*}
        r'_{i+1} &= r_{i+1} -qs_{i+1} = e^{\epsilon}r_{i} - qe^{\epsilon}s_{i} =e^{\epsilon}r'_{i}
    \end{align*}
    An analogous result follows when $r_{i+1} = e^{-\epsilon}r_{i}$, establishing that once a row $i$ meets the $\epsilon$-neighbor-indistinguishability bound, the bound continues to hold through the remaining execution.

    Next, we show that in each inner loop iteration, $\Omega$ decreases.  In one pass of Step 8, either $q = c_j$, in which case, $c_j$ is set to 0 in Step 9, decreasing $\Omega$. Alternately, if $q \ne c_j$, then $r - \gamma s$ is $\epsilon$-neighbor-indistinguishable for $\gamma = q$ but not $\epsilon$-neighbor-indistinguishable for any $\gamma > q$. There must therefore be some position $i$ such that the ratio $\frac{(r - \gamma s)_{i+1} }{(r - \gamma s)_{i}}$ is within $[e^{-\epsilon},e^{+\epsilon}]$ for $\gamma \le q$, and outside the interval for $\gamma > q$. Note that the numerator and denominator are both positive as long as $c_j > 0$, because the algorithm always updates each column so that it remains $\epsilon$-neighbor indistinguishable. By the previous argument, position $i$ cannot be one for which the $\epsilon$-neighbor indistinguishability bound already holds. Since the ratio is continuous in $\gamma$ as long as the ratio is finite, we must have $\frac{(r - \gamma s)_{i+1}}{(r - \gamma s)_{i}} \in \{ e^{-\epsilon},e^{+\epsilon}\}$ when $\gamma = q$. This establishes that the $\epsilon$-neighbor indistinguishability bound goes from not binding to binding during Step 9, decreasing $\Omega$.

    Since $\Omega$ decreases with each inner loop iteration and cannot be negative, the algorithm must halt. The number of inner loop iterations is bounded by the maximum possible value of $\Omega$, which is at most $2n-1$ when the algorithm begins. 

    Next, we show that the outer loop executes a maximum of $n$ times. To see this, note that once a column $j$ is selected in Step 3, the inner loop will continue executing until $c_j = 0$ (by the previous argument, we know the inner loop will halt). At that point, column $j$ will never be selected again, so $c_j = 0$ will continue to hold through the remaining execution. Since each of the $n$ columns can only be selected once, the outer loop executes a maximum of $n$ times.

    Applying the above analysis, Table \ref{tab:runtime} lists the total number of times each step may execute, along with a big-$\mathcal{O}$ bound for each time it executes. Note in particular the max repetitions are total for an entire program execution, and not per iteration of any loop. Totaling the values in the table, the runtime of the entire algorithm is  $\mathcal{O}\left(\max\{\tau n, n^2\}\right)$. \qed

    \begin{table}[!h]
    \begin{center}
    \begin{tabular}{|c|c|c|}
    \hline
         Step & Max Repetitions & Runtime per Repetition \\
         \hline
         1 &  $1$ & $\mathcal{O}(n^2)$ \\
         2 &  $n$ & $\mathcal{O}(1)$ \\
         3 &  $n$  & $\mathcal{O}(\tau)$ \\
         4 & $2n -1$ & $\mathcal{O}(1)$ \\
         5 & $2n-1$ & $\mathcal{O}(n)$ \\
         6 & $2n-1$ & $\mathcal{O}(n)$ \\
         7 & $2n-1$ & $\mathcal{O}(n)$ \\
         8 & $2n-1$ & $\mathcal{O}(n)$ \\
         9 & $2n-1$ & $\mathcal{O}(n)$ \\
         10 & $2n-1$ & $\mathcal{O}(1)$ \\
         11 & $n$ & $\mathcal{O}(1)$ \\
         12 & $1$ & $\mathcal{O}(1)$ \\
         \hline
    \end{tabular}
    \end{center}
    \vspace{1mm}
    \caption{Total execution time attributable to each step of Algorithm \ref{alg:heuristic}. }
    \label{tab:runtime}
    \end{table}
    
%%%%%%%%%%%%%%%%%%%%%%%%%%%%%%%%%%%%%

\subsection{Extreme Point Analysis of Algorithm \ref{alg:heuristic}}

Given inputs $z, \epsilon,$ and $\kappa$ to Algorithm \ref{alg:heuristic}, we denote the output as $O[z,\epsilon, \kappa] \in \R^{n \times n}$. For readability, we will omit the inputs and simply refer to this matrix as $O$. In Theorem \ref{thm:heur_output_extreme_F} below, we prove that $O$ is an extreme point of $F$.  

To do so, we first show in Lemma \ref{lem:first_mismatch} that the probability placed on the diagonal by Algorithm \ref{alg:heuristic} is maximal in a particular sense. Next, we consider the set of all possible vectors that could be valid choices for a column at a particular point in algorithm execution. We show in Lemma  \ref{lem:column_extremeness} that this set is a polytope and that the vector selected by the algorithm is one of its extreme points. Finally, we show that if an algorithm sequentially fills each column with a vector that is extreme in this sense, then $O$ is an extreme point of $F$, thereby finishing the proof of Theorem  \ref{thm:heur_output_extreme_F}.

In the execution of the algorithm, the function $\kappa$ will select columns in a particular sequence. For a column $j$, let $\text{prev}(j)$ denote the set of all columns chosen before $j$, and $\text{post}(j)$ denote the set of all columns chosen after $j$.

\begin{defn}
    Given column $0 \le j \le n-1$, let $F_j$ be the set of matrices $O' \in F$ such that $O'_{l} = O_{l}$ for all $l \in \text{prev}(j)$.
\end{defn}

By construction, when $j$ is the first column selected, $F_j = F$, and when $j$ is the last column selected, $F_j = \{O\}$. Moreover, for any columns $j$ and $j' \in \text{post}(j)$, $F_j \supseteq F_{j'}$.

Notice that each $F_{j}$ is formed by adding linear constraints to $F$, so it is also a convex polytope. The following lemma shows that each time the algorithm fills in a column, it puts the maximum possible probability on the diagonal, given all the previously filled-in columns.

\begin{lem}
    \label{lem:first_mismatch}
    For any $0 \le j \le n-1$ and for any $O' \in F_j$ such that $O'_j \ne O_j$, we have $o_{j, j} > o'_{j, j}$.
    \end{lem}

\begin{proof}
    Assume for contradiction that $o_{j, j} \le o'_{j, j}$. Since $O'_{j} \ne O_{j}$, there must be some row $0 \le i \le n-1$ such that $o'_{i,j} < o_{i, j}$ (otherwise, we would have $o'_{i,j} \ge o_{i, j}$ for all rows $i$, with strict inequality in at least one row; this implies $z O'_{j} > z O_{j} = z_{j}$, so $O' \notin F$, contradicting $O' \in F_j \subseteq F$). 

    \medskip
    
    \noindent We consider three possible cases for $i$.

    \medskip
    
    \noindent \underline{Case 1}: Consider $i = j$. Then $o_{j, j} > o'_{j, j}$, yielding an immediate contradiction.

    \medskip

    \noindent \underline{Case 2}: Consider $i < j$. Since the condition $o'_{i', j} < o_{i', j}$ holds for $i' = i$ but not for $i' = j$, there must be some $i \le \hat i < j$ such that the condition holds for $i' = \hat i$ but not for $i' = \hat i + 1$. That is, 
    
    $$
    o'_{\hat i,j} < o_{\hat i, j} \text{ and } o'_{\hat i+1 ,j} \ge o_{\hat i+1, j}
    $$
    There are two subcases to consider.

   \begin{adjustwidth}{2.5em}{0pt}
    \underline{Subcase 1}: Suppose $o_{\hat i, j} = e^{-\epsilon} o_{\hat i+1, j}$. Then $o'_{\hat i,j} < o_{\hat i, j} = e^{-\epsilon} o_{\hat i+1, j} \le e^{-\epsilon} o'_{\hat i+1 ,j}$, contradicting the $\epsilon$-neighbor indistinguishability of $O'_{j}$.
    \end{adjustwidth}

   \medskip

    \begin{adjustwidth}{2.5em}{0pt}
    \underline{Subcase 2}: Suppose $o_{\hat i, j} > e^{-\epsilon} o_{\hat i+1, j}$. Then the scale variable $s$ that is added to column $j$ during Step 9, cannot always have $s_{\hat i} = e^{-\epsilon} s_{\hat i+1} $. This means that at some point during the outer loop iteration for column $j$, the value of pattern variable $p$ at the end of Step 6 had $p_{\hat i} = -1$, which can only happen if the variable $r$ fulfills $r_{\hat i} = e^{\epsilon} r_{\hat i+1} $. The proof of Theorem \ref{thm:heur_halts} establishes that once the remainder fulfills this bound, it always fulfills the bound. 
    
    Let $\hat r$ be the value of variable $r$ at the start of the outer loop iteration for column $ j$. Then, the value at the end of the outer loop iteration would be given by $\hat r - O_{j}$, and the bound above can be written as
    $$\hat r_{\hat i} - o_{\hat i, j} = e^{\epsilon} (\hat r_{\hat i+1} - o_{\hat i+1, j}) $$

    \noindent Since $O$ is an output of our algorithm, we also have

    $$\hat r = \mathds{1}_n - \sum_{\bar j \in \text{prev}(j)} O_{\bar j}$$

    \noindent Furthermore, since the rows of $O'$ must sum to 1, we can decompose $O'$ as follows.
    \begin{align*}
        \mathds{1}_n &= \sum_{\bar j \in \text{prev}(j)} O'_{\bar j} + O'_{j} + \sum_{\tilde j \in \text{post}(j)} O'_{\tilde j} \\
        &= \sum_{\bar j \in \text{prev}(j)} O_{\bar j} + O'_{j} + \sum_{\tilde j \in \text{post}(j)} O'_{\tilde j} 
    \end{align*}

    \noindent where the second equality follows from $O' \in F_j$. Combining the prior two equalities yields
    $$ 
    \hat r -  O'_{j} = \sum_{\tilde j \in \text{post}(j)} O'_{\tilde j} 
    $$
    \noindent The right-hand side is a sum of $\epsilon-$neighbor indistinguishable columns, so must be $\epsilon-$neighbor indistinguishable. However, the left-hand side is not $\epsilon$-neighbor indistinguishable because,
    \begin{align*}
        (\hat r -  O'_{j})_{\hat i} &= \hat r_{\hat i} - o'_{\hat i , j} > \hat r_{\hat i} - o_{\hat i , j} \\
        &= e^{\epsilon} (\hat r_{\hat i+1} - o_{\hat i+1, j}) \\
        &\ge e^{\epsilon} (\hat r_{\hat i+1} - o'_{\hat i+1, j})  \\
        &= e^\epsilon  (\hat r -  O'_{j})_{\hat i+1}
    \end{align*}
    This yields a contradiction. 
    \end{adjustwidth}

    \noindent Since all subcases yield a contradiction, we deduce a contradiction for the entire case.

    \medskip

    \noindent \underline{Case 3}: Consider $i > j$. This case is also yields a contradiction. It is analogous to the prior case, and is hence omitted. 

    \medskip

    \noindent Since each case yields a contradiction, we conclude $o_{j, j} > o'_{j, j}$.
\end{proof} 

Next, we consider the set of all possible vectors that could be inserted into column $j$ when the algorithm reaches the column. This is defined formally below.

\begin{defn}
    Given column $0 \le j \le n-1$, let $\Gamma_j$ be the set of entries for column $j$ in $F_j$; $\Gamma_j = \{ O'_j : O' \in F_j \}$.
\end{defn}

Notice that $\Gamma_j$ is a linear projection of $F_j$ so it is also a convex polytope. We show the column $O_j$ chosen by the algorithm is an extreme point of $\Gamma_j$.

\begin{lem}
    \label{lem:column_extremeness}
    Given column $0 \le j \le n-1$, $O_j \in ex(\Gamma_j)$.
\end{lem}

\begin{proof}
    Consider any $v \in \Gamma_j$, where $v \ne O_j$. Then $v = O'_j$ for some $O' \in F_j$, and since $v \ne O_j$, $O'_j \ne O_j$, By the Lemma \ref{lem:first_mismatch}, we have $o_{j,j} > v_j$. Let $u \in \R^n$ be the column vector with $u_j = 1$, $u_{j'} = 0 $ for all $j' \ne j$. Then $u^\top O_j = o_{j,j} > v_j = u^\top v$. Then $O_j$ is the unique maximizer of the linear objective $u$ in $\Gamma_j$, so it is an extreme point.
\end{proof}

Finally, the last step is to show that if an algorithm always fills each column $j$ with a vector that is extreme in the sense above, the overall matrix will also be an extreme point of $F$.

\medskip

\noindent \textbf{Proof of Theorem  \ref{thm:heur_output_extreme_F}.}  First, we show $O \in F$. First, note that local variables $A,r,$ and $c$ are only updated in Step 9. This step ensures that the conditions $A \mathds{1}_n + r = \mathds{1}_n$ and $z A + c = z$ are maintained. Since the algorithm continues until $r = 0_n$ and $c = 0^\top_n$, the final value of $A$ fulfills the conditions $A\mathds{1}_n = \mathds{1}_n$ and $zA = z$. Additionally, the update in Step 9 proceeds by adding a positive multiple of an $\epsilon$-scale to a column of $A$. At every point of execution, every column of $A$ is a conic combination of $\epsilon$-scales, and is thus $\epsilon$-neighbor indistinguishable. Since $O$ is the final value of $A$, we have $O \in F$.
    
Next, we show $O \in ex(F)$. Assume for contradiction that there exists $0 < \alpha < 1$ and distinct matrices $X,Y \in F$ such that $\alpha X + (1-\alpha) Y = O$. Then there must be at least one column in which $X$ and $Y$ differ. Let $j$ be the first column chosen by the algorithm for which $X_j \ne Y_j$. Then $X, Y \in F_j$ so $X_j, Y_j \in \Gamma_j$. Moreover $\alpha X_j + (1-\alpha) Y_j = O_j$, contradicting $O_j \in ex(\Gamma_j)$. \qed

%%%%%%%%%%%%%%%%%%%%%%%%%%%%%

\section{Characterization of Unfixed Polytope $U$}
\label{app_characterizing_U}

In this Appendix section, we turn our attention to the general problem of describing extreme points of $U$ (including extreme points that include zero columns). Similar to $ex(F)$, it is possible to characterize $ex(U)$ using a framework that is based on $\epsilon$-scales. We will see that extreme points may be built from conic combinations of scales that are linearly independent of each other in a particular sense (Definition \ref{defn:psi_lin_simple}). To make this explicit, we will define a new polytope $R_U$ which encodes how much of each scale contributes to every column of $T \in U$. 

\subsection{Representing $U$ with $\epsilon$-Scales}

We now define the unfixed representation polytope $R_U$. Analogous to $R_F$, for a representation matrix $B \in R_U$, $b_{i,j}$ quantifies the amount of scale $i$ that enters column $j$ of $T$.
% below, which will form an algebraically tractable transformation of $U$. 

\begin{defn}
    We define the \textit{unfixed representation polytope} $R_U$ as
$$
R_U = \big\{B \in \R^{k \times n}_{\ge 0} \text{ : }  \Psi B \mathds{1}_n = \mathds{1}_n\big\}
$$
\end{defn}

The following theorem will show that every matrix in $U$ can be reached from a representation matrix in $R_U$ under this linear map.

\begin{thm}
\label{thm:RU_to_U_surjection}
    $\Psi $ is an affine surjection from $R_U$ to $U$.
\end{thm}

\begin{proof}
    Since $\Psi$ is a linear map, it is affine, and hence it only remains to show $\Psi(R_U) = U$. 
    
     Suppose $T \in \Psi(R_U)$. Then there is some $B \in R_U$ such that $\Psi B = T$. By definition of $R_U$, $T\mathds{1}_n = \Psi B \mathds{1}_n = \mathds{1}_n$. Column $j$ of $T$ is given by $\Psi B_j$. This is a conic combination of $\epsilon$-neighbor indistinguishable vectors, and is hence $\epsilon$-neighbor indistinguishable. Since every column of $T$ is $\epsilon$-neighbor indistinguishable, $T$ is $\epsilon$-differentially private, establishing $T \in U$.
     
    Conversely, suppose $T \in U$. By Lemma \ref{lem:U_col_scales}, every column $j$ of $T$ can be written as a conic combination of $\epsilon$-scales $T_j = \Psi \beta^j$ for some $\beta^j \in \R^k_{\ge 0}$. Construct the $n \times k$ matrix $B$ where the $j^{th}$ column of $B$ is $\beta^j$. Then $T = \Psi B$.  Furthermore, $\Psi B \mathds{1}_n = T \mathds{1}_n = \mathds{1}_n$, so $B \in R_U$. Therefore, $T \in \Psi(R_U)$. Hence, $\Psi(R_U) \supseteq U$, producing $\Psi(R_U) = U$. 
\end{proof}

In Proposition \ref{prop:extreme_surjection_R_to_U}, we show the linear map $\Psi$ also provides a relation between the extreme points of $U$ and the extreme points of $R_U$. 

\begin{prop}
\label{prop:extreme_surjection_R_to_U}
    For all $T \in ex(U)$ there exists $B \in ex(R_U)$ such that $T = \Psi B$.
\end{prop}

The proof of Proposition \ref{prop:extreme_surjection_R_to_U} is very similar to that of Proposition \ref{prop:extreme_surjection_R_to_F}, as is thus omitted. Hence, every extreme point in $U$ can be written as an extreme point in $R_U$ under the mapping $\Psi$. 

However, as we discovered when $\Psi$ represented matrices in $F$ using $R_F$, this representation  is not unique for matrices in $U$. As demonstrated in Example \ref{ex:unpreserved_extreme_RU_to_U} below, it is possible for multiple extreme points of $R_U$ to map to the same extreme point in $U$, and it is possible for an extreme point of $R_U$ to map to a non-extreme point of $U$ under $\Psi$.

\begin{ex}
\label{ex:unpreserved_extreme_RU_to_U} \textbf{$\Psi$ is not a 1-1 mapping between $ex(R_U)$ and $ex(U)$.} Similar to Example \ref{ex:unpreserved_extreme_RF_to_F}, consider the case when $\epsilon = \ln(2)$ and $n=3$. There are 4 possible $\epsilon$-scales, presented as
$$ \Psi = 
    \begin{bmatrix}
    4/7 & 2/5 & 1/4 & 1/7\\
    2/7 & 1/5 & 1/2 & 2/7\\
    1/7 & 2/5 & 1/4 & 4/7\\
    \end{bmatrix}
$$
Using the Avis-Fukuda vertex enumeration algorithm, we generate all extreme points of $U$ and $R_U$. We find $|ex(U)| = 27$ and $|ex(R_U)| = 36$. In particular, both matrices $B^{(4)}$ and $B^{(5)}$ are extreme points of $R_U$, yet they both map to the same extreme point in $U$ under $\Psi$.

$$ B^{(4)} = 
    \begin{bmatrix}
    7/6 & 0 & 0 \\
    0 & 0 & 0 \\
    2/3 & 0 & 0 \\
    7/6 & 0 & 0 \\
    \end{bmatrix} \in ex(R_U)
$$   

$$ B^{(5)} = 
    \begin{bmatrix}
    0 & 0 & 0 \\
    5/3 & 0 & 0 \\
    4/3 & 0 & 0 \\
    0 & 0 & 0 \\
    \end{bmatrix} \in ex(R_U)
$$

$$ \Psi B^{(4)} = \Psi B^{(5)}  = 
    \begin{bmatrix}
    1 & 0 & 0 \\
    1 & 0 & 0 \\
    1 & 0 & 0 \\
    \end{bmatrix} \in ex(U)
$$  
Additionally, the mapping $\Psi$ may send extreme points in $R_U$ to non-extreme points in $U$. For example, 
$$ B^{(6)} = 
    \begin{bmatrix}
    7/6 & 0 & 0 \\
    0 & 0 & 0 \\
    0 & 0 & 2/3 \\
    7/6 & 0 & 0 \\
    \end{bmatrix} \in ex(R_U)
$$  
yet $\Psi B^{(6)} \notin ex(U)$, as
\begin{align*}
\Psi B^{(6)} &= 
\begin{bmatrix}
    5/6 & 0 & 1/6 \\
    4/6 & 0 & 2/6 \\
    5/6 & 0 & 1/6 \\
\end{bmatrix} \\
& = 
\frac{1}{2} \begin{bmatrix}
    11/12 & 0 & 1/12 \\
    10/12 & 0 & 2/12 \\
    11/12 & 0 & 1/12 \\
\end{bmatrix} +
\frac{1}{2} \begin{bmatrix}
    9/12 & 0 & 3/12 \\
    6/12 & 0 & 6/12 \\
    9/12 & 0 & 3/12 \\
\end{bmatrix}
\end{align*}
\end{ex}

However, we can apply Theorem \ref{thm:ex_preserve_surjection} to characterize the necessary and sufficient conditions for which an extreme point of $R_U$ is also an extreme point of $U$ under $\Psi$ (as we did in Corollary \ref{cor:k_perp_characterization_F}).

\begin{cor}
\label{cor:k_perp_characterization_U}
    Suppose $B \in R_U$. Then $\Psi B \in ex(U)$ if and only if there  exists $D \in K^{\perp}[\Psi]$ such that $B\in \text{opt}_{R_U,D}$ and $\Psi$ maps all the matrices in $\text{opt}_{R_U,D}$ to the same point.
\end{cor}
Since the above result holds for any $B \in R_U$, it holds in particular when $B \in ex(R_U)$. 

\subsection{Characterizing the Extreme Points of $R_U$}
\label{app_subsec_extreme_RU}

Next, we establish necessary and sufficient conditions for points in $R_U$ to be extreme, as as we did for $R_F$ in Theorem \ref{thm:extreme_f_characterization_general}. As we will see, an analogous theorem does hold, but the removal of a fixed-point constraint alters the affinely simplified condition. Whereas the condition of affinely simplified involved linear independence of  weighted differences of scales, we will see that the unfixed case involves linear independence of scales directly.

Given a column vector $x \in \R^k$, let $G(x)$ be the set of scales with positive coefficients in the corresponding positions of $x$, $G(x) = \{ \Psi_u : x_u>0 \}$. 

\begin{defn}
    \label{defn:psi_lin_simple} We call a $k \times n$ matrix $B$ \textit{$\Psi-$linearly simplified} if the additive union $\uplus_{j=0}^{n-1} G(B_j)$ is linearly independent. 
\end{defn}

Let $M_U(\Psi)$ be the set of $k \times n$ matrices $\mu$ such that $\Psi \mu \mathds{1}_n = 0_n$. By construction of $M_U(\Psi)$, if $X,Y \in R_U$, then $X-Y \in M_U(\Psi)$. This suggests that $M_U$ represents ``movement matrices'' that can potentially be added to matrices in $R_U$ to move through the affine space. The following lemma and theorem formalize the equivalence between the extreme points of $R_U$, the notion of $\Psi$-linearly simplified matrices, and these movement matrices.

\begin{lem}
\label{lem:psi_simple_unconstrained}
    Suppose $B \in R_U$. Then $B$ is $\Psi-$linearly simplified if and only if there does not exist a non-zero matrix $\mu \in M_U(\Psi)$ such that $\mu_{i,j}$ is zero whenever $b_{i,j}$ is zero.
\end{lem}

\begin{proof}
    $(\implies)$ We proceed by the contrapositive. Suppose there exists a non-zero movement matrix $\mu \in M_U(\psi)$ such that $\mu_{i,j}$ is zero whenever $b_{i,j}$ is zero. 

    For column $j$, let $P(j)$ be the set of row indices of $B_j$ that are non-zero. Then $G(B_j) = \{\Psi_p: p \in P(j)\}$. We can then write, 
    \begin{align*}
        \Psi \mu_j &= \sum_{p \in P(j)} \Psi_p \mu_{p,j}
    \end{align*} 
    \noindent So then
    \begin{align*}
        0_n & = \Psi \mu \mathds{1}_n = \sum_{j=0}^{n-1}  \Psi \mu_j = \sum_{j=0}^{n-1}  \sum_{p \in P(j)} \Psi_p \mu_{p,j}
    \end{align*}    
    Since $\mu \ne 0_{k \times n}$ , and the nonzero entries in $\mu_j$ can only occur in rows $P(j)$, some of the coefficients in the sum above must be nonzero. Thus, the above is a linear dependency among the vectors of the multiset $\uplus_{j=0}^{n-1} G(B_j)$, implying that $B$ is not $\Psi-$linearly simplified, as required.

    $(\impliedby)$ Again by the contrapositive, suppose $B$ is not $\Psi$-linearly simplified. Then $\uplus_{j=0}^{m-1} G(B_j)$ is linearly dependent. The linear dependency can be written as,
    $$
    \sum_{i,j} \mu_{i,j} \Psi_i  = 0_n
    $$
    for a collection of $\mu_{i,j}$'s, not all zero, where $\mu_{i,j} \ne 0$ only when $b_{i,j} > 0$, and equals 0 otherwise. Let $\mu$ be the matrix with value $\mu_{i,j}$ in position $(i,j)$. Then the sum above can be written in matrix notation as $\Psi \mu \mathds{1}_n = 0_n$. Therefore $\mu \in M_U(\Psi)$.
\end{proof}

\begin{thm}
    \label{thm:u_general_extreme} Suppose $B \in R_U$. Then $B \in ex(R_U)$ if and only if $B$ is $\Psi-$linearly simplified.
\end{thm}

\begin{proof}
    $(\implies)$ We proceed by the contrapositive. Suppose $B$ is not $\Psi$-linearly simplified. By the equivalence of $\Psi$-linearly simplified and movement matrices established in Lemma \ref{lem:psi_simple_unconstrained}, there exists a movement matrix $\mu \in M_U(\Psi)$ such that $\mu_{i,j} = 0$ whenever $b_{i,j} = 0$.  Choose any positive 
    
    $$
    \delta < \min_{\{(i,j) \text{ $|$ } \mu_{i,j} \ne 0\}} \bigg\{\frac{b_{i,j}}{|\mu_{i,j}|}\bigg\}
    $$ 
    Then $B \pm \delta \mu \in R_U$, as
    \begin{enumerate}
        \item For any $i,j$ entry, $\pm \delta \mu_{i,j} \le b_{i.j}$ with equality if and only if $b_{i,j} = 0$. Hence $(B \pm \delta \mu)_{i,j} \ge 0$.       
        \item $
    \Psi (B \pm \delta \mu) \mathds{1}_n = \Psi B  \mathds{1}_n \pm \delta \Psi \mu \mathds{1}_n = \mathds{1}_n \pm 0_n = \mathds{1}_n 
    $
    \end{enumerate}
    But then 
    $$
    B = \frac{1}{2}(B + \delta \mu) + \frac{1}{2} (B - \delta \mu)
    $$
    Therefore $B \notin ex(R_U)$.
    
    $(\impliedby)$ Proceed by the contrapositive yet again. Suppose $B \notin ex(R_U)$. Then there exists distinct $X,Y \in R_U$ and $\theta \in (0,1)$ such that $B = \theta X + (1-\theta)Y$. By non-negativity, the convex combination above implies $x_{i,j} = 0 = y_{i,j}$ whenever $b_{i,j} = 0$. Let $\mu = X - Y$. Then $\mu_{i,j} = 0$ whenever $b_{i,j} = 0$. Also $\Psi \mu \mathds{1}_n = \Psi X \mathds{1}_n - \Psi Y \mathds{1}_n  = \mathds{1}_n - \mathds{1}_n = 0_n$. So $\mu \in M_U(\Psi)$, implying $B$ is not $\Psi$-linearly simplified.
\end{proof}

Speaking intuitively, Theorem \ref{thm:u_general_extreme} says that an extreme point of $U$ can be formed by summing multiples of linearly independent $\epsilon$-scales, each one contributing to a unique column. It is useful to make two points of comparison. Since at most $n$ $\epsilon$-scales can be linearly independent, if we add the condition that all $n$ columns are non-zero, by the pigeonhole principle each column must have a multiple of a single $\epsilon$-scale, so we immediately recover Theorem \ref{thm:nonzero_U}.

In a related study, Ghosh et al. focus on finding an optimal point in $U$ under a natural class of linear objective functions \cite{ghosh2009universally} that we call row-wise concentrating (Definition \ref{defn:row_wise_concentrating}).  In Appendix \ref{app_ghosh_scale}, we explain that the result of Ghosh et al. can be seen as a further restriction on $B$ in Theorem \ref{thm:u_general_extreme}.

One combinatorial comparison can be made between the extreme points of $R_F$ and the extreme points of $R_U$. In the case of an extreme point of $U$, the condition of $\Psi$-linearly simplified implies there must be at most $n$ positive matrix entries. Additionally, any row of $B \in ex(R_U)$ cannot have more than 1 positive entry -- otherwise, the same same scale would appear twice in the multiset $\uplus_{j=0}^{m-1} G(B_j)$, contradicting $\Psi$-linear-simplified. This analysis is collected in the following corollary.

\begin{cor}
\label{cor:nonzero_RU}
    Extreme matrices $B \in ex(R_U)$ contain at least $1$ and at most $n$ positive entries. Also, such matrices $B$ have at most 1 positive entry in every row. 
\end{cor}

%As noted in Corollary \ref{cor:nonzero_RF}, the number of non-zero entries in an extreme point of $B$ is between $|\mathcal{P}|$ and $|\mathcal{P}|+n-1$ (inclusive). Consequently, when every entry of $z$ is positive, the extreme points of $R_F$ have at least as many positive entries as extreme points of $R_U$. %% Intuitively, this arises as additional scales may be required to satisfy $F$'s fixed-point constraint.
\section{Two-Stage Unfixed Optimum Constructor Analysis of Section \ref{algorithmic}}
\label{app_ghosh_scale}

In this section, we establish the key properties of the two-staged unfixed optimum constructor. In order to provide a fair comparison for fixed-point constructors, our design approach includes making reasonable assumptions where possible to enable efficient execution.

Our algorithm leverages a key result from Ghosh et al. \cite{ghosh2009universally}, who study count mechanisms that optimize a natural class of count error measures which we call \textit{row-wise concentrating} (Definition \ref{defn:row_wise_concentrating}).

Ghosh et al. show that any count mechanism that optimizes a row-wise concentrating objective is equivalent to the truncated geometric mechanism with post-processing. We will restate their result in terms of our framework in this section.

Recall from the main text that $\Sigma$ denotes the matrix containing all single-peaked scales in order. That is, $\Sigma_l$ is the single-peaked scale at position $l$. Also recall from Appendix \ref{app_subsec_extreme_RU} that given a column vector $x \in \R^k$, $G(x)$ is defined to be the set of scales with positive coefficients in the corresponding positions of $x$; that is, $G(x) = \{ \Psi_u : x_u>0 \}$. 

\begin{defn}
    Let $\mathcal{T}$ be the set of count mechanisms $T$ such that $T = \Psi B$ for $B \in R_U$ such that $\uplus_{j=0}^{n-1} G(B_j)$  is the set of all single-peaked scales, $\{\Sigma_l\}_{l=0}^{n-1}$.
\end{defn}

Let $\omega_l$ be the entry of $B$ that indicates the amount of scale $\Sigma_l$. 

\begin{obs}
    \label{obs_ghosh_weights}
    The row constraint of $R_U$ gives 
    $$\Psi B \mathds{1}_n = \sum_{l=0}^{n-1} \Psi \omega_l \Sigma_l = \mathds{1}_n$$ 
    It is straightforward to check that the single-peaked scales are linearly independent, and therefore there is a unique set of values for the $\omega_l$, which can be computed as follows.
$$\omega_l = \begin{cases}
    \frac{1-e^{-\epsilon}}{1+e^{\epsilon}}\sum_{i=0}^{n-1} e^{-\epsilon |i-l|}  , &0<l<n-1\\
    \frac{1-e^{-\epsilon}}{(1+e^{\epsilon})(1-e^\epsilon)} \sum_{i=0}^{n-1} e^{-\epsilon |i-l|} , &l \in \{0,n-1\}
\end{cases}
$$
\end{obs}

Note that the $\omega_l$ do not depend on the specific $T \in \mathcal{T}$. Given this notation, a central result from Ghosh et al. can be stated as follows.

\begin{thm} (Adapted from Ghosh et al. \cite{ghosh2009universally})
    \label{thm:ghosh_which_scales} 
    If $T$ is an optimal count mechanism for row-wise concentrating weight matrix $W$, then $T \in \mathcal{T}$.
\end{thm}

Given $T \in \mathcal{T}$, suppose $B, B' \in ex(R_U)$ with $\Psi B = \Psi B' = T$ and $\uplus_{j=0}^{n-1} G(B_j)= \uplus_{j=0}^{n-1} G(B'_j)$ is the set of all single-peaked scales. For each column $j$, $\Psi (B_j - B'_j) = T_j - T_j = 0_n$. The left hand side is a linear combination of the single-peaked scales, but these are linearly independent. Hence $B_j = B'_j$. We can therefore define $\text{col}_T(l)$ to be the column of $B$ with a positive coefficient associated with $\Sigma_l$.

Using this notation, the count error associated with weight matrix $W$ can be written,
\begin{align*}
    \langle W,  T \rangle &= \langle W,  \Psi B \rangle = \sum_{l=0}^{n-1} \omega_l (W
_{col_T(l)})^\top  \Sigma_l
\end{align*}
We next observe that the scales that are part of the optimum are ordered according to where their peaks are. This is made precise in the following lemma.

\begin{lem}
\label{lem:unfixed_left_scales_convexity}
    Given row-wise concentrating weight matrix $W$, and $0 \le l < l' \le n-1$, and $c$ is is the largest value of $\bar c \in \{0,...,n-1\}$ that minimizes $(W_{\bar c})^\top \Sigma_l$, and $c'$ is the largest value of $\bar c$ that minimizes $(W_{\bar c})^\top \Sigma_l'$, then $c' \ge c$.
\end{lem}

\begin{proof}
    We show that the claim holds when $l'=l+1$, as repeating the argument shows that the claim holds for any $l'>l$.

    Suppose $c$ is the largest value of $\bar c$ that minimizes $ (W_{\bar c})^\top \Sigma_l$. If $c=0$, then the claim follows trivially. If $c>0$, because it is a value of $\bar c$ that minimizes $(W_{\bar c})^\top \Sigma_l$, setting $\bar c$ to $c-1$ cannot decrease this quantity, so
    $$
    (W_{c-1})^\top \Sigma_l \ge  (W_{c})^\top \Sigma_l
    $$
Equivalently, $(W_{c-1}- W_c)^\top \Sigma_l \ge 0$. Let $\Delta w_i = w_{i,c-1} - w_{i,c}$. We can rewrite this inequality as
\begin{equation}
 \label{eq:inequality}
\sum_{i=1}^l  (\Delta w_i) \sigma_{i,l} + \sum_{i=l+1}^{n-1}  (\Delta w_i) \sigma_{i,l} \ge 0 \tag{$\star$}
\end{equation}
To prove the claim, we consider two cases: $c \le l$ or $c>l$. 

First consider $c > l$. Then the $\Delta w_i$ in the first sum are all non-positive (because $|c-i| = c-i  > c-1 - i = |(c-1) - i|$ and row-wise concentrating implies $w_{i,c} \ge w_{i,c-1}$), so the first sum is non-positive. Thus, the second sum must be non-negative for the inequality to hold.

Note that for $i \le l$, $\sigma_{i,l+1} = e^{i\epsilon} \sigma_{0,l+1} = \alpha e^{i\epsilon} \sigma_{0,l} = \alpha \sigma_{i, l}$, where $\alpha = \sigma_{0,l+1}/\sigma_{0,l} < 1$. Similarly, for $i \ge l+1$,  $\sigma_{i,l+1} =\beta \sigma_{i, l}$, where $\beta = \sigma_{n-1,l+1}/\sigma_{n-1,l} > 1$.
Then 
\begin{align*}
     (W_{c-1}- W_c)^\top \Sigma_{l+1} &= 
     \sum_{i=1}^{n-1} (\Delta w_{i}) \sigma_{i,l+1} \\
    &= \alpha   \sum_{i=1}^{l} \Delta w_{i} \sigma_{i,l} + \beta \sum_{i=l+1}^{n-1} \Delta w_{i} \sigma_{i,l} 
\end{align*}
Since $\alpha < 1$, and the first sum is non-positive, we have $\alpha   \sum_{i=1}^{l} \Delta w_{i} \sigma_{i,l} \ge \sum_{i=1}^{l} \Delta w_{i} \sigma_{i,l}$. Since $\beta > 1$, and the second sum is non-negative, we have $ \beta \sum_{i=l+1}^{n-1} \Delta w_{i} \sigma_{i,l} \ge \sum_{i=l+1}^{n-1} \Delta w_{i} \sigma_{i,l}$. Thus,
\begin{align*}
     (W_{c-1}- W_c)^\top \Sigma_{l+1} &\ge 
      \sum_{i=1}^{l} \Delta w_{i} \sigma_{i,l} + \sum_{i=l+1}^{n-1} \Delta w_{i} \sigma_{i,l} \\
      &= (W_{c+1}- W_c)^\top \Sigma_{l} \ge 0
\end{align*}
Thus, moving $\Sigma_{l+1}$ from column $c$ to column $c-1$ cannot decrease error; hence if $c'$ is the rightmost value of $\bar c$ that minimizes $ (W_{\bar c})^\top \Sigma_l'$, we must have $c'\ge c$. 

Next, consider the case when $c \le l$. Once again, consider the inequality in (\ref{eq:inequality}). In the second sum the $\Delta w_i$ are all non-negative (because $|c-i| = i-c  < i - (c-1) = |(c-1) - i|$ and row-wise concentrating implies $w_{i,c} \le w_{i,c-1}$).

If the first sum is also non-negative, then both $\alpha   \sum_{i=1}^{l} \Delta w_{i} \sigma_{i,l}$ and $ \beta \sum_{i=l+1}^{n-1} \Delta w_{i} \sigma_{i,l}$ are non-negative so their sum is non-negative.

On the other hand, if the first sum is negative, then because $\alpha < 1$, $\alpha   \sum_{i=1}^{l} \Delta w_{i} \sigma_{i,l} >  \sum_{i=1}^{l} \Delta w_{i} \sigma_{i,l}$. Because $\beta>1$, $\beta \sum_{i=l+1}^{n-1} \Delta w_{i} \sigma_{i,l} \ge  \sum_{i=l+1}^{n-1} \Delta w_{i} \sigma_{i,l}$. Thus $\alpha   \sum_{i=1}^{l} \Delta w_{i} \sigma_{i,l} + \beta \sum_{i=l+1}^{n-1} \Delta w_{i} \sigma_{i,l}  >   \sum_{i=1}^{l} \Delta w_{i} \sigma_{i,l} + \sum_{i=l+1}^{n-1} \Delta w_{i} \sigma_{i,l}  \ge 0$.  

In either case, we have 
$$(W_{c-1}- W_c)^\top \Sigma_{l+1} = \alpha \sum_{i=1}^{l} \Delta w_{i} \sigma_{i,l} + \beta \sum_{i=l+1}^{n-1} \Delta w_{i} \sigma_{i,l}  \ge  0$$
So once again, moving $\Sigma_{l+1}$ from column $c$ to column $c-1$ cannot decrease error, implying $c' \ge c$.
\end{proof}

As noted in the main text, we further restrict attention to weight matrices that are also \textit{row-wise convex}. The following shows that for such a weight matrix, the error corresponding to a single-peaked scale is convex in the index of the column in which it is placed.

\begin{lem}
\label{lem:weight_matrix_scale_convex}
    Let $W$ be a row-wise convex weight matrix. For any scale $s$, the inner product $(W_j)^\top s$ is convex in $j$.
\end{lem}

\begin{proof} 
    Note that $(W_{j+1})^\top s - (W_{j})^\top s = \sum_{i=0}^{n-1} (w_{i,j+1} - w_{i,j}) s_i $
    is the sum of terms that are non-decreasing in $j$, so the total is non-decreasing in $j$. Hence $(W_j)^\top s$ is convex in $j$.
\end{proof}

Our two-staged unfixed optimum constructor (Algorithm \ref{alg:ghosh_scale}) leverages the results of Lemma \ref{lem:unfixed_left_scales_convexity} and Theorem \ref{thm:ghosh_which_scales} above. It begins by considering the single-peaked scale at position $0$. To find the appropriate column for this scale, it begins at column $0$ and scans to the right. Lemma \ref{lem:weight_matrix_scale_convex} tells us that the associated cost is convex in the column, so once we find a local optimum, we know that it is a global optimum. In particular, this is true for the rightmost local optimum. The algorithm then moves on to the next single-peaked scale. By Lemma \ref{lem:unfixed_left_scales_convexity}, we know that the rightmost optimum position for this scale cannot be to the left of the previous one, so the algorithm does not need to return to column 0, but can continue scanning from the current column. The process repeats until all single-peaked scales are assigned to columns.

We are now ready to establish the runtime of Algorithm \ref{alg:ghosh_scale}.

\medskip

\noindent \textbf{Proof of Theorem \ref{thm:unfixed_runtime}.} For each iteration of the while loop, either $l$ increases by 1 or $j$ increases by 1. Both variables start at 0, and cannot be exceed $n-1$ at the start of a loop iteration -- in particular, once $l$ is assigned the value $n$, the algorithm will exit the while loop; if $j$ ever has the value of $n-1$ at the start of the loop, then $next\_error$ will be assigned the value of $\infty$, guaranteeing that the algorithm does not enter the else clause in which $j$ would be incremented. This guarantees that the while loop can execute no more than $2n$ times. 
    
Within the loop, there are two inner products that require $\mathcal{O}(n)$ to complete, and all other instructions are assumed atomic. Thus, the loop requires $\mathcal{O}(n^2)$ steps. Initializing the matrices and $\{\omega_l\}_{l=0}^{n-1}$ in Steps 1 and 2 also requires $\mathcal{O}(n^2)$. Hence, the total runtime is $\mathcal{O}(n^2)$. \qed

Lastly, we prove that Algorithm \ref{alg:ghosh_scale} outputs a count mechanism that minimizes count error.

\medskip

\noindent \textbf{Proof of Theorem \ref{thm:unfixed_optimal}.} First we show that $O \in \mathcal{T}$. To do so, we consider the theoretical $B$ associated with each step of our algorithm. In Step 1, initialize $k \times n$ matrix $B$ of all zeroes. In Step 11, set $B_{\hat i, j} \leftarrow B_{\hat i, j}  + \omega_l$ where $\hat i$ is the column of $\Psi$ equal to $\Sigma_l$. Then the algorithm maintains $A = \Psi B$ at every point in execution. Furthermore, $B$ is formed by placing each single-peaked scale into a unique column, so the multiset $\uplus_{l=0}^{n-1} G(B_j) = \{\Sigma_l\}_{l=0}^{n-1}$ . Thus $O \in \mathcal{T}$. 

Next we show that $O$ is optimal for $W$. Since $O \in \mathcal{T}$, we can write,
\begin{align*}
    \langle W,  O \rangle &=  \sum_{l=0}^{n-1} \omega_l (W
_{\text{col}_O(l)})^\top  \Sigma_l
\end{align*}
By Theorem \ref{thm:ghosh_which_scales}, any optimal $T$ is in $\mathcal{T}$, and so $\langle W,  T \rangle$ can be similarly written as  $\sum_{l=0}^{n-1} \omega_l (W_{\text{col}_T(l)})^\top  \Sigma_l$. Thus, it suffices to show that $(W_{\text{col}_O(l)})^\top  \Sigma_l \le (W_{\text{col}_T(l)})^\top  \Sigma_l$ for all $l$; that is, $\text{col}_O(l)$ is a value of $\bar c$ that minimizes $(W_{\bar c})^\top  \Sigma_l$.

Examining the while loop, variable $l$ records which single-peaked scale is being placed, and variable $j$ represents a potential column to place the current single-peaked scale. $l$ is initialized to zero and does not change until after scale $\Sigma_0$ is placed into a column in Step 11. While $l$ is zero, $j$ begins at $0$ and increments 1 at a time, scanning columns from left to right. At Step 4 of any loop iteration, $current\_error$ is assigned the value $(W_{j})^\top  \Sigma_0$. At Step 6, $next\_error$ is assigned the value $(W_{j+1})^\top  \Sigma_0$ (or $\infty$ if $j$ has reached the last column). The condition $next\_error > current\_error$ indicates that incrementing $j$ any further would increase $(W_{j})^\top  \Sigma_0$ (or that we have reached the last possible column). This implies that at this point in execution, $j$ is the right-most value of $\bar c$ that minimizes $(W_{\bar c})^\top \Sigma_0$. 

The argument for single-peaked scales at position $l>0$ is the same, except that the scan begins with $j$ equal to the column in which scale $\Sigma_{l-1}$ was placed. The scan will still locate the largest value of $\bar c$ that minimizes $(W_{\bar c})^\top \Sigma_l$ because Lemma \ref{lem:unfixed_left_scales_convexity} guarantees that this value cannot be less than that corresponding to the previous scale. \qed

\section{Development of Rule of Thumb for Privacy Budget Allocation in Section \ref{experiments}}
\label{app_additional_figures}

\begin{table*}[ht]
\centering
\begin{tabular}{|c|c|c|c|c|c|c|}
\hline
$\epsilon_t$ & Uniform & Left-skewed & Right-skewed & Bimodal symmetric & Zero-inflated & Top-inflated \\ \hline
0.10 & 0.8 & 0.4 & 0.5 & 0.7 & 0.3 & 0.4 \\ \hline
0.15 & 0.8 & 0.3 & 0.4 & 0.8 & 0.3 & 0.3 \\ \hline
0.22 & 0.4 & 0.3 & 0.3 & 0.5 & 0.3 & 0.3 \\ \hline
0.32 & 0.3 & 0.3 & 0.2 & 0.4 & 0.2 & 0.3 \\ \hline
0.48 & 0.3 & 0.2 & 0.2 & 0.3 & 0.2 & 0.3 \\ \hline
0.71 & 0.2 & 0.2 & 0.2 & 0.3 & 0.2 & 0.2 \\ \hline
1.05 & 0.2 & 0.1 & 0.2 & 0.1 & 0.1 & 0.2 \\ \hline
1.55 & 0.1 & 0.1 & 0.1 & 0.1 & 0.1 & 0.1 \\ \hline
2.29 & 0.1 & 0.1 & 0.1 & 0.1 & 0.1 & 0.1 \\ \hline
3.38 & 0.1 & 0.1 & 0.1 & 0.1 & 0.1 & 0.1 \\ \hline
5.00  & 0.1 & 0.1 & 0.1 & 0.1 & 0.1 & 0.1 \\ \hline
\end{tabular}
\vspace{1mm}
\caption{Optimal values of $f$ for each of our synthetic datasets.}
\label{tab:optimal_f_experiment}
\end{table*}

\begin{figure*}[h]
    \centering
    \includegraphics[width=\linewidth]{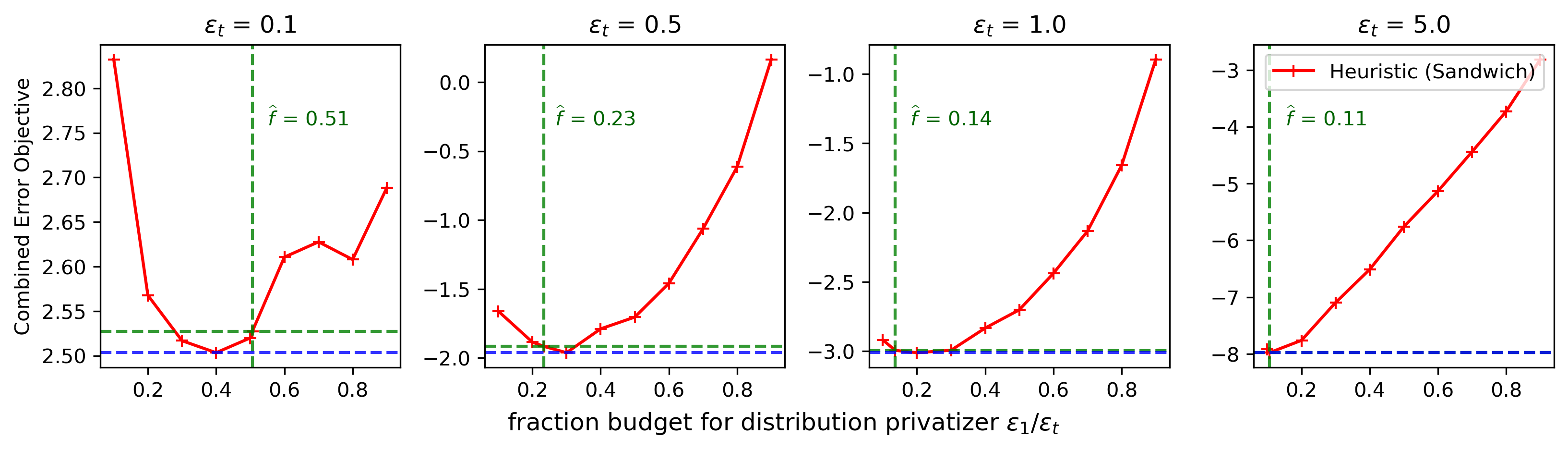}
    \vspace*{0.75mm}
    \caption{Combined error objective as a function of $f$ for the crime dataset. The vertical green line displays the value of $\hat f$ selected by our rule of thumb. The horizontal green line depicts that value of the objective resulting from the rule of thumb. Finally, the horizontal blue line depicts the minimal objective value measured across all parameter settings.} 
    \label{fig:epsilon_split}
\end{figure*}

In the two-stage fixed-point framework, the total privacy budget $\epsilon_t$ must be divided between $\epsilon_1$ for the distribution privatizer, and $\epsilon_2$ for the privatization of the table of counts. We let $f = \epsilon_1/\epsilon_t$. As noted in Section \ref{experiments}, to set $f$, we consider a combined error objective that combines accuracy of counts and accuracy of distribution.
$$
\ln(\text{count\_error}(f)) + \ln(\text{distribution\_error}(f))
$$
Let $f_\star$ be a value of $f$ that minimizes the above objective. As noted in the main text, deviating from $f_\star$ cannot decrease either type of error without also causing an equal or great increase in the other type of error (in a multiplicative sense). To see why, consider an $f'$ that decreases count error by a factor of $a$, and increases distribution error by a factor of $b$.  By optimality, 
\begin{align*}
    &\ln(\text{count\_error}(f_\star)) + \ln(\text{distribution\_error}(f_\star)) \\
    & \le \ln(\text{count\_error}(f')) + \ln(\text{distribution\_error}(f'))
\end{align*}
The right-hand side is 
\begin{align*}
    & \ln(\text{$a^{-1} \cdot$ count\_error}(f_\star)) + \ln(\text{$b \cdot$ distribution\_error}(f_\star)) \\
    &= \ln(\text{count\_error}(f_\star)) + \ln(\text{distribution\_error}(f_\star)) - a + b
\end{align*}
Canceling terms, we deduce $b \ge a$, as claimed.

One might hope to use $f_\star$ to split the total privacy budget in practice, but it turns out that $f_\star$ is data-dependent. To approximate $f_\star$, one could consider privately testing different values of $f$ on the real data -- however, this would require considerable privacy budget. While it may be possible to improve existing algorithms to select $f_\star$ with less privacy cost (e.g., adapting methods from differentially private budget selection \cite{hod2025differentially}), we defer this possibility to future research. For the current work, we propose an alternate approach: a \textit{rule of thumb} that practitioners can use to set $f$ that does not depend on the data.

At a high level, our approach to attain a rule of thumb proceeds as follows. We begin by constructing six new synthetic datasets, chosen to have a wide range of distributional features. We then perform simulations for each dataset to determine the optimal $f$ for a range of total privacy budgets. We fit a parametric model to this data, resulting in a function that provides a recommended $f$ as a function of $\epsilon_t$. Finally, we assess the performance of our rule of thumb by applying it to the three datasets used in the main text. 

Our rule of thumb is fitted using six new synthetic datasets. One reason we created new data for this exercise is to capture a wide range of distributional features that may be observed in real-world applications. A second reason is that we would like our three canonical datasets to serve as a ``holdout'' to assess the performance of our rule of thumb. This mitigates the overfitting that would arise from training a model and evaluating it on the same data. The datasets we selected are as follows.
\begin{itemize}
    \item \textbf{Uniform:} 3,000 draws from a discrete uniform distribution over the set $\{0,1,...,29\}$.
    \item \textbf{Left-skewed:} 10,000 draws from the discrete distribution over the set $\{0,1,...,69\}$ with mass function $p$ such that $p(i)= \frac{139}{140}p(i+1)$ for $0 \le i \le 68$.
    \item \textbf{Right-skewed:} 10,000 draws from the discrete distribution over the set $\{0,1,...,69\}$ with mass function $p$ such that $p(i+1)= \frac{139}{140}p(i)$ for $0 \le i \le 68$.
    \item \textbf{Bimodal symmetric:} 10,000 draws from a binomial distribution over the set $\{0,1,...,39\}$ with parameter $0.7$ and 10,000 draws from a binomial distribution over the same set with parameter $0.3$.
    \item \textbf{Zero-inflated:} 9,800 draws from a discrete uniform distribution over the set $\{0,1,...,79\}$ and 200 additional datapoints set to 0.
    \item \textbf{Top-inflated:} 9,900 draws from a discrete uniform distribution over the set $\{0,1,...,79\}$ and 100 additional datapoints set to 79.
\end{itemize}

For each synthetic dataset above, we perform a set of simulations to assess how $f$ affects the combined error objective. We operationalize count error using expected absolute deviation, and operationalize distribution error using Wasserstein distance. For the distribution privatizer we use the cyclic Laplace mechanism, and for the constructor algorithm we use Algorithm \ref{alg:heuristic} with the sandwich selector. The total privacy budget $\epsilon_t$ is varied from $0.1$ to $5.0$, taking on 11 values that are evenly spaced on a logarithmic scale. We test every choice for $f$ in the set $\left\{\frac{1}{10}, \frac{2}{10}, ..., \frac{9}{10}\right\}$. For every combination of the above parameters, we perform 100 simulations of our two-stage framework, and record the value of $f$ with the lowest combined error.

The results of our simulation study are presented in Table \ref{tab:optimal_f_experiment}. We observe an inverse relationship between $\epsilon_t$ and the optimal $f$ -- generally, as $\epsilon_t$ decreases, the optimal $f$ increases (or remains constant). As a robustness check, we perform additional simulations using MSE for count error, as well as total deviation and KS distance for distribution error. A similar inverse relationship between $\epsilon_t$ and the optimal $f$ is observed in all cases.

Based on the use case for our rule of thumb, we select a functional form with a small number of parameters, which outputs a recommended $f$  between 0 and 1. The form we select is
$$
\hat f(\epsilon_t) = B + (A-B) \exp(-K \epsilon_t)
$$
Here, $A \in [0,1]$ is the limit as $\epsilon_t \to 0^+$, $B \in [0,1]$ is the limit as $\epsilon_t \to \infty$, and $K$ is a parameter that controls the curvature. We fit these parameters to the data in Table \ref{tab:optimal_f_experiment}, yielding values $A = 0.639$, $B = 0.106$, and $K = 2.87$. This results in our rule of thumb, 
$$
\hat f(\epsilon_t)  = 0.106 + 0.533\exp{(-2.87 \epsilon_t)}
$$
We evaluate the performance of this rule of thumb for our three canonical datasets. For each dataset, we compute the combined error objective using the rule of thumb $\hat f$ and also the standard set of values $f \in \left\{\frac{1}{10}, \frac{2}{10}, ..., \frac{9}{10}\right\}$. This can be seen for the crime dataset in Figure \ref{fig:epsilon_split}. The corresponding plots for the binomial and schools datasets are similar and we omit them.  In Figure \ref{fig:epsilon_split}, the vertical green line depicts the value of $f$ chosen by the rule of thumb.  The value of the corresponding combined error objective is depicted by the horizontal green line. The horizontal blue line displays the best value of the objective across all values of $f$ we simulated. Thus, the distance between the two horizontal lines can serve as an approximate measure of the optimality gap induced by our rule of thumb. Since the combined error objective is on a natural log scale, as long as the optimality gap is small, it can be roughly interpreted as a percent change in accuracy.

As observed for the crime data in Figure \ref{fig:epsilon_split}, the optimality gap tends to decrease as we increase $\epsilon_t$.  Scanning the panels from left to right, as $\epsilon_t$ increases from $0.1$ to $0.5$, $1.0$, and $5.0$, the corresponding optimality gap changes from $0.024$ to $0.046$, $0.017$, and $0$, respectively. 
Note that a measured optimality gap of 0 is possible anytime the objective at $\hat f$ matches the minimum objective over all values of $f$ we test. Similar decreasing patterns are also present for the binomial and schools datasets. The worst optimality gap we observe across all experiments is for the binomial data at $\epsilon_t = 0.1$, with a value of $0.42$.

We use this rule of thumb throughout the experiments in Section \ref{experiments}. As we observe in that section, despite the optimality gaps that we observe in this appendix, our fixed-point techniques using the rule of thumb often provide a favorable tradeoff among our three main performance criteria: Large increases in accuracy of distribution can be attained with modest performance losses in both accuracy of counts and runtime.

%%%%%%%%%%%%%%%%%%%%%%%%%%%%%%%%%%

\end{document}